\documentclass[draft,12pt,reqno,a4paper]{amsart}

\usepackage[margin=3cm,footskip=1cm]{geometry}

\usepackage{amssymb} 
\usepackage{amsmath} 
\usepackage{amsthm} 
\usepackage{mathtools}

\usepackage{enumerate} 
\usepackage[latin1]{inputenc}
\usepackage[shortlabels]{enumitem}
\usepackage{color}
\usepackage{graphicx} 
\usepackage{verbatim}

\DeclareMathOperator*{\wslim}{w^\star-lim}

\DeclareMathOperator*{\slim}{s-lim}
\DeclareMathOperator*{\swslim}{s-w^\star-lim}

\DeclarePairedDelimiter\set{\{}{\}}

\DeclarePairedDelimiter\abs\lvert\rvert
\DeclarePairedDelimiter\norm\lVert\rVert

\newcommand{\N}{{\mathbb{N}}} 
\newcommand{\R}{{\mathbb{R}}} 
\newcommand{\C}{{\mathbb{C}}}

 \renewcommand{\c}{{\rm c}}
\newcommand{\e}{{\rm e}} \newcommand{\ess}{{\rm ess}}
 \renewcommand{\i}{{\rm i}}
\renewcommand{\d}{{\rm d}}

\newcommand{\pupo}{{\rm pp}}

\renewcommand{\Re}{{\rm Re}\,} \renewcommand{\Im}{{\rm Im}\,}

\DeclarePairedDelimiter\inp\langle\rangle

\newcommand\parb[2][]{#1 \big ( #2#1\big )}

 \renewcommand{\exp}{{\rm exp}}
\newcommand{\mand}{\text{ and }}

\newcommand{\bD}{{\mathbf D}}

\newcommand{\bX}{{\mathbf X}}

\newcommand{\vA}{{\mathcal A}} \newcommand{\vB}{{\mathcal B}}
\newcommand{\vC}{{\mathcal C}} \newcommand{\vD}{{\mathcal D}}

 \newcommand{\vH}{{\mathcal H}}
 \newcommand{\vL}{{\mathcal L}}

\newcommand{\vT}{{\mathcal T}} \newcommand{\vU}{{\mathcal U}}

\newcommand\Step[1]{
 \par\bigskip
 \noindent
 \textit{#1}.\enspace
}

\theoremstyle{plain}
\newtheorem{thm}{Theorem}[section]

\newtheorem{lemma}[thm]{Lemma} \newtheorem{corollary}[thm]{Corollary}
\theoremstyle{definition}

\newtheorem{cond}[thm]{Condition}

\newtheorem{remark}[thm]{Remark}
\newtheorem{remarks}[thm]{Remarks}
 
\newtheorem*{remarks*}{Remarks}
\newtheorem*{remark*}{Remark}

\numberwithin{equation}{section}

\title {New methods in spectral theory of $N$-body Schr\"odinger operators}

\thanks{
K.I. is supported by JSPS KAKENHI grant nr.\ 17K05325.\, 
E.S. is  supported by the Research Institute for Mathematical Sciences, a Joint
Usage/Research Center located in Kyoto University, and by DFF grant nr.\ 4181-00042.
} 

\author{T. Adachi}
\address[T. Adachi]{Graduate School of Human and Environmental Studies,
 Kyoto University, Kyoto, Japan}
\email{adachi@math.h.kyoto-u.ac.jp}

\author{K. Itakura}
\address[K. Itakura]{Department of Mathematics, Kobe University\\
1-1 Rokkodai, Nada, Kobe, 657-8501, Japan}
\email{itakura@math.kobe-u.ac.jp}

\author{K. Ito}
\address[K. Ito]{Graduate School of Mathematical Sciences, University of Tokyo\\
3-8-1 Komaba, Meguro-ku, Tokyo 153-8914, Japan}
\email{ito@ms.u-tokyo.ac.jp}

\author{E. Skibsted} \address[E. Skibsted]{Institut for Matematiske Fag \\
 Aarhus Universitet\\ Ny Munkegade 8000 Aarhus C, Denmark}
\email{skibsted@math.au.dk}

\date{\today}
\begin{document}
\begin{abstract} 
We develop a new  scheme of proofs for spectral theory 
of the $N$-body Schr\"odinger operators, 
reproducing and extending a series of sharp results under minimum conditions. 
Our main results include Rellich's theorem,  limiting absorption principle bounds, 
 microlocal resolvent bounds, 
H\"older continuity of the resolvent and a  microlocal Sommerfeld uniqueness result. 
We present a new proof of  Rellich's theorem which is unified with exponential decay estimates 
 studied  previously  only for  $L^2$-eigenfunctions.
 Each  pair-potential  is a sum of a long-range term
with first
 order derivatives, 
a short-range term  without derivatives and a singular term  of
operator- or form-bounded type, and the  setup  includes  
hard-core interaction. 
Our  proofs consist of a systematic use of commutators with 
`zeroth order' operators. 
In particular  they   do not rely on 
Mourre's differential inequality technique.
\end{abstract}

\allowdisplaybreaks

\maketitle
\tableofcontents

\section{Introduction}\label{sec:introduction}

In this paper we introduce new  methods 
establishing a series of sharp results 
in spectral theory of  $N$-body Schr\"{o}dinger operators. 
Using  elementary   functional calculus and mostly rather standard commutator arguments,
 we obtain Rellich's theorem,  LAP (\emph{Limiting Absorption Principle}) bounds, 
 microlocal resolvent bounds, 
H\"older continuity of the resolvent and a  microlocal Sommerfeld uniqueness result.
These are fundamental ingredients of  the stationary scattering theory,
which however 
  is only poorly developed so far for $N\geq 2$, and
particularly for $N\geq 3$. Moreover 
 our    results have   interest of their own. 
Previously  sharp, or Besov space, versions of the results 
were obtained only by sophisticated Mourre technology.
In this paper we reformulate and refine the Mourre estimate in terms
of a 
certain `zeroth order' operator $B$,
and we prove  `sharp results'  under natural and minimum
assumptions. 
In fact each  pair-potential  is a sum of a long-range term
with first
 order derivatives, 
a short-range term  without derivatives and a singular term  of
operator- or form-bounded type.
Hard-core interaction  is  also included (without additional complication). 

We provide a unified treatment of Rellich's theorem \cite{IS2,Is}
and exponential decay estimates of $L^2$-eigenfunctions \cite{FH}
(thereby solving a problem stated  in \cite{IS2}). 
 A  sharp version of Rellich's theorem for $N$-body operators
was established only recently in \cite{IS2}. 
However we extend it to an even stronger and more classical form. 
In addition, our proof of the sharp LAP bounds for $N$-body operators, 
which was first obtained by Jensen and Perry \cite{JP},
does not rely on Mourre's differential inequality technique. This is
in contrast to all the existing proofs we are aware of.
Instead, an integral part of our proof of the LAP bounds consists of Rellich's theorem,
being to some extent similar to the proof by Agmon and H\"ormander \cite{AH,Ho} in the $1$-body case. 
The precise setting will be presented in Section~\ref{subsec:Usual generalized $N$-body systems},
and the main results in Section~\ref{subsec:result}.

In the spectral theory of Schr\"{o}dinger operators 
it has been an issue with a long history to achieve `minimum conditions' on the pair-potentials. 
To our knowledge Lavine was the first who proved LAP in the $1$-body case 
for a natural class of potentials by a commutator method \cite{La1, La2}. 
After the discovery of the Mourre method \cite{Mo1,Mo2} the question about minimum conditions, 
in particular the historically painful issue of inclusion 
of arbitrary short-range potentials, was raised again. 
This was particularly the case in the $N$-body setting  
where Mourre's differential inequality technique, 
  which involves a certain double commutator (arising from   commutation
  with  a 
certain `first order' operator $A$) not obviously  compatible with
  short-range potentials, was the only available method for a couple  of decades. 
See \cite{ABG1,ABG2, BGM1,BGM2, Ta} for studies of this problem. 
A similar study of the Mourre method for form-bounded potentials 
was done in \cite{BMP} in the $1$-body setting.

Inventing new techniques, 
we not only reproduce known LAP results for operator-bounded potentials,
but also obtain new ones for form-bounded potentials; a
brief comparison with the literature is given  in
Section~\ref{subsec:result}. This  being said   we still have a 
`double commutator problem', 
although in a disguised form in commutation with
 the operator $B$. Most of our  proofs use
in a systematic way  commutation with   $B$ (the same operator was
used in \cite{GIS}, however  not in a systematic way). It is a
standing open problem to show the basic  results
of  the stationary scattering theory for  $N$-body Schr\"{o}dinger
operators without some rudiment  of  such a 
commutator problem (this  would require  very different  methods). A
similar problem does not appear in  time-dependent  $N$-body
scattering theory, see
\cite {De,Gr}.

For rather different methods we refer to \cite{Ge,GJ}, 
which are also based on an elementary functional calculus
with the Mourre estimate used as an input, 
although in less structured abstract contexts. 
The primary goal of these works is to demonstrate an alternative approach to 
Mourre's differential inequality technique in their abstract settings.
However, these papers do not contain the sharp LAP bounds, and more
smoothness  is required (due to multiple   commutators). 
We would also like to mention that it was realized by Melrose and Vasy
 that the Mourre estimate 
combined with propagation of singularities in a certain calculus 
leads to  LAP for a class of $N$-body Schr\"{o}dinger operators, see \cite{Va}. 
Our goals and techniques are different, again.
We aim at showing sharp results  by  elementary methods and at  a  minimum cost.

\subsection{Setting}\label{subsec:Usual generalized $N$-body systems}

In this subsection we precisely formulate the setting of the paper.
We work on a \emph{generalized $N$-body model with hard-cores}.
This is a natural generalization of the usual $N$-body model,
and hopefully   the terminologies used below would not need any
motivation for the reader.

\subsubsection{$N$-body Hamiltonian}
Let $\bX$ be                    
a finite dimensional real inner product space,
equipped with a
finite family $\{\bX_a\}_{a\in \vA}$ of subspaces closed under intersection:
For any $a,b\in\mathcal A$ there exists $c\in\mathcal A$ such that 
\begin{align}
\bX_a\cap\bX_b=\bX_c.
\label{171028}
\end{align}
The elements of $\vA$ are called {\it cluster decompositions},
and we order and write them as $a\subset b$ if
$\bX_a\supset \bX_b$. 
It is assumed that there exist
$a_{\min},a_{\max}\in \vA$ such that 
$$\bX_{a_{\min}}=\bX,\quad
\bX_{a_{\max}}=\{0\}.$$ 
For a chain of cluster decompositions $a_1\subsetneq \cdots \subsetneq a_k$
the number $k$ is called the \textit{length} of the chain,
and such a chain is said to \textit{connect} $a=a_1$ and $b=a_k$. 
For any $a\in\vA$ we denote 
the maximal length of all the chains connecting $a$ and $a_{\max}$
by $\# a$: 
$$\# a=\max\{k\,|\, a=a_1\subsetneq \dots\subsetneq a_k=a_{\max}\};\quad 
\# a_{\max}=1.$$
We say that the family $\{\bX_a\}_{a\in\vA}$ is of {\it$N$-body type} 
if $\# a_{\min}=N+1$. 
To avoid confusion we remark that $(N+1)$ number of moving particles 
form an $N$-body system after separation of the center of mass.

Let $\bX^a\subset\bX$ be the orthogonal complement of $\bX_a\subset \bX$,
and denote the associated orthogonal decomposition of $x\in\bX$ by 
$$x=x^a\oplus x_a\in\bX^a\oplus \bX_a.$$ 
The component $x_a$ is called the \emph{inter-cluster coordinates},
and $x^a$ the \emph{internal coordinates}. 
We note that the family $\{\mathbf X^a\}_{a\in \mathcal A}$ is closed under addition:
For any $a,b\in\mathcal A$ there exists $c\in\mathcal A$ such that 
$$\bX^a+\bX^b=\bX^c,$$
cf.\ \eqref{171028}. 
A real-valued measurable function $V\colon\bX\to\mathbb R$ is called 
a \textit{potential of $N$-body type} 
if there exist real-valued measurable functions
$V_a\colon\bX^a\to\mathbb R$ (i.e.\ \textit{pair-potentials}) such that 
\begin{align}
V(x)=\sum_{a\in\mathcal A}V_a(x^a)\ \ \text{for }x\in\mathbf X.
\label{eq:170927}
\end{align}
Throughout the paper we may assume $V_{a_{\min}}=0$ without loss of generality. 
We sometimes call $V$ a \textit{soft potential} compared to the following \textit{hard-cores}:
For each $a\in \vA$ let $\Omega_a\subset \bX^a$
be a given non-empty open subset with $\bX^a\setminus \Omega_a$ being compact,
and we set 
\begin{align}
\Omega=\bigcap_{a\in\mathcal A}(\Omega_a+\bX_a).
\label{eq:171002}
\end{align}
Note that by the non-emptiness assumption it automatically follows that $\Omega_{a_{\min}}=X^{a_{\min}}=\{0\}$.
The complement $\bX\setminus \Omega$ corresponds to 
a region where particles can not penetrate due to the existence of `hard-cores'.

Now we present 
conditions of the paper on these interactions. 
Throughout the paper we 
use the standard notation $\langle y\rangle=(1+|y|^2)^{1/2}$
for a vector or a complex number $y$.
We denote the space of  bounded operators 
from a general Banach space $X$ to another $Y$ by $\vL(X,Y)$ 
and abbreviate $\mathcal L(X)=\mathcal L(X,X)$. The space of 
  compact operators from $X$ to $Y$ is denoted by
$\vC(X,Y)$, and  $X^*$ denotes   the dual space of $X$.
\begin{cond}\label{cond:smooth2wea3n1} 
Let $\delta\in (0,1/2]$ be fixed.
For each $a\in\vA\setminus\{a_{\min}\}$ 
there exists a splitting 
$$V_a=V_a^{\textrm{lr}}+V_a^{\textrm{sr}}+V_a^{\textrm{si}}$$ 
into three real-valued measurable functions in $\mathcal
C(H^1_0(\Omega_a),H^1_0(\Omega_a)^*)$ such that 
\begin{enumerate}[label=(\arabic*)]
\item\label{item:cond12bweak3n2} 
$V_a^{\textrm{lr}}$ has first order distributional derivatives 
in $L^1_{\textrm {loc}}(\Omega_a)$, 
and for any $\abs{\alpha}= 1$
\begin{equation*}
\langle x^a\rangle^{1+2\delta}\partial^\alpha
V_a^{\textrm{lr}}\in \vL(H^1_0(\Omega_a), H^1_0(\Omega_a)^*);
\end{equation*}
\item\label{item:cond122same3n2} 
$V_a^{\textrm{sr}}$ satisfies 
\begin{align}
\langle x^a\rangle^{1+2\delta}V_a^{\textrm{sr}}
\in \vL(H^1_0(\Omega_a), L^2(\Omega_a));
\label{eq:2k2we3n}
\end{align}
\item \label{item:cond1223n2} 
$V_a^{\textrm{si}}$ vanishes outside a bounded subset of $\Omega_a$.
\end{enumerate}
\end{cond} 

In the case where hard-cores are absent, the following 
operator-bounded version is also available. 
It is almost identical with the  condition of \cite {ABG1}.
\begin{cond}\label{cond:smooth2wea3n12} 
Let $\delta\in (0,1/2]$ be fixed. 
For each $a\in\vA\setminus\{a_{\min}\}$ 
hard-core interaction  is  absent, i.e.  $\Omega_a=\bX^a$,
and 
there exists a splitting 
$$V_a=V_a^{\textrm{lr}}+V_a^{\textrm{sr}}$$ 
into two real-valued measurable functions in $V_a\in
\vC(H^2(\bX^a),L^2(\bX^a))$ such that 
\begin{enumerate}[label=(\arabic*)]
\item\label{item:cond12bweak3n3} 
$V_a^{\textrm{lr}}$ has first order distributional derivatives 
in $L^1_{\textrm {loc}}(\bX^a)$, and for any $\abs{\alpha}= 1$
\begin{equation*}
\langle x^a\rangle^{1+2\delta}\partial^\alpha
V_a^{\textrm{lr}}\in \vL(H^2(\bX^a), H^{-1}(\bX^a));
\end{equation*}
\item\label{item:cond122same3n3} 
$V_a^{\textrm{sr}}$ satisfies 
\begin{align*}
\langle
  x^a\rangle^{1+2\delta}V_a^{\textrm{sr}}\in \vL(H^2(\bX^a),
  L^2(\bX^a)).
\end{align*}
\end{enumerate}
\end{cond} 
\begin{remark*}
It is easily checked that the proof of the Mourre estimate in \cite{IS1} 
works under  either  Condition~\ref{cond:smooth2wea3n1}  or  Condition~\ref{cond:smooth2wea3n12},
see \eqref{eq:mourre comm}, \eqref{eq:mourre commb} and Lemma~\ref{lemma:Mourre1_hard} below. 
In this paper the Mourre estimate, or Lemma~\ref{lemma:Mourre1_hard}, 
is used as an `input', and we will not give a proof of it.
\end{remark*}

For an $N$-body potential \eqref{eq:170927} and an exterior region \eqref{eq:171002}
satisfying 
{Conditions~\ref{cond:smooth2wea3n1} or \ref{cond:smooth2wea3n12}}
we define the \emph{generalized $N$-body Hamiltonian} $H$ as 
\begin{align*}
 H=H_0+V;\quad H_0=-\tfrac12\Delta,\quad
 \text{on }\vH=L^2(\Omega). 
\end{align*}
Here $\Delta$ is the Laplace--Beltrami operator associated with the inner product on $\bX$,
and we impose the Dirichlet boundary condition on $\partial\Omega$. 
More precisely, $H$ 
is defined as the self-adjoint operator 
associated with the closed quadratic form $\tilde H$:
\begin{align*}
\langle \tilde H\rangle_\psi
=\tfrac12\langle p\psi,p\psi\rangle+\langle \psi,V\psi\rangle
\ \ 
\text{for }\psi\in Q(\tilde H)=H^1_0(\Omega),
\end{align*}
where $p=-\mathrm i\nabla$.
Note that $V$ is  infinitesimally $H_0$-small in the form sense.
For  later use let us be more careful about the domain $\mathcal D(H)$. 
Note that the above quadratic form $\tilde H$ may be considered as a 
bounded operator $H^1_0(\Omega)\to (H^1_0(\Omega))^*$.
Then the self-adjoint operator $H$ is realized by letting 
\begin{align}
H=\tilde H_{|\mathcal D(H)};\quad 
\mathcal D(H)=\bigl\{\psi\in H^1_0(\Omega)\,\big|\, \tilde H\psi\in \mathcal H\bigr\}.
\label{eq:100516}
\end{align}
This is exactly the definition of the Friedrichs extension, see e.g.\ \cite[proof of Theorem XI.7.2]{Yo}. 
Note in particular
that, if Condition~\ref{cond:smooth2wea3n12} is adopted, 
then we have 
\begin{align}
\mathcal D(H)=\mathcal D(H_0)=H^2(\mathbf X)\supset C^\infty_{\mathrm c}(\mathbf X).
\label{eq:180216}
\end{align}
We henceforth denote the quadratic form $\tilde H$ on $H^1_0(\Omega)$,
or the bounded operator $\tilde H\colon H^1_0(\Omega)\to(H^1_0(\Omega))^*$, simply 
by $H$ if there is no confusion.

\subsubsection{Sub-Hamiltonians}
We recall the definition of the \emph{sub-Hamiltonian $H^a$ associated with a cluster 
decomposition $a\in\vA$}. 
For $a=a_{\min}$, 
noting that $V^{a_{\min}}=0$ and $\Omega^{a_{\min}}=\{0\}$, we define
$$H^{a_{\min}}=0\ \ \text{on }\mathcal H^{a_{\min}}=L^2(\Omega^{a_{\min}})=\mathbb C. $$
For $a\neq a_{\min}$, since $\{\mathbf X_b\cap\mathbf X^a\}_{b\subset a}$ forms 
a family of subspaces of $(\# a-1)$-body type in $X^a$, 
we can consider, similarly to the \emph{full} Hamiltonian $H$, 
\begin{equation*}
V^a(x^a)=\sum_{b\subset a} V_{b}(x^b)
,\quad
\Omega^a=\bigcap_{b\subset a} \bigl[\Omega_b+(\bX_b\cap \bX^a)\bigr].
\end{equation*} 
Then we define the associated sub-Hamiltonian $H^a$ as 
\begin{align*}
 H^a=-\tfrac 12
\Delta_{x^a} +V^a\ \ 
\text{on }\mathcal H^a=L^2(\Omega^a)
\end{align*}
with the Dirichlet boundary condition on $\partial \Omega^a$. 
We remark that in particular
\begin{align*}
V^{a_{\max}}=V,\quad
 \Omega^{a_{\max}}=\Omega,\quad
 H^{a_{\max}}=H,\quad 
 \mathcal H^{a_{\max}}=\mathcal H.
\end{align*}

The \textit{thresholds} of $H$ are the eigenvalues of sub-Hamiltonians $H^a$, 
$a\in\mathcal A\setminus \{a_{\max}\}$.
We set 
\begin{align*}
 \vT (H)= \bigcup \bigl\{\sigma_{\pupo}( H^a)\,\big|\,a\in\vA\setminus \{a_{\max}\}\bigr\}.
\end{align*} 
It is known that under 
{Conditions~\ref{cond:smooth2wea3n1} or \ref{cond:smooth2wea3n12}} 
that the set 
$\vT (H)$ is closed and   at most countable. Moreover the set of non-threshold eigenvalues
  is discrete in $\R\setminus \vT (H)$,  and it 
 can only  accumulate  at  points in
$\vT (H)$  from below. See
 \ Lemma~\ref{lemma:Mourre1_hard}, Remark \ref{rem:180123}  and
\cite{FH, IS1, Pe}. 
By the so-called HVZ theorem 
the essential spectrum of $H$ is given by the formula
\begin{equation*}
 \sigma_{\ess}(H)= \bigl[\min \vT(H),\infty\bigr),
\end{equation*} cf. 
\cite[Theorem~X\hspace{-.05em}I\hspace{-.05em}I\hspace{-.05em}I.17]{RS}.

\subsubsection{Unique continuation property}
Most of the results of the paper depend on the \textit{unique continuation property}.
Due to singularities of pair-potentials $V_a$ 
and the hard-cores $\bX^a\setminus\Omega_a$ this property does not necessarily hold for our Hamiltonian $H$.
In this paper we are going to \emph{assume} this property  rather than imposing technical sufficient conditions 
on {$V_a$} and $\Omega_a$.

For any $a\in\mathcal A\setminus\set{a_{\min}}$
we introduce the \emph{locally $H_0^1$ space} by 
\begin{align*}
H^1_{0,\mathrm{loc}}(\Omega^a)
=\bigl\{\psi\in L^2_{\mathrm{loc}}(\Omega^a)\,\big|\, 
\chi\psi\in H^1_0(\Omega^a)\text{ for any }\chi
\in C^\infty_{\mathrm c}(\bX^a)\bigr\}.
\end{align*}
Then, since $H^a\colon H_0^1(\Omega^a)\to (H_0^1(\Omega^a))^*
\subset \mathcal D'(\Omega)$ is a local operator, 
it naturally extends as 
$$H^a\colon H^1_{0,\mathrm{loc}}(\Omega^a)\to \mathcal D'(\Omega^a)$$
by using a partition of unity. 
The vector $H^a\psi$ for $\psi\in H^1_{0,\mathrm{loc}}(\Omega^a)$  may be referred to as  
`$H^a\psi$ \emph{in the distributional sense}'.
We call a function $\phi\in H^1_{0,\mathrm{loc}}(\Omega^a)$
a {\it generalized Dirichlet eigenfunction for $H^a$ with eigenvalue $E\in \mathbb C$}, 
if it satisfies 
$$H^a\phi=E\phi\ \ \text{in the distributional sense}.$$ 
\begin{cond}\label{cond:smooth2wea3n2} 
The unique continuation property holds for $H^a$ for all $a\neq a_{\min}$:
If $\phi\in H^1_{0,\mathrm{loc}}(\Omega^a)$ 
is a generalized Dirichlet eigenfunction for $H^a$ 
and $\phi=0$ on a non-empty open subset of $\Omega^a$, then $\phi=0$ on $\Omega^a$.
\end{cond}

For a particular result on the unique continuation property for
$N$-body Schr\"odinger operators (without hard-core interaction) we refer to \cite{Geo}. To our
knowledge this property is not well understood in the $N$-body case.

\subsection{Results}\label{subsec:result}
In this subsection we state all the main results of the paper.

\subsubsection{Rellich type  theorems}

Let us first 
recall the definitions of the \emph{Besov spaces 
associated with the multiplication operator
$|x|$ on $\vH$}. 
Set
\begin{align*}
F_0&=F\bigl(\bigl\{ x\in \Omega\,\big|\,\abs{x}<1\bigr\}\bigr),\\
F_n&=F\bigl(\bigl\{ x\in \Omega\,\big|\,2^{n-1}\le \abs{x}<2^n\bigr\} \bigr)
\quad \text{for }n=1,2,\dots,
\end{align*}
where $F(S)$ is the sharp characteristic function of any given  subset $S\subset \Omega$. 
The Besov spaces $\mathcal B$, $\mathcal B^*$ and $\mathcal B^*_0$ are defined as 
\begin{align*}
\mathcal B&=
\bigl\{\psi\in L^2_{\mathrm{loc}}(\Omega)\,\big|\,\|\psi\|_{\mathcal B}<\infty\bigr\},\quad 
\|\psi\|_{\mathcal B}=\sum_{n=0}^\infty 2^{n/2}
\|F_n\psi\|_{{\mathcal H}},\\
\mathcal B^*&=
\bigl\{\psi\in L^2_{\mathrm{loc}}(\Omega)\,\big|\, \|\psi\|_{\mathcal B^*}<\infty\bigr\},\quad 
\|\psi\|_{\mathcal B^*}=\sup_{n\ge 0}2^{-n/2}\|F_n\psi\|_{{\mathcal H}},
\\
\mathcal B^*_0
&=
\Bigl\{\psi\in \mathcal B^*\,\Big|\, \lim_{n\to\infty}2^{-n/2}\|F_n\psi\|_{{\mathcal H}}=0\Bigr\},
\end{align*}
respectively.
Denote the standard \emph{weighted $L^2$ spaces} by 
$$
L_s^2=\inp{x}^{-s}L^2(\Omega)\ \ \text{for }s\in\mathbb R ,\quad
L_{-\infty}^2=\bigcup_{s\in\R}L^2_s,\quad
L^2_\infty=\bigcap_{s\in\mathbb R}L_s^2
.
$$
Then note that for any $s>1/2$
\begin{align}\label{eq:nest}
 L^2_s\subsetneq \mathcal B\subsetneq L^2_{1/2}
\subsetneq \mathcal H
\subsetneq L^2_{-1/2}\subsetneq \mathcal B^*_0\subsetneq \mathcal B^*\subsetneq L^2_{-s}.
\end{align}

Now we have  Rellich type  theorems  on  the following form extending \cite{IS1,IS2}.

\begin{thm}\label{thm:priori-decay-b_0} 
Suppose 
Condition~\ref{cond:smooth2wea3n1} or Condition \ref{cond:smooth2wea3n12}.
Let $\phi\in \mathcal B_0^*\cap H^1_{0,\mathrm{loc}}(\Omega)$, $E\in\mathbb R$
and $\rho\ge 0$, and assume that 
\begin{align*}
(H-E)\phi(x)=0\ \ \text{for $|x|> \rho$ in the distributional sense}.
\end{align*}
Set $\alpha_0
=\sup\bigl\{\alpha\ge 0\,\big|\, \mathrm e^{\alpha|x|}\phi\in \mathcal B_0^*
\bigr\}\in [0,\infty]$. Then 
\begin{align}
E+\tfrac12\alpha_0^2\in\mathcal T(H)\cup\{\infty\}.
\label{eq:1711051}
\end{align}
\end{thm} 

\begin{thm}\label{thm:priori-decay-b_0b} 
Suppose 
Condition~\ref{cond:smooth2wea3n1}.
Let $\phi\in \mathcal B_0^*\cap H^1_{0,\mathrm{loc}}(\Omega)$, $E\in\mathbb R$
and $\rho\ge 0$, and assume that 
\begin{align*}
(H-E)\phi(x)=0\ \ \text{for $|x|> \rho$ in the distributional sense}.
\end{align*}
Suppose 
$\sup\bigl\{\alpha\ge 0\,\big|\, \mathrm e^{\alpha|x|}\phi\in \mathcal B_0^*
\bigr\}=\infty$. 
  Then there exists $\rho'\ge 0$ such that 
\begin{align}
\phi(x)= 0\ \ \text{for }|x|> \rho'.
\label{eq:17111522}
\end{align}
\end{thm}

The combination of Theorems~\ref{thm:priori-decay-b_0} 
and \ref{thm:priori-decay-b_0b} extends the classical Rellich theorem
for  $N=1$ and $E>0$. We refer to \cite{Is} and the references therein for an account of the
history of Rellich's theorem. With the unique continuation property we
can extend the  classical Rellich theorem to any $N$, see  
\ref{item:clasReExt} below.

\begin{corollary}\label{cor:1710291201}
Suppose 
Conditions ~\ref{cond:smooth2wea3n1}  and \ref{cond:smooth2wea3n2}.
\begin{enumerate} [1)]
\item
Let $\phi\in \mathcal B_0^*\cap H^1_{0,\mathrm{loc}}(\Omega)$ 
be a generalized Dirichlet eigenfunction for $H$ 
with real eigenvalue $E\in\mathbb R$,
and set 
\begin{align*}
\alpha_0
=\sup\bigl\{\alpha\ge 0\,\big|\, \mathrm e^{\alpha|x|}\phi\in \mathcal B_0^* 
\bigr\}\in [0,\infty].
\end{align*} 
Then 
$E+\tfrac12\alpha_0^2\in \mathcal T(H)\cup\{\infty\}$, and if 
 $\alpha_0=\infty$ the function  $\phi=0$ on $\Omega$.
\item \label{item:absencePos}
There are no positive thresholds for $H$,
and there are no nonzero generalized Dirichlet eigenfunctions for $H$
in $\mathcal B_0^*$ with a  positive eigenvalue.
\item  \label{item:clasReExt} 
Let $\phi\in \mathcal  B_0^*\cap H^1_{0,\mathrm{loc}}(\Omega)$, $E>0$
and $\rho\ge 0$, and assume that 
\begin{align*}
(H-E)\phi(x)=0\ \ \text{for $|x|> \rho$ in the distributional sense}.
\end{align*}
  Then there exists $\rho'\ge 0$ such that 
$\phi(x)= 0\ \ \text{for }|x|> \rho'$.
\end{enumerate}
\end{corollary} 
\begin{proof}
The assertions follow from
Theorems~\ref{thm:priori-decay-b_0} 
and \ref{thm:priori-decay-b_0b}. 
Note that \ref{item:absencePos}  needs an 
induction argument on $N$, cf. \cite{IS1}.
\end{proof}

Using  the proof of Theorem~\ref{thm:priori-decay-b_0} 
we can  extend the result of \cite{Pe}, or \cite[Theorem~6.11]{HuS},
showing that the set of non-threshold eigenvalues of $H$ can accumulate only at points in $\mathcal T(H)$ from below,
see Lemma~\ref{lem:14.10.4.1.17ffaabb} and Remark~\ref{rem:180123}.

\subsubsection{LAP bounds}

Next we present the sharp \emph{LAP bounds}.
For any interval $I\subset \mathbb R$ let us set 
$$I_{\pm}:=\{z\in\mathbb C\,|\,\Re z \in I,\,0<\pm\Im z \leq 1\}.$$

\begin{thm}\label{thmlapBnd} 
Suppose 
{Condition~\ref{cond:smooth2wea3n1} or Condition \ref{cond:smooth2wea3n12}.
Let} $I\subset \mathbb R\setminus (\sigma_{\mathrm{pp}}(H)\cup\mathcal T(H))$ be a compact interval.
Then there
 exists $C>0$ such that for all $z\in I_\pm $ and $\psi\in \mathcal B$
\begin{align}\label{eq:lap-besov-spacebnd}
 \norm{R(z)\psi}_{\mathcal B^*}+\|pR(z)\psi\|_{\mathcal B^*}\leq C\norm{\psi}_{\mathcal B}.
\end{align}
\end{thm}

The LAP bound stated in Theorem~\ref{thmlapBnd} 
is new in that they allow form-bounded local singularities (by Condition~\ref{cond:smooth2wea3n1}). 
To our knowledge the LAP bounds for form-bounded singularities
have been obtained only for $N=1$ and with $\Omega=\R^3$, see  \cite{BMP}. 
See also \cite[pp.~270--271]{ABG2} for a review of the LAP bounds. 

The weighted $L^2$ space version of the LAP bounds, proven in
\cite{ABG1,Ta}, are essentially consequences of \eqref{eq:nest} and
\eqref{eq:lap-besov-spacebnd} since 
{Condition~\ref{cond:smooth2wea3n12} 
is almost identical with the conditions of \cite{ABG1,Ta}.}
Similarly  our result includes 
the Besov space refinements of \cite{BGM1,BGM2} (at least
essentially).  In conclusion, 
Theorem~\ref{thmlapBnd} extends previous results for the usual
$N$-body Schr\"{o}dinger operators without {hard-cores.}

We also mention that the LAP bounds for hard-core models were considered 
in \cite {BGS} with some regularity conditions on the obstacles 
and with operator-bounded local singularities. 
Our result is more general, again.

\subsubsection{Rescaled Graf function}\label{subsec:resc-graf-funct}

To state  \textit{microlocal resolvent bounds} 
 we  introduce certain  (rescaled) operators $A_R, B_R$.
They are defined in terms of a function $r_1\in C^\infty (\bX)$ 
with the following properties. 
We will not verify the existence of such $r_1$,
but only enumerate the properties required in this paper.
Let us denote the gradient vector field and the Hessian of $r_1^2/2$ by 
$$\tilde\omega_1=\tfrac12\mathop{\mathrm{grad}} r_1^2,\quad 
\tilde h_1=\tfrac12\mathop{\mathrm{Hess}}r_1^2,$$
respectively. It is standard to identify 
the tangent space of $\bX$ at each $x\in\bX$ with $\bX$ itself.
Then we may consider $\tilde \omega_1(x)\in \bX$ for each $x\in\mathbf X$.
Similarly, using the inner product structure on $\bX$, 
we may consider $\tilde h_1(x)\colon \bX\to \bX$ as a   linear map for each $x\in \bX$. 
Let  $\mathbb N_0=\{0\}\cup\mathbb N$ with $\mathbb
N=\{1,2,\ldots\}$. We assume that there exist $r_1\in C^\infty(\mathbf X)$ and 
a smooth partition  of unity $\{\eta_{1,a}\}_{a\in \vA}$ on $\bX$ obeying:
 \begin{enumerate}[\normalfont (1)]
 \item\label{item:6k} 
There exist $c,C>0$ such that 
for any $a,b\in\mathcal A$ with $a\not\subset b$ and $x\in\mathop{\mathrm{supp}}\eta_{1,b}$
\begin{align}
|x^a|\ge c,\quad |x^b|\le C;
\label{eq:171018}
\end{align}

\item 
There exists $C'>0$ such that 
for any $x\in\mathbf X$
\begin{align}
r_1(x)\ge 1,\quad 
\bigl|r_1(x)-|x|\bigr|\le C';
\label{eq:17121}
\end{align}

 \item \label{item:5k}
There exists $c'>0$ such that 
for any $a\in\mathcal A$ and $x\in\bX$ with $|x^a|\le c'$ 
\begin{align}
\tilde \omega_1^a(x)=0;
\label{eq:17100721}
\end{align}

\item \label{item:4k}
For any $x,y\in\bX$
$$\bigl\langle y,\tilde h_1(x)y\bigr\rangle\ge \sum_{a\in\mathcal A}\eta_{1,a}(x) |y_a|^2;$$ 

\item\label{item:7ik} 
For any $\alpha\in \N_0^{\dim \bX}$ and $k\in
 \N_0$ there exists $C_{\alpha k}>0$ such that for any $x\in\bX$ 
 \begin{equation}\label{eq:boun_constr}
 \sum_{a\in\mathcal A}\bigl|\partial ^\alpha\eta_{1,a}(x)\bigr|
+\bigl|\partial^\alpha(x\cdot\nabla)^k\parb{\tilde \omega_1(x)-x}\bigr|\leq C_{\alpha k}.
 \end{equation}
 \end{enumerate}
For a construction of such $r_1$ and $\{\eta_{1,a}\}_{a\in\mathcal A}$ we refer to \cite{Gr},
see also \cite{De, Sk}.
We note that \eqref{eq:17121} in fact is verified from the other properties.
The former bound of \eqref{eq:17121} always holds if we add a large positive constant to $r_1$,
and the latter follows by integrating \eqref{eq:boun_constr} with $\alpha=0$ and $k=0$.

Now we set for large $R\ge 1$
$$
r_R(x)=Rr_1(x /R),\quad 
\tilde\omega_R=\tfrac12\mathop{\mathrm{grad}}r_R^2,\quad
\omega_R=\mathop{\mathrm{grad}}r_R,\quad
\eta_R(x)=\eta_1(x/R),
$$
and define the self-adjoint operators $A_R,B_R$ on $\mathcal H$ as 
\begin{align}
A_R=\tfrac12\bigl(\tilde\omega_R\cdot p+p\cdot \tilde \omega_R\bigr)
,\quad
B_R=\tfrac12\parb{\omega_R\cdot p+p\cdot \omega_R};\quad 
p=-\i\nabla,
\label{eq:1710220}
\end{align} 
respectively. 
The vector field $\tilde\omega_R$ is called the \textit{rescaled Graf vector field}, 
and $A_R$ is a \textit{conjugate operator} in the Mourre theory,
cf. \cite {Sk, IS1}. 
We remark that, however, the operator $A_R$ is only auxiliarily  used in this paper, 
or we can actually remove it completely from this paper. 
Instead, $B_R$ plays a central role in our theory. 
We will investigate its properties in Section~\ref{sec:17093015}.

In the sequel we will often suppress the dependence on the parameter $R\ge 1$ of the above quantities,
writing simply 
$$r=r_R,\quad 
\tilde\omega=\tilde\omega_R,\quad \omega=\omega_R,\quad A=A_R,\quad B=B_R,\quad \eta=\eta_R.$$

\subsubsection{Microlocal resolvent bounds and applications}

Now we present some  microlocal resolvent bounds.
Define a function $d:\R\to\R$ as
\begin{equation*}
d(\lambda)=\begin{cases}
\inf\bigl\{\lambda-\tau\,\big|\,\tau\in \vT(H)\cap (-\infty,\lambda]\bigr\}&
\text{if }\vT(H)\cap (-\infty,\lambda]\neq \emptyset,\\
1&\text{if }\vT(H)\cap (-\infty,\lambda]=\emptyset,
\end{cases}
\end{equation*} 
and introduce for any $\lambda\in\mathbb R$ and $I\subset \mathbb R$
\begin{align}
\gamma(\lambda)=\sqrt{2d(\lambda)},\quad 
\gamma_-(I)=\inf_{\lambda\in I}\gamma(\lambda)=\inf_{\lambda\in I}\sqrt{2d(\lambda)}.
\label{eq:17111910}
\end{align}
With $\delta\in (0,1/2]$ from Condition~\ref{cond:smooth2wea3n1} or  
\ref{cond:smooth2wea3n12} let 
\begin{align}
\kappa=\delta/(1+2\delta)\in (0,1/4].
\label{eq:171114}
\end{align}

\begin{thm}\label{microLoc} 
Suppose 
Conditions~\ref{cond:smooth2wea3n1} or \ref{cond:smooth2wea3n12}. 
Let $I\subset \mathbb R\setminus (\sigma_{\mathrm{pp}}(H)\cup\mathcal T(H))$
be a compact interval, and take both $R\ge 1$ and $\tilde\gamma>0$ sufficiently large.
Then for any $\beta\in (0,\kappa)$ and 
$F\in C^\infty(\mathbb R)$ with 
$$\mathop{\mathrm{supp}}F\subset (-\infty,\gamma_-(I))\cup
(\tilde\gamma,\infty) \mand 
F'\in C^\infty_{\mathrm c}(\mathbb R),$$
there exists $C>0$ such that 
for all $z\in I_\pm$ and $\psi\in L^2_{1/2+\beta}$
 \begin{align*}
 \begin{split}
\|F(\pm B)R(z)\psi\|_{L^2_{-1/2+\beta}}
\le C\|\psi\|_{L^2_{1/2+\beta}},
\end{split}
 \end{align*} 
respectively.
\end{thm}
\begin{remarks*}
\begin{enumerate}
\item
 The critical exponent $\kappa$ in Theorem~\ref{microLoc}
and Corollaries~\ref{cor:holdCon} and \ref{cor:Som} below,
which is worse than $\delta$,
comes from an optimization in the proof of Lemma~\ref{lemma:sing}
under Condition~\ref{cond:smooth2wea3n1}. 
Under Condition~\ref{cond:smooth2wea3n12} 
we can actually choose $\kappa=\delta$.
\item
If higher commutators were available like in \cite{GIS}, 
similar bounds would hold for powers of the resolvent. 
However, we have at most  second commutators available 
under 
Conditions~\ref{cond:smooth2wea3n1} or \ref{cond:smooth2wea3n12},
cf.\ Lemma~\ref{lemma:sing}. For the same reason we need $\beta$ to be small.
\end{enumerate}
\end{remarks*}

The first application of the above results is H\"older continuity of the resolvent, 
and in particular the {LAP.}

\begin{corollary} \label{cor:holdCon} 
Suppose 
 Conditions~\ref{cond:smooth2wea3n1} or \ref{cond:smooth2wea3n12}.
Let $I\subset \mathbb R\setminus (\sigma_{\mathrm{pp}}(H)\cup\mathcal T(H))$
be a compact interval, 
and let $s>1/2$ and $\beta\in (0, \min\{\kappa, s-1/2\})$.
Then there 
exists $C>0$ such that 
for all $k\in\{0,1\}$, $z\in I_\pm$ and $z'\in I_\pm$, respectively,
\begin{align}
\|p^kR(z)-p^kR(z')\|_{\mathcal L(L^2_s,L^2_{-s})}
\le C|z-z'|^\beta.
\label{eq:171125}
\end{align}
In particular, for any $E\in I$ and $s>1/2$ the following boundary values exist:
\begin{align*}
p^kR(E\pm \i 0):=\lim_{\epsilon\to 0_+}p^kR(E\pm\i \epsilon)
\text{\ \ in }\mathcal L(L^2_s,L^2_{-s}),
\end{align*} 
respectively. The same boundary values are realized (in an extended
form) as 
\begin{align*}
p^kR(E\pm \i 0)= \swslim_{\epsilon\to 0_+} p^kR(E\pm\i \epsilon)
\text{\ \ in }\vL(\mathcal B,\mathcal B^*),
\end{align*}
respectively.
\end{corollary}

The second application is a  \emph{microlocal Sommerfeld uniqueness result},
which characterizes the limiting resolvents $R(E\pm\mathrm i0)$
by the \emph{Helmholtz equation} and  \emph{microlocal radiation conditions}.
Given $\psi\in L^2_{\mathrm{loc}}(\Omega)$,
we say a function $\phi\in H^1_{0,\mathrm{loc}}(\Omega)$ is 
a {\it generalized Dirichlet solution to $(H-E)\phi=\psi$}, 
if it satisfies 
$$(H-E)\phi=\psi\ \ \text{in the distributional sense}.$$

\begin{corollary} \label{cor:Som} 
Suppose 
Conditions~\ref{cond:smooth2wea3n1} or \ref{cond:smooth2wea3n12}.
Let $E\in \mathbb R\setminus (\sigma_{\mathrm{pp}}(H)\cup\mathcal T(H))$,
and take $R\ge 1$ sufficiently large.
Let $\psi\in r^{-\beta}\mathcal B$ with $\beta\in [0,\kappa)$.
Then $\phi=R(E\pm \i 0)\psi\in \mathcal B^*\cap H^1_{0,\mathrm{loc}}(\Omega)$
satisfies 
\begin{enumerate}[(1)]
\item \label{item:13.7.29.0.28}
$\phi$ is a generalized Dirichlet solution to $(H-E)\phi=\psi$,
\item \label{item:13.7.29.0.29}
there exists $\tilde\gamma>0$ such that 
for any $F\in C^\infty(\mathbb R)$ with
$$\mathop{\mathrm{supp}}F\subset (-\infty,\gamma(E))\cup (\tilde\gamma,\infty)\mand
F'\in C^\infty_{\mathrm c}(\mathbb R),$$
the functions 
$F(\pm B)\phi$ belong to $r^{-\beta}\mathcal B^*_0$,
\end{enumerate}
respectively.
Conversely, if $\phi'\in L^2_{-\infty}\cap H^1_{0,\mathrm{loc}}(\Omega)$ satisfies 
\begin{enumerate}[(1${}'$)]
\item \label{item:13.7.29.0.28b}
$\phi'$ is a generalized Dirichlet solution to $(H-E)\phi'=\psi$,
\item \label{item:13.7.29.0.29b}
there exists $\gamma>0$ such that 
for any $F\in C^\infty(\mathbb R)$ with
$$\mathop{\mathrm{supp}}F\subset (-\infty,\gamma) \mand
F'\in C^\infty_{\mathrm c}(\mathbb R),$$
the functions $F(\pm B)\phi'$ belong to $\mathcal B^*_0$,
\end{enumerate}
then $\phi'=R(E\pm \i 0 )\psi$, respectively.
\end{corollary}

\section{Preliminaries: Operator $B$}\label{sec:17093015}

In this section we provide various preliminaries needed for the proofs of our main results. 
In Section~\ref{subsubsec:Smooth sign function}
we introduce notation frequently used in the later arguments.
Section~\ref{subsec:Functional calculus} 
includes the Helffer--Sj\"ostrand formula and its direct application to compute commutators.
In Section~\ref{subsec:Operatat B} we formulate the self-adjoint realization of 
the conjugate operator $B$ from \eqref{eq:1710220}. 
We investigate the first commutator $\mathrm i[H,B]$ in Section~\ref{subsec:17111117}, 
and the second commutator $\mathrm i[\mathrm i[H,B],B]$ in Section~\ref{subsec:17111119}. 
The control of the second commutator $\mathrm i[\mathrm i[H,B],B]$ is a key,
and Section~\ref{subsec:17111119} is one of the most technical parts
of  the paper. 
Remark that for the absense of positive $L^2$-eigenvalues \cite{IS1} only
 a  first commutator  was needed.

\subsection{Notation}\label{subsubsec:Smooth sign function}

This is a short subsection devoted to  some notation only.

Let $T$ be an linear operator on $\mathcal H=L^2(\Omega)$ such that
$T,T^*:L^2_\infty\to L^2_\infty$, and let $t\in\mathbb R$.  Then we
say that   $T$ is an {\emph{operator of order $t$}}, if 
 for each $s\in\mathbb R$  the restriction  $T_{|L^2_\infty}$ extends to
 an operator $T_s\in\vL(L^2_{s}, L^2_{s-t})$. Alternatively stated,
 for  any $R\geq 1$ and with $r=r_R$
\begin{align*}
\|r^{s-t}Tr^{-s}f\|\le C_s\|f\| \text{ for all }f\in L^2_\infty.
\end{align*} Note  (for consistency) that  $T_s$ extends the restriction $T_{|
\vD(T)\cap L^2_{s}}$. 
 If   $T$ is of {order $t$}, we write 
\begin{align}
T=O(r^t).
\label{eq:1712022}
\end{align} 
Note also that, if $T=O(r^t)$ and $S=O(r^s)$, then $T^*=O(r^t)$ and $TS=O(r^{t+s})$.

Define the \emph{Sobolev spaces $\mathcal H^s$ of order $s\in\mathbb R$ 
associated with $H$} as
\begin{align}
\mathcal H^s=(H+E_{\mathrm{min}}+1)^{-s/2}\mathcal H;\quad
E_{\mathrm{min}}=\min\sigma(H).
\label{eq:17111317}
\end{align}
We note that
$$\vH^1=Q(H)=H^1_0(\Omega),\quad 
\mathcal H^2=\mathcal D(H),\quad 
\mathcal H^{-1}=(H^1_0(\Omega))^*,\quad 
\mathcal H^{-2}=(\mathcal H^2)^*.$$

Let $\chi\in C^\infty(\R)$ be a real-valued function such that 
\begin{align}
\chi(t)
=\left\{\begin{array}{ll}
1 &\mbox{ for } t \le 1, \\
0 &\mbox{ for } t \ge 2,
\end{array}
\right.
\quad
\chi'\le 0,
\ \mand\ 
\chi^{1/2},|\chi'|^{1/2}\in C^\infty(\mathbb R).
\label{eq:14.1.7.23.24}
\end{align}
 We then define \emph{smooth cut-off functions} 
$\chi_m,\bar\chi_m,\chi_{m,n}\in C^\infty(\mathbf X)$ for $n>m\ge 0$ and $R\ge 1$ as 
\begin{align}
\chi_m=\chi(r/2^m),\quad \bar \chi_m=1-\chi_m,\quad \chi_{m,n}=\bar\chi_m\chi_n.
\label{eq:1711021}
\end{align} 
Here $r=r_R$ in fact depends on $R\ge 1$, but the dependence on $R$ is suppressed.

Next we construct a \emph{smooth sign function} $\zeta\in C^\infty(\mathbb R)$.
The approach of \cite{IS2} was based largely 
on such a function,
but our construction below is a slightly simplified one.
Choose a function $\zeta_1\in C^\infty(\mathbb R)$
such that 
\begin{align*}
\zeta_1'(b)\ge 0 \ \ \text{for }b\in\mathbb R,
\quad
\sqrt{\zeta_1'}\in C^\infty_{\mathrm c}(\mathbb R),\quad 
\zeta_1(b)=\left\{
\begin{array}{ll}
-1& \text{for }b\le -1,\\
2b&\text{for }-1/4\le b\le 1/4,\\
1& \text{for }b\ge 1,
\end{array}
\right.
\end{align*}
and let $\zeta\in C^\infty(\mathbb R)$ be defined as 
\begin{align*}
\zeta(b)=\zeta_\epsilon(b)=\zeta_1(b/\epsilon)
\ \ \text{for }
\epsilon\in(0,1),\ b\in\mathbb R 
.
\end{align*}
We will use the following elementary properties of $\zeta$ (we omit
the proof).

\begin{lemma}\label{lem:170928}
For any $\epsilon\in(0,1)$
the above function $\zeta\in C^\infty(\mathbb R)$ satisfies that 
\begin{align*}
\zeta'(b)\ge 0\ \ \text{for }b\in\mathbb R,\quad 
\zeta'(b)=0\text{\ \ for }|b|\ge \epsilon,\quad
\sqrt{\zeta'}\in C^\infty_{\mathrm c}(\mathbb R),
\end{align*}
and, in addition, that for any $c>0$ and $b\in\mathbb R$
\begin{align*}
b\zeta(b)+c\zeta'(b)\ge \min\{\epsilon/8,2c/\epsilon\}.
\end{align*}
\end{lemma}

\subsection{Functional calculus}\label{subsec:Functional calculus}

Here we present the Helffer--Sj\"ostrand formula to write down 
functions of self-adjoint operators, 
and its application to commutators. 
The following results are abstract and time-independent versions of similar results 
in the literature, typically more sophisticated ones.
We omit most of the proofs, 
and only refer e.g.\ to \cite[Lemma 3.5]{HeS}, \cite[Section 2]{GIS} or \cite[Appendix C]{DG}.

For any $t\in\mathbb R$ we set 
\begin{align*}
\mathcal F^t=\bigl\{f\in C^\infty(\mathbb R)\,\big|\, 
|f^{(k)}(x)|\le C_k\langle x\rangle^{t-k}\text{ for any }k\in\mathbb N_0\text{ and }x\in\mathbb R\bigr\}
.
\end{align*}
It is known that for any $f\in\mathcal F^t$, $t\in\mathbb R$,
there always exists an \textit{almost analytic extension} $\tilde f\in C^\infty (\C)$ such that 
\begin{equation*} 
\tilde f_{|\mathbb R}=f, 
\quad 
|\tilde f(z)|\le C\langle z\rangle^t,\quad 
\bigl|(\bar{\partial}\tilde{f})(z)\bigr| 
\leq C_k|\Im z|^{k}\left\langle z\right\rangle^{t-k-1} 
\ \ \text{for any }k\in\mathbb N_0.
\end{equation*} 
Here one can choose $\tilde f\in C^\infty_{\mathrm c}(\mathbb C)$
if $f\in C^\infty_{\mathrm c}(\mathbb R)$.

\begin{subequations}
\begin{lemma}\label{lem:A1} 
Let $T$ be a self-adjoint operator on $\vH$, and let $ f\in \mathcal F^t$ with $t\in\mathbb R$.
Take an almost analytic extension $\tilde f\in C^{\infty }(\C)$ of $f$,
and set 
\begin{align}
\mathrm d\mu_f(z)=\pi^{-1}(\bar\partial\tilde f)(z)\,\mathrm du\mathrm dv;\quad 
z=u+\i v.
\label{eq:17092920}
\end{align}
Then for any $k\in\mathbb N_0$ with $k>t$ 
the operator $f^{(k)}(T)\in\mathcal L(\mathcal H)$ is expressed as 
\begin{equation}\label{82a0} 
f^{(k)}(T) 
=
(-1)^kk!\int _{\C}(T -z)^{-k-1}\,\mathrm d\mu_f(z)
.
\end{equation}
\end{lemma}
\end{subequations}

The expression \eqref{82a0} is the well known Helffer--Sj\"ostrand formula.
We omit its verification. The formula 
   is  useful when we compute and bound commutators. 
In general there are  several variations of  the  definition of a commutator,
and in this paper we do not fix a particular one. 
It will be  clear from the context 
in what sense we will be considering a commutator.
Typically, for symmetric operators $T,S$,  we first define $\mathrm
i[T,S]$ as the  symmetric quadratic form 
\begin{align*}
\langle \mathrm i[T,S]\rangle_\psi
=2\langle \mathop{\mathrm{Im}}(ST)\rangle_\psi
=\mathrm i\langle T\psi,S\psi\rangle-\mathrm i\langle S\psi,T\psi\rangle
\ \ \text{for }\psi\in \mathcal D(T)\cap\mathcal D(S),
\end{align*}
and then extend it to a larger space. 

Let us provide an example of a commutator formula derived from Lemma~\ref{lem:A1}.

\begin{corollary}\label{cor:A2} 
Let $T$ be a self-adjoint operator on $\vH$, 
 $S$ be a  symmetric  relatively $T$-bounded operator,
and assume that there exists a bounded extension
$$(|T|+1)^{-\epsilon/2}\bigl(\mathrm i[T,S]\bigr) (|T|+1)^{-\epsilon/2} \in\mathcal L(\mathcal H)
\ \ \text{for some }\epsilon\in [0,2].$$
Let a   real-valued $f\in \mathcal F^t$  with $t<1-\epsilon$ be given, and let $\mathrm d\mu_f$
be given by \eqref{eq:17092920}.
Then, as a quadratic form on $\mathcal D(f(T))\cap \mathcal D(S)$,
\begin{equation*}
\mathrm i[f(T),S] 
=-
\int _{\C}
(T-z)^{-1}\bigl(\mathrm i[T,S]\bigr) (T-z)^{-1}\,\mathrm d\mu_f(z),
\end{equation*}
and it extends to  a bounded self-adjoint operator on $\mathcal H$.
\end{corollary}
\begin{proof}
By \eqref{82a0} the assertion for $t<0$ is obvious. 
Then the general case $t<1-\epsilon$ follows by approximating $f$ by functions from $C^\infty_{\mathrm c}(\mathbb R)$. 
We omit the details.
\end{proof}

Another example is the commutator $[f(H),r^s]$ treated below.
As in \cite{GIS} `phase-space localizations' stated in terms of 
functions of $H$, $r$ and $B$ will be  important.
The other two commutators of this triple of operators will be discussed later in 
Lemmas~\ref{lem:171113b} and \ref{lem:17111416}.
In the proofs of the main theorems 
we will repeatedly use Lemmas~\ref{lem:171113}, \ref{lem:171113b},
\ref{lem:17111416} and \ref{lem:17111515}.

\begin{lemma}\label{lem:171113}
Suppose 
{Condition~\ref{cond:smooth2wea3n1} or Condition \ref{cond:smooth2wea3n12}}. 
\begin{enumerate}
\item 
For any  $f\in \mathcal F^t$ with $t<0$ the operator   $f(H)$ is of order $0$.
\item\label{item:180128}
Let $f\in \mathcal F^t$ with $t<1/2$, $R\ge 1$ and $s\in\mathbb R$.
Then $\mathrm i[f(H),r^s]$ 
has an expression, 
as a sesquilinear  form on $\mathcal D(f(H))\cap L^2_{\max\{0,s\}}$,
\begin{align}\label{eq:fComR}
\mathrm i[f(H),r^s] 
=
-s
\int _{\C}
(H-z)^{-1}\bigl(\mathop{\mathrm{Re}}(r^{s-1}\omega\cdot p)\bigr) (H-z)^{-1}\,\mathrm d\mu_f(z).
\end{align} 
In particular  $\mathrm i[f(H),r^s]$ is of order $s-1$.
\end{enumerate}
\end{lemma}
\begin{proof}\begin{subequations}
\textit{(1)}\quad
Fix any $R\ge 1$, which defines $r=r_R$.
By \eqref{82a0}  it suffices to show that
for each $s\in\mathbb R$ there exist $C(s)>0$  such that 
\begin{align}
\|(H-z)^{-1}\|_{\mathcal L(L^2_s)}\le C(s)|\mathop{\mathrm{Im}}z|^{-1}\bigl(\inp{z}/|\mathop{\mathrm{Im}}z|\bigr)^{|s|+1}
\ \ \text{for }z\in\mathbb C\setminus\mathbb R.
\label{eq:171202}
\end{align}
By considering the adjoint we may assume $s\ge 0$ without  loss of generality. 
We first let $s\in[0,1]$. 
Then as an operator in $\mathcal L(L^2_s,L^2_{-s})$ 
we can calculate 
\begin{align}
\begin{split}
&\mathrm i\bigl[(H-z)^{-1},r^s\bigr]
\\&
=\slim_{t\to 0}t^{-1}\bigl((H-z)^{-1}\mathrm e^{\mathrm itr^s}-\mathrm e^{\mathrm itr^s}(H-z)^{-1}\bigr)
\\&
=\slim_{t\to 0}t^{-1}(H-z)^{-1}\bigl(\mathrm e^{\mathrm itr^s}H_0-H_0\mathrm e^{\mathrm itr^s}\bigr)(H-z)^{-1}
\\&
=\slim_{t\to 0}(H-z)^{-1}
\bigl[-\tfrac s2\bigl(r^{s-1}\mathrm e^{\mathrm itr^s}\omega\cdot p
+p\cdot\omega \mathrm e^{\mathrm itr^s}r^{s-1}\bigr)\bigr](H-z)^{-1}
\\&
=-s(H-z)^{-1}\bigl(\mathop{\mathrm{Re}}(r^{s-1}\omega\cdot p)\bigr)(H-z)^{-1}.
\end{split}
\label{eq:171202b}
\end{align}
Here, noting that $\mathcal H^2\subset \mathcal H^1$ and that 
$\mathrm e^{\mathrm itr^s}$ preserves $\mathcal H^1$, 
we could safely implement integrations by parts for  the third equality of \eqref{eq:171202b}
without boundary contributions
 {even under Condition~\ref{cond:smooth2wea3n1}}.
The last expression of \eqref{eq:171202b} is obviously bounded on
$\mathcal H$. Whence by the formula
\begin{align}\label{eq:1comForm}
r^s(H-z)^{-1}r^{-s}
&=
(H-z)^{-1}
-\mathrm is (H-z)^{-1}\bigl(\mathop{\mathrm{Re}}(r^{s-1}\omega\cdot p)\bigr)(H-z)^{-1}r^{-s},
\end{align}
 \eqref{eq:171202} for $s\in[0,1]$ follows; here we use the elemenatry bound
 \begin{align}\label{eq:pDis}
   \norm{(H-z)^{-1}p}\leq C
   (|\mathop{\mathrm{Im}}z|^{-1/2}+\inp{z}^{1/2}|\mathop{\mathrm{Im}}z|^{-1})\leq 2C\inp{z}/|\mathop{\mathrm{Im}}z|.
 \end{align}
 \end{subequations}

For $s\ge 1$ we write $s=s'+[s]$ in terms of the integer part $[s]$ of
$s$, and first use \eqref{eq:1comForm} with $s$
replaced by $s'$. Next we move each of the  $[s]$ factors of $r$
though the resolvents to the right using the formula
\eqref{eq:1comForm} repeatedly. This procedure yields an expansion of
$r^s(H-z)^{-1}r^{-s}$ 
into  a sum of terms having at most  $2+[s]$ factors of resolvents with at most
$1+[s]$ factors of $p$ distributed so that \eqref{eq:pDis} applies.

\smallskip
\noindent
\textit{(2)}\quad 
For $s\le 0$ the formula  \eqref{eq:fComR} follows from
Corollary~\ref{cor:A2}, and by the proof of (1) we see that indeed the right-hand side of \eqref{eq:fComR} is of order $s-1$.
For $s>0$ we apply \eqref{eq:fComR} to $r$ replaced by $r/(1+\epsilon
 r)$ for $\epsilon>0$. We let $\epsilon\to 0$ and obtain (2) in that case
 also.
\end{proof}

\subsection{Self-adjoint realization}\label{subsec:Operatat B}

Now we provide the self-adjoint realization of the  operator $B$.
Recall the definition \eqref{eq:17111317} of $\mathcal H^s$.

\begin{lemma}\label{lemma:Flow} 
Let $R\ge 1$ be sufficiently large. Then the operator $B$ defined as \eqref{eq:1710220}
is essentially self-adjoint on $C^\infty_{\mathrm c}(\Omega)$,
and the self-adjoint extension, denoted by $B$ again,
satisfies that for some $C>0$
\begin{align}
\mathcal D(B)\supset \mathcal H^1,\quad 
\|B\psi\|_{\mathcal H}\le C\|\psi\|_{\mathcal H^1}\ \ \text{for any }\psi\in\mathcal H^1.
\label{eq:171005}
\end{align}
In addition, $\mathrm e^{\mathrm itB}$ for each $t\in\mathbb R$
naturally restricts/extends
as bounded operators $\mathrm e^{\mathrm itB}\colon \mathcal H^{\pm1}\to\mathcal H^{\pm1}$, and 
they satisfy
\begin{align}
\sup_{t\in [-1,1]}\|\mathrm e^{\mathrm itB}\|_{\mathcal L(\mathcal H^{\pm 1})}<\infty,
\label{eq:17100722}
\end{align}
respectively. Moreover, the restriction $\mathrm e^{\mathrm itB}\in \mathcal L(\mathcal H^1)$ is 
strongly continuous in $t\in\mathbb R$.
\end{lemma}
\begin{remarks}\label{remark:self-adjo-realss}
\begin{enumerate}
\item
The same assertions except for \eqref{eq:171005} hold true also for the operator $A$ from \eqref{eq:1710220},
but we do not state it since we do not use it.
\item
For related results in more
general geometric
settings, see \cite[Lemma A.8]{IS1} and \cite[Lemma 2.8]{IS3}.
\item \label{item:self-adjo-real2} The condition $R\ge 1$ being large
  is only needed under Condition ~\ref{cond:smooth2wea3n1} (to
  guarantee that $\omega$ is complete). Under Condition
  \ref{cond:smooth2wea3n12} the completeness holds for any $R\ge
  1$. Moreover in that case we also have invariance of $\vH^{\pm 2}$
  and in fact
\begin{align}
\sup_{t\in [-1,1]}\|\mathrm e^{\mathrm itB}\|_{\mathcal L(\mathcal H^{\pm 2})}<\infty;
\label{eq:17100722b}
\end{align} the proof is similar.
\end{enumerate}
\end{remarks}
\begin{proof}
By the properties \eqref{eq:17100721} and \eqref{eq:boun_constr} 
we can find large $R\ge 1$ such that 
the rescaled vector field 
$\omega=\omega_R$ is complete on $\Omega$. 
Then there exists a globally defined flow
\begin{align}
 y\colon \mathbb R\times \Omega\to \Omega,\ \ 
(t,x)\mapsto y(t,x)=\exp(t\omega)(x),
\label{eq:14.12.10.6.32}
\end{align}
generated by $\omega$.
In other words, $y$ is a solution to the equation
\begin{align*}
\partial_ty(t,x)=\omega(y(t,x)),\quad
y(0,x)=x.
\end{align*}
We introduce the associated one-parameter group $\{U(t)\}_{t\in\mathbb R}$ of unitary operators
on ${\mathcal H}$ by 
\begin{align}
\begin{split}
(U(t)\psi)(x)
&=J(t,x)^{1/2}\psi(y(t,x))
\\&
=\exp \left(\int_0^t\tfrac12(\mathop{\mathrm{div}}\omega)(y(s,x))\,\mathrm{d}s\right)\psi(y(t,x)),
\end{split}
\label{eq:12.6.7.1.10c}
\end{align}
where $J(t,\cdot)$ is the Jacobian of the mapping 
$y(t,\cdot)\colon \Omega\to \Omega$.

Now we define $B$ as the generator of the group $\{U(t)\}_{t\in\mathbb R}$:
\begin{align*}
\mathcal D(B)&=\bigl\{\psi\in\mathcal H\,\big|\, 
\lim_{t\to 0}{}(\mathrm it)^{-1}(U(t)\psi-\psi)\text{\ \ exists in }\mathcal H\bigr\},\\
B\psi&=\lim_{t\to 0}{}(\mathrm it)^{-1}(U(t)\psi-\psi)\ \ 
\text{for }\psi\in\mathcal D(B).
\end{align*}
Since $U(t)$ is unitary, 
the generator $B$ is self-adjoint on $\mathcal H$. 
In addition, since $C^\infty_\c(\Omega)$ is preserved under $U(t)$, 
the space $C^\infty_\c(\Omega)$ is a core for $B$ by \cite[Theorem X.49]{RS}. 
By the last expression from \eqref{eq:12.6.7.1.10c}
the generator $B$ takes the form 
\begin{align}
B=\tfrac12(\omega\cdot p+p\cdot \omega)\ \ \text{on }C^\infty_\c(\Omega),
\label{eq:180116}
\end{align} 
which actually coincides with \eqref{eq:1710220}.
Then \eqref{eq:171005} follows by extension from 
the dense subspace $C^\infty_\c(\Omega)\subset \mathcal H^1$.

To prove \eqref{eq:17100722} it suffices to discuss the upper case by taking the adjoint.
For any $\psi\in C_\c^\infty(\Omega)\subset\vH^1$ consider the quantity  
$$f(t)=\inp{H_0+1}_{U(t)\psi}=\bigl\langle U(t)\psi,(H_0+1)U(t)\psi\bigr\rangle.$$ 
By using \eqref{eq:180116} we can compute and bound its derivative as 
\begin{align*}
\pm f'(t)=\pm \bigl\langle p\cdot h p-\tfrac14(\Delta^2 r)\bigr\rangle_{U(t)\psi}
\le C_1f(t),
\end{align*}
where $h=\mathop{\mathrm{Hess}}r$, and $C_1$ is a constant independent of 
$\psi\in C^\infty_{\mathrm c}(\Omega)$, cf.\ \eqref{eq:formulasyb} and \eqref{eq:formulasybb}
below.
Then by the Gronwall lemma and a  density argument 
we obtain the uniform boundedness of $U(t)=\mathrm e^{\mathrm itB}\colon \mathcal H^1\to\mathcal H^1$
for $t\in[-1,1]$.

Finally by a  density argument, \eqref{eq:17100722}
and the regularity of the flow \eqref{eq:14.12.10.6.32}
it is easy to see that $\e^{\i tB}\in \mathcal L(\mathcal H^1)$ is  strongly continuous 
in $t\in\mathbb R$.
\end{proof}

We remark that Lemma~\ref{lemma:Flow} has the following generalization
(proved in the same way).
We will use this generalized version when we compute a  second
commutator in the proof of Lemma~\ref{lemma:sing}.

\begin{lemma}\label{lemma:Flowb}
Let $v\in\mathfrak X(\Omega)$ be a smooth and complete
vector field on $\Omega$, and assume that 
there exists $C>0$ such that for any $x\in\bX$ 
\begin{align*}
|v(x)|\le C,\quad \abs{v'(x)}\le C,\quad 
\abs{\mathop{\mathrm{grad}}(\mathop{\mathrm{div}} v)(x)}\leq C.
\end{align*}
Then the differential operator 
$$B_v=\mathop{\mathrm{Re}}(v\cdot p)=\tfrac12(v\cdot p+p\cdot v)$$ 
is essentially self-adjoint on $C_\c^\infty(\Omega)$, 
and the self-adjoint extension, denoted by $B_v$ again, 
satisfies for some $C'>0$
\begin{align*}
\mathcal D(B_v)\supset\mathcal H^1,\quad 
\|B_v\psi\|\le C'\|\psi\|_{\mathcal H^1}\ \ \text{for any }\psi\in\mathcal H^1.
\end{align*}
In addition  the operators $\mathrm e^{\mathrm itB_v}$, $t\in\mathbb
R$,  naturally restrict/extend
to bounded operators $\mathrm e^{\mathrm itB_v}\colon \mathcal H^{\pm1}\to\mathcal H^{\pm1}$, and 
they satisfy
\begin{align*}
\sup_{t\in [-1,1]}\|\mathrm e^{\mathrm itB_v}\|_{\mathcal L(\mathcal H^{\pm 1})}<\infty,
\end{align*}
respectively.
Moreover, the restriction $\mathrm e^{\mathrm itB_v}\in \mathcal L(\mathcal H^1)$ is 
strongly continuous in $t\in\mathbb R$.
\end{lemma}

Finally in this subsection we present some basic  properties of $B$ (proved in the same manner as we proved  Lemma~\ref{lem:171113}).

\begin{lemma}\label{lem:171113b}
{Let} $R\ge 1$ be sufficiently large.
\begin{enumerate}
\item For any  $F\in \mathcal F^t$ with $t<0$ the operator   $F(B)$ is of order $0$.
\item Let $F\in \mathcal F^t$ with $t<1$, and let $s\in\mathbb R$.
Then $\mathrm i[F(B),r^s]$ 
has an expression, 
as a sesquilinear  form on $\mathcal D(F(B))\cap L^2_{\max\{0,s\}}$,
\begin{align*}
\mathrm i[F(B),r^s] 
= -s
\int _{\C}
(B-z)^{-1}\bigl(\omega^2r^{s-1}\bigr) (B-z)^{-1}\,\mathrm d\mu_F(z).
\end{align*}  

In particular  $\mathrm i[F(B),r^s]$ is of order $s-1$.
\end{enumerate}
\end{lemma}

\subsection{First commutator}\label{subsec:17111117}
Here we are going to compute the commutator $\mathrm i[H,B]$,
and bound it below.
To be rigorous about (form) domains we define 
$\mathrm i[H,B]$ first as a (bounded) quadratic form on $\mathcal H^2$:
\begin{align}
\langle \mathrm i[H,B]\rangle_\psi
=2\langle \mathop{\mathrm{Im}}(BH)\rangle_\psi
=\mathrm i\langle H\psi,B\psi\rangle-\mathrm i\langle B\psi,H\psi\rangle
\ \ \text{for }\psi\in \mathcal H^2.
\label{eq:17100614}
\end{align}
Let us set 
$$
\tilde\omega=\tfrac12\mathop{\mathrm{grad}} r^2,\quad 
\tilde h=\tfrac12\mathop{\mathrm{Hess}}r^2,\quad
\omega=\mathop{\mathrm{grad}} r,\quad 
h=\mathop{\mathrm{Hess}}r.
$$
Then \textit{formal} computations would suggest that 
\begin{align*}
A=\mathrm i[H,r^2],\quad
B=\mathrm i[H,r] ,\quad
A=r^{1/2}Br^{1/2},
\end{align*}
and hence that 
\begin{align}\label{eq:formulasyb}
\i [H,A]&=p\cdot \tilde h p-\tfrac18\parb{\Delta^2r^2}-\tilde \omega \cdot (\nabla V),\\
\mathrm i[H,B]&=r^{-1/2}\bigr(\mathrm i[H,A] -B^2\bigr)r^{-1/2} +\tfrac14r^{-2} \omega\cdot h\omega.
\label{eq:formulasybb}
\end{align} 
Thus we {could expect that $\mathrm i[H,B]$  extends  continuously onto larger spaces,
and this is partly justified in the following lemma, which provides a
formula for the commutator. 
We note that in the case where Condition~\ref{cond:smooth2wea3n1} is adopted
direct  integration  by parts in \eqref{eq:17100614}
would possibly need additional smoothness of  $\psi\in\mathcal H^2$
and the hard-core boundaries, and in 
general it is not obvious how  to consider $\mathrm i[H,B]$ 
 as the   unique extension of the corresponding form on $C^\infty_{\mathrm c}(\Omega)$.
Instead, we  shall compute $\mathrm i[H,B] \in \mathcal L(\mathcal H^2,\mathcal H^{-2})$ by combining  the
realization \eqref{eq:limit00} below with a density argument, 
 cf.\ \eqref{eq:171202b}.}

\begin{lemma}\label{lem:17100620}
Suppose 
{Condition~\ref{cond:smooth2wea3n1} or Condition \ref{cond:smooth2wea3n12}}, 
 and let $R\ge 1$ be sufficiently large.
Denoting  the extension of the quadratic form ${\bf D}B:=\mathrm
i[H,B]$ given in   \eqref{eq:17100614}  by the same notation, 
it has expressions
\begin{align}
{\bf D}B
&=\slim_{t\to 0} t^{-1}\parb{H\e^{\i tB}-\e^{\i tB}H}\ \ 
\text{in }
{\mathcal L(\mathcal H^2,\mathcal H^{-1})
\cap \mathcal L(\mathcal H^1,\mathcal H^{-2})}
\label{eq:limit00b}
\end{align}
and, more explicitly, 
\begin{align}\label{eq:formulasy}
\begin{split}
{\bf D}B
=r^{-1/2}\parb{L -B^2}r^{-1/2}
, 
\end{split}
\end{align} 
where $L\in \mathcal L(\mathcal H^2,\mathcal H^{-1})
\cap \mathcal L(\mathcal H^1,\mathcal H^{-2})$ 
is defined as 
\begin{align}\label{eq:mourre comm}
\begin{split}
L&=p\cdot \tilde h p
-\tfrac18(\Delta^2 r^2)
+\tfrac14r^{-1} \omega\cdot h\omega
\\&\phantom{{}={}}{}
+\sum_{a\in\mathcal A}
\Bigl(-\tilde\omega^a\cdot \bigl(\nabla^a V_a^{\mathrm{lr}}\bigr)
+\bigl(V_a^{\mathrm{sr}}\tilde\omega^a\bigr)\cdot \nabla^a
-\nabla^a\cdot
\bigl(V_a^{\mathrm{sr}}\tilde\omega^a\bigr)
+V_a^{\mathrm{sr}}\mathop{\mathrm{div}} \tilde\omega^a
\Bigr).
\end{split}
\end{align}
Here $\tilde\omega^a$ and $\nabla^a$ denote  the projection 
onto the internal components of $\tilde\omega$ and $\nabla$,
respectively, for any  cluster decomposition $a\in\mathcal A$.
\end{lemma}
{\begin{remarks*}
\begin{enumerate}
\item
Under Condition~\ref{cond:smooth2wea3n1} the quantities 
$\mathbf DB$ and $L$ actually belong to 
$\mathcal L(\mathcal H^1,\mathcal H^{-1})$, and  in this case the
limit \eqref{eq:limit00b} may be taken in $\mathcal L(\mathcal
H^1,\mathcal H^{-1})$, cf.  \eqref{eq:1710061628}.
\item
If we consider $\mathbf DA=\mathrm i[H,A]$ in some extended sense, we can also write 
\begin{align}
L=\mathbf DA+\tfrac14r^{-1} \omega\cdot h\omega,
\label{eq:mourre commb}
\end{align}
cf.\ \eqref{eq:formulasyb} and \eqref{eq:mourre comm}. 
However, since we will not `undo' the commutator $\mathrm i[H,A]$  
or in other ways use the operator $A$ itself, 
we have suppressed $A$ from \eqref{eq:mourre comm}. 
We emphasize that our theory does not depend on $A$ but on $B$. 
\end{enumerate}
\end{remarks*}}
\begin{proof} 
 We may consider $\mathrm i[H,B]\in \mathcal L(\mathcal H^2,\mathcal
H^{-2})$  since $\mathcal H^1\subset \mathcal D(B)$.
By Lemma~\ref{lemma:Flow}  
it follows that 
\begin{align}\label{eq:limit00}
\begin{split}
\mathrm i[H,B]
=\slim_{t\to 0} t^{-1}\parb{H\e^{\i tB}-\e^{\i tB}H}\ \ 
\text{in }\mathcal L(\mathcal H^2,\mathcal H^{-2}).
\end{split}
\end{align}   Let us write 
$$H=H_0+V^{\mathrm{lr}}+V^{\mathrm{sr}}+V^{\mathrm{si}};\quad 
V^*=\sum_{a\in\mathcal A}V^*_a\ \ \text{for }*=\mathrm{lr},\mathrm{sr},\mathrm{si}.
$$

We first consider Condition~\ref{cond:smooth2wea3n1}.
 Then by Lemma~\ref{lemma:Flow} in fact  
$H\e^{\i tB}-\e^{\i tB}H\in \mathcal L(\mathcal H^1,\mathcal H^{-1})$
for each $t\in \mathbb R$.   
For sufficiently large $R\ge 1$ we compute as a quadratic form on $C^\infty_{\mathrm c}(\Omega)\subset\mathcal H^1$ 
\begin{align*}
H\e^{\i tB}-\e^{\i tB}H
&=\int_0^t\frac{\mathrm d}{\mathrm ds}\mathrm e^{\mathrm i(t-s)B}H \mathrm e^{\mathrm i sB}\,\mathrm ds
\\&
=\int_0^t\mathrm e^{\mathrm i(t-s)B}\bigl(\mathrm i\bigl[H_0+V^{\mathrm{lr}}+V^{\mathrm{sr}},B\bigr]\bigr)
\mathrm e^{\mathrm i sB}\,\mathrm ds
\\&
=\int_0^t\mathrm e^{\mathrm i(t-s)B}r^{-1/2}\bigl(L-B^2\bigr)r^{-1/2}
\mathrm e^{\mathrm i sB}\,\mathrm ds.
\end{align*} 
Hence for each $t\in\mathbb R\setminus \{0\}$ we obtain 
\begin{align}
t^{-1}\bigl(H\e^{\i tB}-\e^{\i tB}H\bigr)
&=
t^{-1}
\int_0^t\mathrm e^{\mathrm i(t-s)B}r^{-1/2}\bigl(L-B^2\bigr)r^{-1/2}\mathrm e^{\mathrm i sB}\,\mathrm ds 
\label{eq:1710061626}
\end{align}
as a quadratic form on $C^\infty_{\mathrm c}(\Omega)$, 
but  both sides of \eqref{eq:1710061626} obviously extend continuously
to $\mathcal H^1$. 
 The right-hand side of \eqref{eq:1710061626}
has a strong limit $r^{-1/2}\bigl(L-B^2\bigr)r^{-1/2}$  for  $t\to 0$,
and consequently 
it follows that 
\begin{align}
\slim_{t\to 0} t^{-1}\parb{H\e^{\i tB}-\e^{\i tB}H}=r^{-1/2}\bigl(L-B^2\bigr)r^{-1/2}
\ \ \text{in }\mathcal L(\mathcal H^1,\mathcal H^{-1}).
\label{eq:1710061628}
\end{align}
Hence the lemma is proven {under Condition~\ref{cond:smooth2wea3n1}}. 

If Condition~\ref{cond:smooth2wea3n12} holds we  compute  as a
  quadratic form on $C^\infty_{\mathrm c}(\bX)$
  \begin{align*}
    \mathrm i\bigl[H_0+V^{\mathrm{lr}}+V^{\mathrm{sr}},B\bigr]=r^{-1/2}\bigl(L-B^2\bigr)r^{-1/2},
  \end{align*} leading to \eqref{eq:1710061626}
 as a quadratic form on $C^\infty_{\mathrm c}(\bX)$. Due to Remark
 \ref {remark:self-adjo-realss}~\eqref{item:self-adjo-real2}
we can extend \eqref{eq:1710061626}  (uniquely) to $\vH^2$. The right-hand side
has the  strong limit $r^{-1/2}\bigl(L-B^2\bigr)r^{-1/2}$  in ${\mathcal L(\mathcal H^2,\mathcal H^{-1})
\cap \mathcal L(\mathcal H^1,\mathcal H^{-2})}$. The proof of the
lemma is complete under  Condition~\ref{cond:smooth2wea3n12} also.

\end{proof}
\begin{remark*}
The completeness of the vector field $\omega$
comes in handy in giving an interpretation of the formal commutator $\mathrm i[H,B]$. 
If $\omega$ is incomplete in $\Omega$, 
we can not freely `do and undo' the commutator $\mathrm i[H,B]$ due to 
the boundary contribution coming from integration  by parts.
See \cite [Proposition 6.2]{BGS} for an explicit
 formula for this boundary contribution under regularity conditions on the boundary of $\Omega$.
We also note that for $N=1$ only the \textit{forward} completeness suffices, see \cite{IS3}.
\end{remark*}

We quote the following result without proof.
It is the so-called \emph{Mourre estimate} for the $N$-body Schr\"odinger operator. 
Here again we suppress the usual conjugate operator $A$,
since in this paper it suffices to have this estimate for 
the quadratic form $L$ defined by  \eqref{eq:mourre comm};
the expression in terms of $A$ is not needed.

\begin{lemma}\label{lemma:Mourre1_hard} 
Suppose 
{Condition \ref{cond:smooth2wea3n1} or Condition \ref{cond:smooth2wea3n12}}. 
Let $I\subset\R\setminus \vT(H)$ be a compact interval, let $\epsilon>0$, 
and take $R\ge 1$ large enough.
Then for any $\lambda\in I$
there exist a
neighbourhood $\vU$ of $\lambda$ and a compact operator $K$ on
$\vH$ such that for all  real-valued $f\in C^\infty_{\c}(\vU)$
\begin{equation*}
f(H)^*Lf(H)\geq f(H)^*\bigl(2d(\lambda)-\epsilon
-K\bigr)f(H).
\end{equation*}
\end{lemma} 
\begin{proof}
We  proved this version of the Mourre estimate in \cite{IS1} for a
class of more
regular pair-potentials 
using the properties of the rescaled Graf function stated in
Subsection \ref{subsec:resc-graf-funct}.  
Note that although there is an extra term $\tfrac14r^{-1} \omega\cdot h\omega$
in  \eqref{eq:mourre commb}, it is obviously harmless.
We omit the details. 
\end{proof}

We will always implement Lemma~\ref{lemma:Mourre1_hard}
in combination with Lemma~\ref{lem:17100620} in the following form.

\begin{corollary}\label{cor:171007}
Suppose 
{Conditions~\ref{cond:smooth2wea3n1} or \ref{cond:smooth2wea3n12}}. 
Let $\lambda\in\mathbb R\setminus\mathcal T(H)$ and $\sigma\in (0,\gamma(\lambda))$,
and take $R\ge 1$ large enough
and a neighborhood $\mathcal U\subset\mathbb R$ of $\lambda$ small enough.
Then for any real-valued function 
$f\in C^\infty_{\mathrm c}(\mathcal U)$ there exists $C>0$ 
such that, as quadratic forms on $\mathcal H$,
 \begin{align*}
 \begin{split}
 f(H)(\mathbf DB)f(H)
\ge 
f(H)r^{-1/2}\bigl(\sigma^2-B^2\bigr)r^{-1/2}f(H)
-Cr^{-2}. 
 \end{split}
 \end{align*} 
\end{corollary}
\begin{proof}
In order to apply Lemma~\ref{lemma:Mourre1_hard} we fix variables in the following order: 
First fix any $\lambda\in\mathbb R\setminus\mathcal T(H)$ and $\sigma\in (0,\gamma(\lambda))$ 
as in the assertion.
We then let $I=\set{\lambda}$ and take  any $\epsilon\in (0,2d(\lambda)-\sigma^2)$. 
Wtih  these quantities $I$ and $\epsilon$ fixed we consider
 any large $R\ge 1$ in agreement with 
both Lemmas~\ref{lem:17100620} and \ref{lemma:Mourre1_hard}. 
We fix a neighbourhood $\mathcal U$ of $\lambda$ 
and a compact operator $K$ on $\mathcal H$ as in Lemma~\ref{lemma:Mourre1_hard},
and let $f\in C^\infty_{\mathrm c}(\mathcal U)$ be any real-valued
function. We need to show the quadratic form bound. 

First,  by Lemma~\ref{lem:17100620} we  have 
\begin{align}
\begin{split}
f(H)(\mathbf DB)f(H)
&=f(H)r^{-1/2}\parb{L -B^2}r^{-1/2}f(H)
\\&\ge 
r^{-1/2}f(H)Lf(H)r^{-1/2}
-f(H)r^{-1/2}B^2r^{-1/2}f(H)
\\&\phantom{{}={}}{}
+\bigl[f(H),r^{-1/2}\bigr]L\bigl[r^{-1/2},f(H)\bigr]
\\&\phantom{{}={}}{}
+2\mathop{\mathrm{Re}}\bigl(r^{-1/2}f(H)L\bigl[r^{-1/2},f(H)\bigr]\bigr)
.
\end{split}\label{eq:1710083}
\end{align}
By Lemma~\ref{lemma:Mourre1_hard} we can bound the first term 
on the right-hand side of \eqref{eq:1710083} as 
\begin{align}
r^{-1/2}f(H)Lf(H)r^{-1/2}
\ge 
r^{-1/2}f(H)\bigl(2d(\lambda)-\epsilon-K\bigr)f(H)r^{-1/2}.
\label{eq:171008341}
\end{align}
Since $K$ is compact on $\mathcal H$, we can choose $m\in\mathbb N_0$ large enough that
\begin{align}\label{eq:Kbnd}
2d(\lambda)-\epsilon-\|K-\chi_{m}K\chi_{m}\|\geq \sigma^2, 
\end{align}
where $\chi_m$ is the smooth cut-off function from \eqref{eq:1711021}.
The bounds \eqref{eq:171008341} and \eqref{eq:Kbnd} imply that 
\begin{align}
\begin{split}
r^{-1/2}f(H)Lf(H)r^{-1/2}
&\ge 
\sigma^2 r^{-1/2}f(H)^2r^{-1/2}
-r^{-1/2}f(H)\chi_{m}K\chi_{m}f(H)r^{-1/2}
\\&
=
\sigma^2 f(H)r^{-1}f(H)
-\sigma^2 \bigl[r^{-1/2},f(H)\bigr]\bigl[f(H),r^{-1/2}\bigr]
\\&\phantom{{}={}}{}
+2\sigma^2 \mathop{\mathrm{Re}}\bigl(r^{-1/2}f(H)\bigl[f(H),r^{-1/2}\bigr]\bigr)
\\&\phantom{{}={}}{}
-r^{-1/2}f(H)\chi_{m}K\chi_{m}f(H)r^{-1/2}
.
\end{split} 
\label{eq:171008350} 
\end{align} 
By \eqref{eq:1710083}, \eqref{eq:171008350} it remains to bound the
third and fourth terms of \eqref{eq:1710083} and the second to fourth
terms of \eqref{eq:171008350}, but all of them can be treated by using
Lemma~\ref{lem:171113}. In fact one easily checks using the explicit
representations \eqref{eq:fComR} and \eqref{eq:mourre comm} that all
of these terms are of order $-2$ in the sense of \eqref{eq:1712022}.
\end{proof}

Finally we compute and bound commutators of functions of $H$ and $B$.

\begin{lemma}\label{lem:17111416}
Suppose 
 {Condition~\ref{cond:smooth2wea3n1} or Condition \ref{cond:smooth2wea3n12}},
and let $R\ge 1$ be sufficiently large. 
For any $f\in \mathcal F^t$ and $F\in \mathcal F^{t'}$ with $t<-1/2$ and $t'<1$ 
the commutators $\mathrm i[f(H),B]$, $\mathrm i[f(H),F(B)]$ extend to  bounded sesquilinear  forms 
on $\mathcal H$ from $\mathcal D(B)$, $\mathcal D(F(B))$, 
and they have expressions
\begin{align*}
\mathrm i[f(H),B]
&=
-\int _{\C}(H-z)^{-1}\bigl(\mathbf DB\bigr) (H-z)^{-1}\,\mathrm d\mu_f(z)
,\\
\mathrm i[f(H),F(B)]
&=
-\int _{\mathbb C}(B-z)^{-1}\bigl(\mathrm i[f(H),B]\bigr)(B-z)^{-1}\,\mathrm d\mu_F(z),
\end{align*}
respectively. 
Moreover, with  the notation   \eqref{eq:1712022}
\begin{align*}
\mathrm i[f(H),B]=O(r^{-1})
,\quad 
\mathrm i[f(H),F(B)]=O(r^{-1})
.
\end{align*}
\end{lemma}
\begin{proof}
By \eqref{eq:limit00b} and interpolation
$\i[H,B]\in\vL(\vH^{3/2},\vH^{-3/2})$. This yields the first formula
 by Corollary~\ref{cor:A2}. The second follows from the first and Corollary~\ref{cor:A2}.
 Using the expression for $\bD B$ from Lemma~\ref{lem:17100620} one
 easily checks the last assertions, cf.   the proofs of 
Lemmas~\ref{lem:171113} and \ref{lem:171113b}. 
\end{proof}

\subsection{Second commutator}\label{subsec:17111119}

Here we provide a realization of the 
second commutator $\mathrm i[\mathbf DB,B]$, 
and bound it in some operator space. 
Under {Condition~\ref{cond:smooth2wea3n12}
one can naturally define and interpret this second commutator. 
On the other hand, under Condition~\ref{cond:smooth2wea3n1}} there is a `domain problem',
and this prevents us from directly defining $\mathrm i[\mathbf DB,B]$ as a quadratic form even on $\mathcal H^2$, 
like we first did for $\mathbf DB=\mathrm i[H,B]$.
However, for our application it suffices to consider 
an alternative strong limit of the form \eqref{eq:17100813} below, cf.\ \eqref{eq:171202b}
and \eqref{eq:limit00b}.
We note that only the first commutator does not suffice 
for the LAP bounds either in the abstract Mourre theory, see \cite{ABG2}.
 
Note that by Lemmas~\ref{lemma:Flow} and \ref{lem:17100620} we may consider 
$${(\mathbf DB)\e^{\i tB}-\e^{\i tB}(\mathbf DB)\in 
\mathcal L(\mathcal H^2,\mathcal H^{-2}).}$$

\begin{lemma}\label{lemma:sing}
Suppose 
{Condition~\ref{cond:smooth2wea3n1} or Condition \ref{cond:smooth2wea3n12}}, 
and let $R\ge 1$ be sufficiently large. 
Then there exists the strong limit 
\begin{align} 
\mathrm i[\mathbf DB,B]&:=\slim_{t\to 0}t^{-1}
\bigl((\mathbf DB)\mathrm e^{\mathrm i tB}-\e^{\mathrm i tB}(\mathbf DB)\bigr)
\ \ \text{in }\mathcal L(\mathcal H^2,\mathcal H^{-2}).
\label{eq:17100813}
\end{align} 
Moreover, with the notation  \eqref{eq:1712022}
\begin{align}
(H-\mathrm i)^{-1}\bigl(\mathrm i[\mathbf DB,B]\bigr)(H+\mathrm i)^{-1}
=O(r^{-1-2\kappa}),
\label{eq:1710081345}
\end{align}
where $\kappa=\delta/(1+2\delta)$ as in  \eqref{eq:171114}.
\end{lemma}
\begin{proof} 
First we consider Condition~\ref{cond:smooth2wea3n1}.
By  \eqref{eq:formulasy} and \eqref{eq:mourre comm} 
\begin{align}
  \begin{split}
\mathbf DB&=\mathcal L
+\sum_{a\in\mathcal A}\mathcal V_a;\\
\mathcal L
&=r^{-1/2}\bigl(
p\cdot \tilde h p
-\tfrac18(\Delta^2 r^2)
+\tfrac14r^{-1} \omega\cdot h\omega
-B^2\bigr)r^{-1/2},
\\
\mathcal V_a
&=-\omega^a\cdot \bigl(\nabla^a V_a^{\mathrm{lr}}\bigr)
+V_a^{\mathrm{sr}}\mathop{\mathrm{div}}\omega^a -2\mathop{\mathrm{Im}}\bigl((V_a^{\mathrm{sr}}\omega^a)\cdot p^a\bigr). 
  \end{split}
\label{eq:17112}
\end{align}
We consider the contribution from each of these terms.

The contribution from $\mathcal L$ is straightforward.
Similarly to the proof of Lemma~\ref{lem:17100620}
we can prove the existence of the strong limit 
\begin{align*}
\mathrm i[\mathcal L,B]&:=\slim_{t\to 0}t^{-1}
\bigl(\mathcal L\e^{\i tB}-\e^{\i tB}\mathcal L\bigr)
\ \ \text{in }\mathcal L(\mathcal H^1,\mathcal H^{-1})
\end{align*}
by extension from $C^\infty_{\mathrm c}(\Omega)$, cf.\ \eqref{eq:1710061626} and \eqref{eq:1710061628}. 
In fact the expression $\vL$ simplifies as $\mathcal L=p\cdot h
p-\tfrac14(\Delta^2 r)$, cf. the familiar formula (for any $v\in\mathfrak X(\Omega)$)
\begin{align}\label{eq:comF}
  \i[p^2, v\cdot p+ p\cdot v]=2p(v'+v'^t)p-(\Delta \mathop{\mathrm{div}}v).
\end{align}
Using this representation we compute 
\begin{align*}
\mathrm i[\mathcal L,B]
&=
p\cdot \Bigl(2 h^2 -\bigl((\omega\cdot \nabla) h\bigr) \Bigr)p
-\tfrac12(\nabla\cdot h \nabla\Delta r)
+\tfrac14\bigl(\omega\cdot \nabla(\Delta^2 r)\bigr).
\end{align*}
 The  contribution from the right-hand side   is easily checked 
 using  \eqref{eq:boun_constr} to agree with  \eqref{eq:1710081345} (the contribution  is an operator
 of order $-2$).

Next  we consider the contribution from $\mathcal V_a$. 
Since we can write 
\begin{align}
\begin{split}
&(H-\mathrm i)^{-1}\bigl(\mathcal V_a\e^{\i tB}-\e^{\i tB}\mathcal V_a\bigr)(H+\mathrm i)^{-1}
\\&
=
(H-\mathrm i)^{-1}\mathcal V_a(H+\mathrm i)^{-1}\e^{\i tB}
-\e^{\i tB}(H-\mathrm i)^{-1}\mathcal V_a(H+\mathrm i)^{-1}
\\&\phantom{{}={}}{}
+(H-\mathrm i)^{-1}\mathcal V_a(H+\mathrm i)^{-1}\bigl(H\e^{\i tB}-\e^{\i tB}H\bigr)(H+\mathrm i)^{-1}
\\&\phantom{{}={}}{}
+(H-\mathrm i)^{-1}\bigl(H\e^{\i tB}-\e^{\i tB}H\bigr)(H-\mathrm i)^{-1}\mathcal V_a(H+\mathrm i)^{-1},
\end{split}
\label{eq:17110217}
\end{align}
there exists the strong limit 
\begin{align*}
\mathrm i[\mathcal V_a,B]
:=\slim_{t\to 0}
t^{-1}\bigl(\mathcal V_a\e^{\i tB}-\e^{\i tB}\mathcal V_a\bigr)
\ \ \text{in }
\mathcal L(\mathcal H^2,\mathcal H^{-2}),
\end{align*}
which has the  expression, with appropriate weights from both sides,
\begin{align}
\begin{split}
&
(H-\mathrm i)^{-1}\mathrm i[\mathcal V_a,B](H+\mathrm i)^{-1}
\\&
=
-2\mathop{\mathrm{Im}}\bigl((H-\mathrm i)^{-1}\mathcal V_a(H+\mathrm i)^{-1}B\bigr)
\\&\phantom{{}={}}{}
+2\mathop{\mathrm{Re}}\bigl((H-\mathrm i)^{-1}\mathcal V_a(H+\mathrm i)^{-1}(\mathbf DB)(H+\mathrm i)^{-1}\bigr)
.
\end{split}
\label{eq:1710081346}
\end{align}
Using  the expression \eqref{eq:17112}
it follows that the last term on the right-hand side of \eqref{eq:1710081346} agrees
with \eqref{eq:1710081345}, so it only remains 
 to examine the first term.
We set 
\begin{align*}
\tilde \eta_b(x)=\eta_b\bigl(x/r^{1/(1+2\delta)}\bigr)=\eta_{1,b}\bigl(x/(Rr^{1/(1+2\delta)})\bigr).
\end{align*}
and decompose 
\begin{align*}
B=\sum_{b\in\mathcal A}\tilde B_b;\quad 
\tilde B_b=\tfrac12\bigl((\tilde\eta_b \omega)\cdot p+p\cdot (\tilde\eta_b \omega)\bigr).
\end{align*}
Then the investigation of the first term of \eqref{eq:1710081346} reduces to that of 
\begin{align}
\begin{split}
2\mathop{\mathrm{Im}}\bigl((H-\mathrm i)^{-1}\mathcal V_a(H+\mathrm i)^{-1}\tilde B_b\bigr)
\ \ \text{for }b\in\mathcal A.
\end{split}
\label{17102010}
\end{align}

We first consider the case  $a\not\subset b$. 
Then by \eqref{eq:171018} we have 
\begin{align*}
|x^a|\ge c_1Rr^{1/(1+2\delta)}\ \ \text{on }\mathop{\mathrm{supp}}\tilde\eta_b.
\end{align*} 
This combined with \eqref{eq:2k2we3n} implies that for  
 $a\not\subset b$ 
the contribution  \eqref{17102010} agrees with \eqref{eq:1710081345}.

It remains to consider the case   $a\subset b$.
 We further decompose 
\begin{align*}
\tilde B_b&=(\tilde B_b)^b+(\tilde B_b)_b
\end{align*}
with
\begin{align*}
(\tilde B_b)^b&=\tfrac12\bigl((\tilde\eta_b \omega^b)\cdot p^b+p^b\cdot (\tilde\eta_b \omega^b)\bigr),
\quad 
(\tilde B_b)_b=\tfrac12\bigl((\tilde\eta_b \omega_b)\cdot p_b+p_b\cdot (\tilde\eta_b \omega_b)\bigr).
\end{align*}
 Accordingly  \eqref{17102010} decomposes as 
\begin{align}
\begin{split}
2\mathop{\mathrm{Im}}\bigl((H-\mathrm i)^{-1}\mathcal V_a(H+\mathrm i)^{-1}\tilde B_b\bigr)
&=
2\mathop{\mathrm{Im}}\bigl((H-\mathrm i)^{-1}\mathcal V_a(H+\mathrm i)^{-1}(\tilde B_b)^b\bigr)
\\&\phantom{{}={}}{}
+
2\mathop{\mathrm{Im}}\bigl((H-\mathrm i)^{-1}\mathcal V_a(H+\mathrm i)^{-1}(\tilde B_b)_b\bigr)
.
\end{split}
\label{eq:17102123}
\end{align}
To bound the first term of \eqref{eq:17102123} note that \eqref{eq:171018} implies that 
\begin{align*}
|x^b|\le C_1Rr^{1/(1+2\delta)}\ \ \text{on }\mathop{\mathrm{supp}}\tilde\eta_b,
\end{align*}
so that, combined with \eqref{eq:boun_constr}, 
\begin{align}
|\omega^b|\le C_2Rr^{-2\delta/(1+2\delta)}\ \ \text{on }\mathop{\mathrm{supp}}\tilde\eta_b.
\label{eq:1711021813}
\end{align}
Hence the first term of \eqref{eq:17102123} certainly agrees with \eqref{eq:1710081345}.
As for the second term of \eqref{eq:17102123}, 
we first note that the vector field $v_b=\tilde\eta_b \omega_b\in\mathfrak X(\Omega)$ is complete on $\Omega$. 
To see this  it suffices (since it  is bounded) to show that $v_b$ is
tangent to $\partial(\Omega_c+\mathbf X_c)$ for all $c\in\mathcal A$:
If $c\not\subset b$, we have 
$$|x^c|\ge c_2Rr^{1/(1+2\delta)}\ \ \text{on }\mathop{\mathrm{supp}}\tilde\eta_b,$$
and hence by boundedness of $\partial\Omega_c\subset \mathbf X^c$
$$v_b=0\ \ \text{on }\partial(\Omega_c+\mathbf X_c)$$ 
for large $R\ge 1$.
If $c\subset b$ 
$$v_b\in \mathbf X_b\subset \mathbf X_c,$$
also    implying that  $v_b$ is tangent to $\partial(\Omega_c+\mathbf X_c)$.
Now  $v=v_b$ is complete on $\Omega$, 
and indeed it  satisfies the assumptions of Lemma~\ref{lemma:Flowb}. 
By using the lemma  we can  move   the operator $(\tilde B_b)_b$
to the center in the second term of \eqref{eq:17102123}.
We calculate similarly to \eqref{eq:17110217} 
\begin{align}
\begin{split}
&2\mathop{\mathrm{Im}}
\bigl((H-\mathrm i)^{-1}\mathcal V_a(H+\mathrm i)^{-1}(\tilde B_b)_b\bigr)
\\&
=\slim_{t\to 0}t^{-1}
(H-\mathrm i)^{-1}\Bigl(
-\bigl(
\mathcal V_a\mathrm e^{\mathrm it(\tilde B_b)_b}
-
\mathrm e^{\mathrm it(\tilde B_b)_b}\mathcal V_a
\bigr)
\\&\phantom{{}={}\slim_{t\to 0}t^{-1}(H-\mathrm i)^{-1}\bigl[}{}
+\bigl(H\mathrm e^{\mathrm it(\tilde B_b)_b}-\mathrm e^{\mathrm it(\tilde B_b)_b}H\bigr)
(H-\mathrm i)^{-1}\mathcal V_a
\\&\phantom{{}={}\slim_{t\to 0}t^{-1}(H-\mathrm i)^{-1}\bigl[}{}
+\mathcal V_a(H+\mathrm i)^{-1}
\bigl(H\mathrm e^{\mathrm it(\tilde B_b)_b}-\mathrm e^{\mathrm it(\tilde B_b)_b}H\bigr)
\Bigr)(H+\mathrm i)^{-1}
.
\end{split}
\label{eq:17110218}
\end{align}
Noting that $(\nabla_bV_a^{\mathrm{sr}})=0$ due to the property $a\subset b$, 
we can mimic   the proof of Lemma~\ref{lem:17100620} and calculate
 the strong limit 
\begin{align*}
\mathrm i\bigl[\mathcal V_a,(\tilde B_b)_b\bigr]
&:=\slim_{t\to 0}t^{-1}
\bigl(\mathcal V_a\mathrm e^{\mathrm it(\tilde B_b)_b}
-\mathrm e^{\mathrm it(\tilde B_b)_b}\mathcal V_a\bigr)
\ \ \text{in }\mathcal L(\mathcal H^1,\mathcal H^{-1})
\end{align*}
as the  explicit  expression (with $v_b=\tilde{\eta}_b\omega_b$)
\begin{align*}
\begin{split}
\mathrm i\bigl[&\mathcal V_a,(\tilde B_b)_b\bigr]
=\bigl((v_b\cdot \nabla_b)\omega^a\bigr)\cdot \bigl(\nabla^a V_a^{\mathrm{lr}}\bigr)
-\bigl(v_b\cdot \nabla_b
\mathop{\mathrm{div}}\omega^a\bigr)V_a^{\mathrm{sr}}\\
&+2\mathop{\mathrm{Im}}\bigl(\bigl(
V_a^{\mathrm{sr}}(v_b\cdot \nabla_b)\omega^a\bigr)\cdot
p^a\bigr)-2\mathop{\mathrm{Im}}\bigl(\bigl((
V_a^{\mathrm{sr}}\omega^a\cdot \nabla)v_b\bigr)\cdot
p_b\bigr)+( V_a^{\mathrm{sr}}\omega^a\cdot
\nabla)\mathop{\mathrm{div}}v_b\\&=:T_1+\cdots+T_5.
\end{split}
\end{align*}
To see that the contribution from the first term in the big  brackets of \eqref{eq:17110218}
 agrees
with   \eqref{eq:1710081345} we plug in $T_1+\cdots+T_5$ and bound separately. To treat
$T_1, T_2$ and $T_3$ we use 
\begin{align}
\omega_b\cdot \nabla_b=r^{-1}x\cdot\nabla+r^{-1}(\tilde\omega-x)\cdot\nabla-\omega^b\cdot\nabla^b,
\label{eq:1711021812}
\end{align}
\eqref{eq:boun_constr} and \eqref{eq:1711021813} to obtain the
$r^{-1-2\kappa}$ decay. The terms $ T_4$ and $ T_5$ contribute by
terms with this
decay too thanks to the bounds 
\begin{align}\label{eq:bndvbeta}
  \partial
^\beta v_b=O(r^{-1/(1+2\delta)})=O(r^{-2\kappa
 });\quad |\beta|\geq 1.
\end{align}

We can compute the expressions 
for the second and third terms in the big  brackets of
\eqref{eq:17110218}  as in 
\eqref{eq:17110217} and \eqref{eq:1710081346}. Using Lemma
\ref{lemma:Flowb} with $v=\tilde{\eta}_b\omega_b$, \eqref{eq:comF},
 \eqref{eq:bndvbeta} and the
proof of Lemma~\ref{lem:17100620} we easily check that these 
 contributions also agree with  \eqref{eq:1710081345}.
Hence we are done with the proof under 
  Condition~\ref{cond:smooth2wea3n1}.

The proof under  Condition~\ref{cond:smooth2wea3n12} is  simpler.
In fact 
we may then consider $\mathrm i[\mathbf DB,B]$ naturally extended 
from $C^\infty_{\mathrm c}(\mathbf X)$, and 
the extension coincides with \eqref{eq:17100813},
cf. Remarks \label{sec:second-commutator}
\ref{remark:self-adjo-realss} \eqref{item:self-adjo-real2}  and the proof of
Lemma~\ref{lem:17100620}. Again we use the decomposition
\eqref{eq:17112}. However now we can treat the
contribution from $\mathcal V_a$  more directly by using the more 
freedom of distributing factors of momenta (avoiding the
previous somewhat technical `commuting back and forth argument'). Noting that 
\begin{align}
B=B^a+B_a;\quad 
B^a=\tfrac12(\omega^a\cdot p^a+p^a\cdot \omega^a)
,\ \ 
B_a=\tfrac12(\omega_a\cdot p_a+p_a\cdot \omega_a),
\label{eq:1802165}
\end{align}  we decompose
\begin{align*}
  \mathrm i[\mathcal V_a,B]=\mathrm i[\mathcal V_a,B^a]+\mathrm i[\mathcal V_a,B_a].
\end{align*} By not doing the commutation the contribution from the first term is  seen to
be of the form $O(r^{-1-2\delta})$.
 Next, for  the second term, we compute as above 
\begin{align*}
\begin{split}
\mathrm i\bigl[&\mathcal V_a, B_a\bigr]
=\bigl((\omega_a\cdot \nabla_a)\omega^a\bigr)\cdot \bigl(\nabla^a V_a^{\mathrm{lr}}\bigr)
-\bigl(\omega_a\cdot \nabla_a
\mathop{\mathrm{div}}\omega^a\bigr)V_a^{\mathrm{sr}}\\& +
2\mathop{\mathrm{Im}}\bigl(\bigl(
V_a^{\mathrm{sr}}(\omega_a\cdot \nabla_a)\omega^a\bigr)\cdot
p^a\bigr)-2\mathop{\mathrm{Im}}\bigl(\bigl((
V_a^{\mathrm{sr}}\omega^a\cdot \nabla)\omega_a\bigr)\cdot
p_a\bigr)
+( V_a^{\mathrm{sr}}\omega^a\cdot
\nabla)\mathop{\mathrm{div}}\omega_a.
\end{split}
\end{align*} For the first three terms (containing $\omega_a\cdot \nabla_a$) we substitute 
\eqref{eq:1711021812} (now with $b=a$), bound separately  and then conclude that the contribution from
$\mathrm i\bigl[\mathcal V_a, B_a\bigr]$  is of the form $O(r^{-2})$. In conclusion we obtain that the
contribution from $\mathrm i[\mathcal V_a,B]$ to \eqref{eq:1710081345}
under
Condition~\ref{cond:smooth2wea3n12}
is of the form $O(r^{-1-2\delta})=O(r^{-1-2\kappa})$.
\end{proof}

Finally we consider as a continuation of Lemma \ref{lem:17111416} the second commutator of a function of $H$ and $B$.

\begin{lemma}\label{lem:17111515}
Suppose 
{Condition~\ref{cond:smooth2wea3n1} or Condition \ref{cond:smooth2wea3n12}},
and let $R\ge 1$ be sufficiently large.
For any $f\in \mathcal F^t$ with $t<-1$
the second commutator $\mathrm i\bigl[\mathrm i[f(H),B],B\bigr]$
extends to  a bounded sesquilinear  form on $\mathcal H$ from $\mathcal D(B)$, 
and it has the  expression 
\begin{align*}
\mathrm i\bigl[\mathrm i[f(H),B],B\bigr]
&
=-\int_{\C} (H-z)^{-1}\bigl(\mathrm i[\bD B,B]\bigr)(H-z)^{-1}\d \mu_f(z)
\\&\phantom{{}={}}{}
+2\int_{\C} (H-z)^{-1}(\bD B)(H-z)^{-1} (\bD B)(H-z)^{-1}\d \mu_f(z).
\end{align*} 
Moreover, with  the notation  \eqref{eq:1712022}
\begin{align}\label{eq:order}
  \mathrm i\bigl[\mathrm i[f(H),B],B\bigr] =O(r^{-1-2\kappa}).
\end{align}
\end{lemma}
\begin{proof}
By Lemmas~\ref{lem:17100620}, \ref{lem:17111416} and \ref{lemma:sing} we  calculate, 
as a sesquilinear  form on $\mathcal D(B)$, 
\begin{align*}
\mathrm i\bigl[\mathrm i[f(H),B],B\bigr]
&
=
\slim_{t\to 0}t^{-1}\bigl(\mathrm i[f(H),B]\e^{\i tB}-\e^{\i tB}\mathrm i[f(H),B]\bigr)
\\&
=
-\int_{\C} (H-z)^{-1}\bigl(\mathrm i[\bD B,B]\bigr)(H-z)^{-1}\d \mu_f(z)
\\&\phantom{{}={}}{}
+2\int_{\C} (H-z)^{-1}(\bD B)(H-z)^{-1} (\bD B)(H-z)^{-1}\d \mu_f(z).
\end{align*} 
By using the expression for $\bD B$ from Lemma~\ref{lem:17100620} and
\eqref{eq:1710081345} we obtain the boundedness and \eqref{eq:order} from the above  representation.
\end{proof}

\section{Proof of Rellich type  theorems}\label{sec:Proof}

In this section we  prove  Theorem~\ref{thm:priori-decay-b_0} under
Condition~\ref{cond:smooth2wea3n1} or Condition \ref{cond:smooth2wea3n12}, and
we prove Theorem \ref{thm:priori-decay-b_0b} 
under 
Condition~\ref{cond:smooth2wea3n1}.
 The proofs are given in Sections~\ref{subsec:17121421} and 
\ref{subsec:17121422}, respectively. 
The idea of proof comes from a combination of \cite{IS1}, \cite{IS2}, \cite{IS3} and \cite{FH}.
Note that in \cite{IS1} the second commutator $\mathrm i[\mathrm
i[H,B],B]$ was not needed,
but for Theorem~\ref{thm:priori-decay-b_0}  it is.
We also note that 
our arguments get  simpler if one considers only  
polynomial decay estimates at  non-threshold  energies $E$, see \cite{IS2}.

\subsection{Exponential decay estimates}\label{subsec:17121421}
 Throughout this subsection we impose Conditions~\ref{cond:smooth2wea3n1} or \ref{cond:smooth2wea3n12}.
  We introduce the regularized weights
\begin{align}
\Theta= \Theta_{m,n,R,\nu}^{\alpha,\beta}
=\chi_{m,n}\mathrm e^{\theta};\quad 
n,m\in \N_0, \ n>m,\ R\ge 1
\label{eq:15.2.15.5.8bb}
\end{align}
with exponents $\theta$ given by 
\begin{align*}
\theta=\theta_{R,\nu}^{\alpha,\beta}
=\alpha r+\beta\int_0^{r}(1+s/2^\nu)^{-1-2\kappa}\,\mathrm ds;\quad
\alpha,\beta\ge 0,\ \nu \in \N_0.
\end{align*}
Here $r=r_R$ indeed depends on $R\ge 1$, and 
$\kappa\in (0,1/4]$ is from \eqref{eq:171114}.
We are going to investigate the Heisenberg derivative 
of the `propagation observable' $P$ defined as 
\begin{align}
P=P^{\alpha,\beta,f}_{m,n,R,\nu,\epsilon}
=\Theta f(H) \zeta(B)f(H)\Theta
\in\mathcal L(\mathcal H);
\quad f\in C^\infty_{\mathrm c}(\mathbb R),\ \epsilon\in(0,1),
\label{eq:17092721b}
\end{align} 
where $\zeta=\zeta_{\epsilon}\in C^\infty(\mathbb R)$ is the smooth sign function 
from Section~\ref{subsubsec:Smooth sign function}. 

In the following we denote the derivatives of $\Theta$ and $\theta$ in $r$ by primes, e.g., 
\begin{align}
\begin{split}
\theta'=\alpha+\beta\theta_0^{-1-2\kappa},\ \ 
\theta''=-\beta \bigl(1+2\kappa\bigr)2^{-\nu}\theta_0^{-2-2\kappa}
;\quad
\theta_0=1+r/2^\nu.
\end{split}\label{eq:12.5.1.19.56}
\end{align}
In particular, noting that $2^{-\nu}\theta_0^{-1}\leq r^{-1}$, we have
\begin{align*}
|\theta^{(k)}|\leq C_k \beta r^{1-k}\theta_0^{-1-2\kappa}
\ \ \text{for }
k=2,3,\dots.
\end{align*}

\begin{lemma}\label{lem:14.10.4.1.17ffaabb}
Suppose Condition~\ref{cond:smooth2wea3n1} or Condition 
\ref{cond:smooth2wea3n12}. Let $E\in\mathbb R$ and $\alpha_0\ge 0$ satisfy $\lambda:=E+\alpha_0^2/2\not\in\mathcal T(H)$.
Then there exist $c,C>0$, $n_0\in \N$, $R\ge 1$, $\alpha_1\in
\{0\}\cup (0,\alpha_0)$, $\beta,\epsilon\in(0,1)$ and 
 real-valued $f\in C^\infty_{\mathrm c}(\mathbb R)$,
such that 
 for all   $n>m\ge 2n_0$, $\nu\ge 2n_0$
and $\alpha\in[\alpha_1,\alpha_0]$ 
\begin{align}
\begin{split}
2\mathop{\mathrm{Im}}\big(P(H-E)\bigr)
&\ge 
cr^{-1}\theta_0^{-2\kappa}\Theta^2
-C\bigl(\chi_{m-1,m+1}^2+\chi_{n-1,n+1}^2\bigr)r^{-1}\mathrm e^{2\theta}
\\&\phantom{{}={}}{}
-\mathop{\mathrm{Re}}\big(\Theta Q_1\Theta(H-E)\bigr)-\mathop{\mathrm{Re}}\big(\Theta \theta_0^{-\kappa}Q_2\theta_0^{-\kappa}\Theta(H-E)\bigr);
\end{split}
\label{eq:17092028b}
\end{align} here $Q_1,Q_2\in\mathcal
L(L^2_{-1/2},L^2_{1/2})$ are  symmetric,  (possibly)  depending on $\alpha$ and the
  other parameters  except though for $n,m$ and $\nu$, and the
  estimate \eqref{eq:17092028b} is understood as a quadratic form on
  $\mathcal H^2$. 
  \begin{remark*} We have stated more properties of $Q_1$ and $Q_2$
    than needed, for example boundedness on $L^2$ suffices and the
    independence of $n,m$ and $\nu$ is irrelevant.  
\end{remark*}
\end{lemma}
\begin{remark}\label{rem:180123}
We note that the constants, in particular $c,C>0$, can be chosen locally uniformly in 
$E\in\mathbb R$ and $\alpha_0>0$ with $E+\alpha_0^2/2\not\in\mathcal T(H)$.
This in fact enables us to apply the arguments of \cite{Pe} to
conclude that  the set of non-threshold eigenvalues of $H$ can
accumulate only at points in $\mathcal T(H)$ from below.
\end{remark}
\begin{proof}
\textit{Step I.} \enspace
We intend to apply Corollary~\ref{cor:171007}, and for that purpose
we fix some  variables in the following order: 
Let $E\in\mathbb R$ and $\alpha_0\ge 0$ be given and define $\lambda$
correspondingly.  Fix then 
 any $\sigma\in(0,\gamma(\lambda))$. 
Choose a big $R\ge 1$ and a neighborhood $\mathcal U\subset\mathbb R$ of $\lambda$
as in Corollary~\ref{cor:171007},
and let $f\in C^\infty_{\mathrm c}(\mathcal U)$ be a function 
such that $0\le f\le 1$ in $\mathcal U$ and $f=1$ in a neighborhood $I$ of $\lambda$. 
Then Corollary~\ref{cor:171007} asserts that 
\begin{align}
\begin{split}
f(H)(\mathbf DB)f(H)
\ge 
f(H)r^{-1/2}\bigl(\sigma^2-B^2\bigr)r^{-1/2}f(H)
-C_1r^{-2}
,
\end{split}
\label{eq:171025}
\end{align}
which we will implement in Step I\hspace{-.05em}V below. We fix
$\alpha_1\in \{0\}\cup (0,\alpha_0)$  such that
$$\inf_{\alpha\in [\alpha_1,\alpha_0]} d\bigl(E+\tfrac12\alpha^2,\mathbb R\setminus I\bigr)>0.$$ 

With these variables we consider the operator  $P$ given in   \eqref{eq:17092721b}.
Note though that we have not yet fixed $c,C>0$, $n_0\ge 1$, 
  $\beta,\epsilon\in(0,1)$, 
$Q_1$ and $Q_2$. We will need $\epsilon^2<\sigma^2/2$. These quantities  will be chosen in Step II.
In the following estimates the dependence on  
$\beta,\epsilon\in(0,1)$, $\epsilon^2<\sigma^2/2$,  and $n_0\ge 1$ will  always be emphasized, 
and  the estimates will be uniform in $n,m,\nu$ and $\alpha$
fulfilling $n>m\ge 2n_0$, $\nu\ge 2n_0$ and  $\alpha\in[\alpha_1,\alpha_0]$ (as required for \eqref{eq:17092028b}). 

We will repeatedly use (small variations of) Lemmas~\ref{lem:171113}, \ref{lem:171113b} and \ref{lem:17111416}
to bound commutators of functions of $H$, $r$ and $B$, mostly  without
reference. It is assumed that  $R\ge 1$ is chosen so large that not
only \eqref{eq:171025} is valid, but also that these lemmas apply for
this (fixed)  $R$.

\Step{Step I\hspace{-.05em}I}
We are calculate and bound 
the left-hand side of \eqref{eq:17092028b}.
By the definition \eqref{eq:17092721b}   we  compute 
\begin{align}
\begin{split}
2\mathop{\mathrm{Im}}\big(P(H-E)\bigr)
=\bD P
&
=
2\mathop{\mathrm{Re}}
\bigl[ \bigl(\mathbf D \Theta \bigr)f(H)\zeta(B)f(H)\Theta \bigr]
\\&\phantom{{}={}}{}
+\Theta f(H) \bigl(\bD \zeta(B)\bigr) f(H)\Theta.
\end{split}
\label{eq:1709281948}
\end{align}
We claim that the two terms on the right-hand side of \eqref{eq:1709281948} are 
bounded from below as 
\begin{align}\label{eq:com1}
\begin{split}
&
2\mathop{\mathrm{Re}}
\bigl[ (\mathbf D \Theta )f(H)\zeta(B)f(H)\Theta \bigr]
\\&
\ge 
2^{n_0+1}\alpha r^{-1/2}\Theta f(H)B\zeta(B)f(H)\Theta r^{-1/2}
-\bigl(C_2+C_3(\epsilon) 2^{-n_0}\bigr)\alpha r^{-1}\Theta^2
\\&\phantom{{}={}}{}
+2^{n_0}\beta  r^{-1/2}\theta_0^{-\kappa}\Theta f(H)B\zeta(B)f(H)\Theta\theta_0^{-\kappa}r^{-1/2}
-C_3(\epsilon)\beta r^{-1}\theta_0^{-2\kappa}\Theta^2
\\&\phantom{{}={}}{}
-C_3(\epsilon)r^{-2}\Theta^2
-C_3(\epsilon)\bigl(\chi_{m-1,m+1}^2+\chi_{n-1,n+1}^2\bigr)r^{-1}\mathrm e^{2\theta}
, 
\end{split}
\intertext{and}
\label{eq:17092820}
\begin{split}
&\Theta f(H) \bigl(\bD \zeta(B)\bigr)f(H) \Theta
\\&
\ge 
\bigl(\tfrac12\sigma^2 -\epsilon^2\bigr)r^{-1/2}\Theta f(H)\zeta'(B)f(H)\Theta r^{-1/2}
\\&\phantom{{}={}}{}
+\tfrac12\sigma^2 
r^{-1/2}\theta_0^{-\kappa}\Theta f(H)\zeta'(B)f(H)\Theta \theta_0^{-\kappa}r^{-1/2}
-C_3(\epsilon)r^{-1-2\kappa}\Theta^2
,
\end{split}
\end{align} 
both of which are uniform in $\beta,\epsilon\in(0,1)$,
$\epsilon^2<\sigma^2/2$  and $n_0\ge 1$, 
and also in $\alpha\in[\alpha_1,\alpha_0]$, $n>m\ge 2n_0$ and $\nu\ge 2n_0$. 
We will prove these bounds \eqref{eq:com1} and \eqref{eq:17092820} 
later in Steps I\hspace{-.05em}I\hspace{-.05em}I and I\hspace{-.05em}V, respectively.
 For the moment let us assume them.
Then by \eqref{eq:1709281948}--\eqref{eq:17092820} and Lemma~\ref{lem:170928}
we obtain 
\begin{align}
\begin{split}
2\mathop{\mathrm{Im}}\big(P(H-E)\bigr)
&
\ge 
\min\{2^{n_0-2}\alpha \epsilon,(\sigma^2 -2\epsilon^2)/\epsilon \}
r^{-1/2}\Theta f(H)^2\Theta r^{-1/2}
\\&\phantom{{}={}}{}
-\bigl(C_2+C_3(\epsilon) 2^{-n_0}\bigr)\alpha r^{-1}\Theta^2
\\&\phantom{{}={}}{}
+\min\{2^{n_0-3}\beta \epsilon,\sigma^2 /\epsilon\}
r^{-1/2}\theta_0^{-\kappa}\Theta f(H)^2\Theta\theta_0^{-\kappa}r^{-1/2}
\\&\phantom{{}={}}{}
-C_3(\epsilon)\beta r^{-1}\theta_0^{-2\kappa}\Theta^2
-2C_3(\epsilon)r^{-1-2\kappa}\Theta^2
\\&\phantom{{}={}}{}
-C_3(\epsilon)\bigl(\chi_{m-1,m+1}^2+\chi_{n-1,n+1}^2\bigr)r^{-1}\mathrm e^{2\theta}
.
\end{split}
\label{eq:171103}
\end{align}

Now we are going to remove $f(H)^2$ from the first and third terms of \eqref{eq:171103}
with some controllable errors, and for that end we 
 introduce $f_1\in \mathcal F^{-1}$ by
 \begin{align}
   \label{eq:f1}
   f_1(t)=\bigl(1-f(t)^2\bigr)\bigl(t-\lambda\bigr)^{-1},
 \end{align} 
so that (uniformly in $\alpha\in [\alpha_1,\alpha_0]$)
\begin{align*}
1-f(H)^2
&\le C_4f_1(H)\bigl(H-E-\tfrac12\alpha^2\bigr).
\end{align*}
Then we estimate
\begin{align}
\begin{split}
&r^{-1/2}\Theta \bigl(1-f(H)^2\bigr)\Theta r^{-1/2}
\\&\le 
C_4\mathop{\mathrm{Re}}\bigl[r^{-1/2}\Theta f_1(H)\Theta r^{-1/2}\bigl(H-E-\tfrac12\alpha^2\bigr)\bigr]
\\&\phantom{{}={}}{}
+\tfrac12C_4\mathop{\mathrm{Re}}\bigl[r^{-1/2}\Theta f_1(H)\omega^2(\Theta r^{-1/2})''\bigr]
+C_4\mathop{\mathrm{Im}}\bigl[r^{-1/2}\Theta f_1(H)B(\Theta r^{-1/2})' \bigr]
\\&
\le 
C_4\mathop{\mathrm{Re}}\bigl[r^{-1/2}\Theta f_1(H)\Theta r^{-1/2}(H-E)\bigr]
\\&\phantom{{}={}}{}
+C_5r^{-3/2}\Theta^2
+C_5\beta r^{-1}\Theta^2
+C_5\bigl(\chi_{m-1,m+1}^2+\chi_{n-1,n+1}^2\bigr)r^{-2}\mathrm e^{2\theta}
.
\end{split}
\label{eq:171214}
\end{align}
Here we have used the inequality  
$$|\omega^2-1|\le C_7 r^{-1/2},$$
which is a consequence of \eqref{eq:boun_constr} (a stronger bound
holds, but this is not  needed).
Hence by \eqref{eq:171214} it follows that uniformly in large $n_0\ge 1$ and small $\beta\in(0,1)$ 
\begin{align}
\begin{split}
r^{-1/2}\Theta f(H)^2\Theta r^{-1/2}
&\ge 
c_1r^{-1}\Theta^2
-\mathop{\mathrm{Re}}\bigl(\Theta Q_1\Theta(H-E)\bigr)
\\&\phantom{{}={}}{}
-C_5\bigl(\chi_{m-1,m+1}^2+\chi_{n-1,n+1}^2\bigr)r^{-2}\mathrm e^{2\theta},
\end{split}
\label{eq:17110320}
\end{align}
where $Q_1=C_4r^{-1/2}f_1(H)r^{-1/2}$.
Similarly, we can show that 
uniformly in large $n_0\ge 1$ and small $\beta\in(0,1)$
\begin{align}
\begin{split}
r^{-1/2}\theta_0^{-\kappa}\Theta f(H)^2\Theta\theta_0^{-\kappa}r^{-1/2}
&\ge 
c_2r^{-1}\theta_0^{-2\kappa}\Theta^2
-\mathop{\mathrm{Re}}\bigl(\Theta \theta_0^{-\kappa}Q_2\theta_0^{-\kappa}\Theta(H-E)\bigr)
\\&\phantom{{}={}}{}
-C_6\bigl(\chi_{m-1,m+1}^2+\chi_{n-1,n+1}^2\bigr)r^{-2}\mathrm e^{2\theta}
,
\end{split}
\label{eq:17110321}
\end{align}
where $Q_2=C_4r^{-1/2}f_1(H)r^{-1/2}$.

Finally by \eqref{eq:171103}, \eqref{eq:17110320} and \eqref{eq:17110321}
it follows that 
\begin{align*}
\begin{split}
&2\mathop{\mathrm{Im}}\big(P(H-E)\bigr)
\\&
\ge 
\Bigl[
c_1\min\{ 2^{n_0-2}\alpha\epsilon,(\sigma^2 -2\epsilon^2)/\epsilon \}
-\bigl(C_2+C_3(\epsilon) 2^{-n_0}\bigr)\alpha \Bigr]
r^{-1}\Theta^2
\\&\phantom{{}={}}{}
+\Bigl[
c_2\min\{2^{n_0-3}\beta \epsilon,\sigma^2 /\epsilon\}
-C_3(\epsilon)\beta-2C_3(\epsilon)\theta_0^{2\kappa}r^{-2\kappa}\Bigr]r^{-1}\theta_0^{-2\kappa}\Theta^2
\\&\phantom{{}={}}{}
-C_{7}(n_0,\epsilon)\bigl(\chi_{m-1,m+1}^2+\chi_{n-1,n+1}^2\bigr)r^{-1}\mathrm e^{2\theta}
\\&\phantom{{}={}}{}
-\min\{ 2^{n_0-2}\alpha\epsilon,(\sigma^2 -2\epsilon^2)/\epsilon \}\mathop{\mathrm{Re}}\bigl(\Theta Q_1\Theta(H-E)\bigr)\\&\phantom{{}={}}{}
-\min\{2^{n_0-3}\beta \epsilon,\sigma^2 /\epsilon\}\mathop{\mathrm{Re}}\bigl(\Theta \theta_0^{-\kappa}Q_2\theta_0^{-\kappa}\Theta(H-E)\bigr).
\end{split}
\end{align*}
We can bound $\theta_0^{2\kappa}r^{-2\kappa}\leq 2^{2\kappa
}2^{-4\kappa n_0}$ on the support of $\Theta^2$ for $\nu\geq 2n_0$. Now, first let $\epsilon>0$ be small enough,
then choose $\beta\in(0,1)$ small enough,
and finally let $n_0\ge 1$ be large enough.
We  conclude  the desired bound \eqref{eq:17092028b} since the first
square bracket expression is then non-negative while the second is
bounded from 
below by some constant $c>0$. This $c$ and  $C=C_{7}(n_0,\epsilon)$
work.

\Step{Step I\hspace{-.05em}I\hspace{-.05em}I} 
 We  prove \eqref{eq:com1}. 
Substitute the expression
\begin{align*}
\bD \Theta 
&=\mathop{\mathrm{Re}}\bigl(\Theta' \omega\cdot p\bigr)
=
\theta'\Theta B
+\chi_{m,n}'\mathrm e^\theta B
-\tfrac{\mathrm i}2\omega^2
\bigl(\theta''\Theta+\theta'{}^2\Theta+\chi_{m,n}''\mathrm e^\theta+2\chi_{m,n}'\theta'\mathrm e^\theta\bigr)
\end{align*}
and then we can first write 
\begin{align}
\begin{split}
&2\mathop{\mathrm{Re}}\bigl[ \bigl(\mathbf D \Theta \bigr)f(H)\zeta(B)f(H)\Theta \bigr]
\\&=
2\mathop{\mathrm{Re}}\bigl[ \theta'\Theta Bf(H)\zeta(B)f(H)\Theta \bigr]
+2\mathop{\mathrm{Re}}\bigl[ \chi_{m,n}'\mathrm e^\theta Bf(H)\zeta(B)f(H)\Theta \bigr]
\\&\phantom{{}={}}{}
+
\mathop{\mathrm{Im}}
\bigl[\omega^2
\bigl(\theta''\Theta+\theta'{}^2\Theta+\chi_{m,n}''\mathrm e^\theta+2\chi_{m,n}'\theta'\mathrm e^\theta\bigr)f(H)\zeta(B)f(H)\Theta \bigr].
\end{split}
\label{eq:171029}
\end{align}

The first term of \eqref{eq:171029} can be calculated by \eqref{eq:12.5.1.19.56} as 
\begin{align}
\begin{split}
&2\mathop{\mathrm{Re}}\bigl[ \theta'\Theta Bf(H)\zeta(B)f(H)\Theta \bigr]
\\&
=
2\alpha\mathop{\mathrm{Re}}\bigl[\Theta Bf(H)\zeta(B)f(H)\Theta \bigr]
+2\beta\mathop{\mathrm{Re}}\bigl[\theta_0^{-1-2\kappa}\Theta Bf(H)\zeta(B)f(H)\Theta \bigr]
\\&
\ge 
2\alpha \Theta f(H)B\zeta(B)f(H)\Theta 
-\bigl(C'_{1}+C'_{2}(\epsilon)r^{-1}\bigr)\alpha r^{-1}\Theta^2
\\&\phantom{{}={}}{}
+2\beta \mathop{\mathrm{Re}}\bigl[ \theta_0^{-1-2\kappa}\Theta f(H)B\zeta(B)f(H)\Theta \bigr]
\\&\phantom{{}={}}{}
-\bigl(C'_{1}+C'_{2}(\epsilon)r^{-1}\bigr)\beta r^{-1}\theta_0^{-1-2\kappa}\Theta^2
.
\end{split}
\label{eq:17111513}
\end{align}
We define $Z_1\in \mathcal F^{1/2}$ as 
\begin{align*}
Z_1(b)=b\sqrt{b^{-1}\zeta(b)}.
\end{align*}
Then we have $Z_1(b)^2=b\zeta(b)$, so that  for the first term of \eqref{eq:17111513}
\begin{align}
\begin{split}
&2\alpha \Theta f(H)B\zeta(B)f(H)\Theta 
\\&
\ge 
2\alpha \Theta f(H)Z_1(B)\bigl(2^{n_0}r^{-1}\bar\chi_{n_0}\bigr)Z_1(B)f(H)\Theta 
\\&
\ge 
2^{n_0+1}\alpha r^{-1/2}\Theta f(H)B\zeta(B)f(H)\Theta r^{-1/2}
-\bigl(C'_{3}+C'_{4}(\epsilon) r^{-1}\bigr)2^{n_0}\alpha r^{-2}\Theta^2;
\end{split}
\label{eq:17111513b}
\end{align} here we used Lemma  \ref{lem:171113b} repeatedly.

Similarly,  for the third term of \eqref{eq:17111513}
\begin{align}
\begin{split}
&2\beta \mathop{\mathrm{Re}}\bigl[ \theta_0^{-1-2\kappa}\Theta f(H)B\zeta(B)f(H)\Theta \bigr]
\\&
\ge 
2\beta \theta_0^{-\kappa}\Theta f(H)Z_1(B)\theta_0^{-1}Z_1(B)f(H)\Theta \theta_0^{-\kappa}
-C'_{5}(\epsilon)\beta r^{-1}\theta_0^{-1-2\kappa}\Theta^2
\\&
\ge 
2\beta \theta_0^{-\kappa}\Theta f(H)Z_1(B)\bigl(2^{n_0-1}r^{-1}\bar\chi_{n_0}\bigr)Z_1(B)f(H)\Theta \theta_0^{-\kappa}
\\&\phantom{{}={}}{}
-C'_{5}(\epsilon)\beta r^{-1}\theta_0^{-1-2\kappa}\Theta^2
\\&
\ge 
 2^{n_0}\beta r^{-1/2}\theta_0^{-\kappa}\Theta f(H)B\zeta(B)f(H)r^{-1/2}\theta_0^{-\kappa}\Theta
\\&\phantom{{}={}}{}
-C'_{5}(\epsilon)\beta r^{-1}\theta_0^{-1-2\kappa}\Theta^2
-C'_{6}(\epsilon)2^{n_0}\beta r^{-2}\theta_0^{-2\kappa}\Theta^2
.
\end{split}
\label{eq:17111513bb}
\end{align}
By \eqref{eq:17111513}, \eqref{eq:17111513b} and \eqref{eq:17111513bb} we obtain 
for the first term of \eqref{eq:171029}
\begin{align}
\begin{split}
&2\mathop{\mathrm{Re}}\bigl[ \theta'\Theta Bf(H)\zeta(B)f(H)\Theta \bigr]
\\&
\ge 
2^{n_0+1}\alpha r^{-1/2}\Theta f(H)B\zeta(B)f(H)r^{-1/2}\Theta 
-\bigl(C'_{7}+C'_{8}(\epsilon) 2^{-n_0}\bigr)\alpha r^{-1}\Theta^2
\\&\phantom{{}={}}{}
+2^{n_0}\beta  r^{-1/2}\theta_0^{-\kappa}\Theta f(H)B\zeta(B)f(H)r^{-1/2}\theta_0^{-\kappa}\Theta
-C'_{8}(\epsilon)\beta r^{-1}\theta_0^{-2\kappa}\Theta^2
.
\end{split}
\label{eq:17111513bbb}
\end{align}

On the other hand, note that by \eqref{eq:boun_constr} we have 
\begin{align*}
|\omega\cdot \nabla \omega^2|\le C'_{9}r^{-2}.
\end{align*}
Then the second and third terms of \eqref{eq:171029} are bounded as 
\begin{align}
\begin{split}
&2\mathop{\mathrm{Re}}\bigl[ \chi_{m,n}'\mathrm e^\theta Bf(H)\zeta(B)f(H)\Theta \bigr]
\\&\phantom{{}={}}{}
+
\mathop{\mathrm{Im}}
\bigl[\omega^2
\bigl(\theta''\Theta+\theta'{}^2\Theta+\chi_{m,n}''\mathrm e^\theta
+2\chi_{m,n}'\theta'\mathrm e^\theta\bigr)f(H)\zeta(B)f(H)\Theta \bigr]
\\&
\ge 
-C'_{10}\alpha r^{-1}\Theta^2
-C'_{11}(\epsilon)\beta r^{-1}\theta_0^{-1-2\kappa}\Theta^2
-C'_{11}(\epsilon)r^{-2}\Theta^2
\\&\phantom{{}={}}{}
-C'_{11}(\epsilon)\bigl(\chi_{m-1,m+1}^2+\chi_{n-1,n+1}^2\bigr)r^{-1}\mathrm e^{2\theta}.
\end{split}
\label{eq:17111513bbbb}
\end{align}
By \eqref{eq:171029}, \eqref{eq:17111513bbb} and \eqref{eq:17111513bbbb} we obtain \eqref{eq:com1}.

\Step{Step I\hspace{-.05em}V} 
Here we prove \eqref{eq:17092820}.
Take a real-valued function $\tilde f\in C_\c^\infty(\mathcal U)$ 
such that $\tilde f =1$ on $\mathop{\mathrm{supp}}f$,
and set 
$$g(\lambda)=\lambda\tilde f(\lambda).$$
Then we can write 
\begin{align}
f(H)\bigl(\mathbf D\zeta(B)\bigr)f(H)
&=f(H)\bigl(\mathrm i[g(H),\zeta(B)]\bigr)f(H).
\label{eq:17102923a}
\end{align}
By Lemmas~\ref{lem:17111416}, \ref{lem:A1} and \ref{lem:17111515} we have 
\begin{align}
\begin{split}
\mathrm i[g(H),\zeta(B)]
&
=
\mathop{\mathrm{Re}}\bigl[\zeta'(B)\bigl(\mathrm i[g(H),B]\bigr)\bigr]
\\&\phantom{{}={}}{}
+\mathop{\mathrm{Re}}\int _{\mathbb C}(B-z)^{-2}\bigl[\mathrm i[g(H),B],B\bigr](B-z)^{-1}
\,\mathrm d\mu_{\zeta}(z)
\\&
\ge 
Z_2(B)\bigl(\mathrm i[g(H),B]\bigr)Z_2(B)
-C'_{1}(\epsilon)r^{-1-2\kappa}
,
\end{split}
\label{eq:17102923}
\end{align}
where we have set $Z_2=\sqrt{\zeta'}\in C^\infty_{\mathrm c}(\mathbb R)$.
Using \eqref{eq:171025} the  contribution from the first term of \eqref{eq:17102923}
is bounded  as 
\begin{align}
\begin{split}
&
f(H)Z_2(B)\bigl(\mathrm i[g(H),B]\bigr)Z_2(B)f(H)
\\&
\ge 
Z_2(B)f(H)\bigl(\mathrm i[H,B]\bigr)f(H)Z_2(B)
-C'_{2}(\epsilon)r^{-2}
\\&
\ge 
Z_2(B)f(H)r^{-1/2}\bigl(\sigma^2-B^2\bigr)r^{-1/2}f(H)Z_2(B)
-C'_{3}(\epsilon)r^{-2}
\\&
\ge 
\bigl(\tfrac12\sigma^2 -\epsilon^2\bigr)r^{-1/2}f(H)\zeta'(B)f(H)r^{-1/2}
\\&\phantom{{}={}}{}
+\tfrac12\sigma^2 r^{-1/2}\theta_0^{-\kappa}f(H)\zeta'(B)f(H)\theta_0^{-\kappa}r^{-1/2}
-C'_{4}(\epsilon)r^{-2}
.
\end{split}
\label{eq:17102923b}
\end{align} 
  We obtain 
\eqref{eq:17092820} by combining \eqref{eq:17102923a}--\eqref{eq:17102923b}.
\end{proof}

\begin{proof}[Proof of Theorem~\ref{thm:priori-decay-b_0}]  
Let $\phi\in \mathcal B_0^*\cap H^1_{0,\mathrm{loc}}(\Omega)$, $E\in\mathbb R$, $\rho\ge 0$,
and $\alpha_0\in[0,\infty]$ be as in the assumptions of Theorem~\ref{thm:priori-decay-b_0}.
We assume 
$$E+\tfrac12\alpha_0^2\not\in\mathcal T(H)\cup \{\infty\}$$
and deduce a contradiction.
For the above $E$ and $\alpha_0$ we 
choose $c$, $C$, $n_0$, $R$, $\alpha_1$, $\beta$, $\epsilon$, $f$ and $Q$
in agreement  with Lemma~\ref{lem:14.10.4.1.17ffaabb}.
Note that we may take $n_0$ larger if necessary so that 
$2^{2n_0-3}>\rho$.
Note that for all $n>m\ge 2n_0$
$$\chi_{m-2,n+2}\phi\in \mathcal D(H).$$
We can also choose $\alpha\in[\alpha_1,\alpha_0]$ such that $\alpha+\beta>\alpha_0$.
With these variables we evaluate the inequality \eqref{eq:17092028b}
in the state $\chi_{m-2,n+2}\phi\in \mathcal D(H)$ 
and then obtain for all  $n>m\ge 2n_0$ and $\nu\ge 2n_0$
\begin{align}
\begin{split}
\bigl\|r^{-1/2}\theta_0^{-\kappa}\Theta\phi\bigr\|^2
&
\le 
C_1(m)\|\chi_{m-1,m+1}\phi\|^2
+C_2(\nu) 2^{-n}\|\chi_{n-1,n+1}\mathrm e^{\alpha r}\phi\|^2
.
\end{split}
\label{eq:11.7.16.3.22a}
\end{align}
The second term on the right-hand side of (\ref{eq:11.7.16.3.22a})
 vanishes when $n\to\infty$ since 
$\mathrm e^{\alpha r}\phi\in \mathcal B^*_0$, 
and consequently by Lebesgue's monotone convergence
theorem
\begin{align}
\bigl\|\bar\chi_m r^{-1/2}\theta_0^{-\kappa}
\mathrm{e}^{\theta}\phi\bigr\|^2
 &\le 
C_1(m)\|\chi_{m-1,m+1}\phi\|^2.
\label{eq:11.7.16.3.43a}
\end{align}
Next we let $\nu \to\infty$ in \eqref{eq:11.7.16.3.43a}
 invoking again Lebesgue's monotone convergence
theorem,
and then it follows that 
 $$\bar\chi_m r^{-1/2}\mathrm e^{(\alpha+\beta) r}\phi\in \mathcal H.$$ 
Consequently $\mathrm e^{(\alpha+\beta)r}\phi\in \mathcal B_0^*$,
but this is a contradicts that   $\alpha+\beta>\alpha_0$. 
\end{proof}

\subsection{Super-exponentially decaying eigenfunctions}\label{subsec:17121422}

Throughout this subsection we impose
Condition~\ref{cond:smooth2wea3n1}. We shall use a function introduced in \cite{Ya2}. 
Although we do not present its construction, 
we list the properties required in the arguments of the paper.
\begin{lemma}\label{lemma:another-vector-field} 
There exists a real-valued function $Y\in C^\infty({\bf X}\setminus\{0\})$ 
such that 
\begin{enumerate}[\normalfont (1)]
\item\label{item:2}
$Y$ is homogeneous of degree one;
\item \label{item:3}
$Y(x)\geq 1$ for $|x|=1$;
\item \label{item:4} 
$Y$ is convex;
\item \label{item:5} 
There exists $c\in (0,1)$ such that for all $a\in \vA$
\begin{align}
\label{eq:dependence}Y(x)=Y(x_a)\ \text{ for }|x_a|\ge (1-c) |x|.
\end{align}
\end{enumerate}
\end{lemma}
\begin{proof}
We omit the proof. See \cite{Ya2} or \cite[Theorem~7.4]{HuS}.
\end{proof}

Define
\begin{align*}
q(x)=q_R(x)=\chi(|x|/R)+\bigl(1-\chi(|x|/R)\bigr)Y(x);\quad R\ge 1.
\end{align*}
We set 
\begin{align*}
\omega_q=\mathop{\mathrm{grad}}q,\quad 
 h_q=\mathop{\mathrm{Hess}}q,
\end{align*}
and 
\begin{align}
\label{eq:conju_new}
B_q=\mathop{\mathrm{Re}}(\omega_q\cdot p)=\tfrac12(\omega_q\cdot p+p\cdot \omega_q).
\end{align}

\begin{lemma}\label{lemma:Flowbb} 
Let $R\ge 1$ be sufficiently large. Then the operator $B_q$ defined by \eqref{eq:conju_new}
is essentially self-adjoint on $C^\infty_{\mathrm c}(\Omega)$,
and the self-adjoint extension, denoted by $B_q$ again,
satisfies that for some $C>0$
\begin{align*}
\mathcal D(B_q)\supset \mathcal H^1,\quad 
\|B_q\psi\|_{\mathcal H}\le C\|\psi\|_{\mathcal H^1}\ \ \text{for any }\psi\in\mathcal H^1.
\end{align*}
In addition the operators  $\mathrm e^{\mathrm itB_q}$, $t\in\mathbb
R$, 
naturally restrict/extend
as bounded operators $\mathrm e^{\mathrm itB_q}\colon \mathcal H^{\pm1}\to\mathcal H^{\pm1}$, and 
they satisfy
\begin{align*}
\sup_{t\in [-1,1]}\|\mathrm e^{\mathrm itB_q}\|_{\mathcal L(\mathcal H^{\pm 1})}<\infty,
\end{align*}
respectively. 
\end{lemma}
\begin{proof}
The assertion is obvious by Lemma~\ref{lemma:Flowb}.
\end{proof}

\begin{lemma}\label{lem:17100620bb}
Let $R\ge 1$ be sufficiently large.
Then the quadratic form ${\bf D}B_q:=\mathrm i[H,B_q]$ defined on
$\mathcal H^2$ is given by 
\begin{align*}
\begin{split}
{\bf D}B_q=p&\cdot h_qp
-\tfrac14(\Delta^2q)\\
&+\sum_{a\in\mathcal A}
\Bigl(-\omega_q^a\cdot \bigl(\nabla^a V_a^{\mathrm{lr}}\bigr)
+\bigl(V_a^{\mathrm{sr}}\omega_q^a\bigr)\cdot \nabla^a
-\nabla^a\cdot \bigl(V_a^{\mathrm{sr}}\omega_q^a\bigr)+
V_a^{\mathrm{sr}}\mathop{\mathrm{div}}  \omega_q^a 
\Bigr)
.
\end{split}
\end{align*}

Moreover there exist $C,C'>0$ such that, as quadratic forms on $\mathcal H^2$, 
\begin{align}\label{eq:formulasybbb}
{\bf D}B_q\ge -r^{-1/2-\delta}(CH+C')r^{-1/2-\delta}
. 
\end{align} 
\end{lemma}
\begin{proof}
  We can argue as in the proof of Lemma~\ref{lem:17100620}.  The bound
  \eqref{eq:formulasybbb} follows by using the computed expression for
  ${\bf D}B_q$ and the fact that $ h_q(x)\ge 0$ for $|x|>2R$.
\end{proof}

Now we proceed somewhat  as in the previous subsection.
We introduce the regularized weights slightly different from \eqref{eq:15.2.15.5.8bb}:
\begin{align}
\Theta= \Theta_{m,n,R}^{\alpha,\delta'}
=\eta_{m,n}\mathrm e^{\theta};\quad 
n>m\ge 0,\ R\ge 1.
\label{eq:15.2.15.5.8bbbb}
\end{align}
Here we set, as in \eqref{eq:14.1.7.23.24} and \eqref{eq:1711021}, 
\begin{align*}
\eta_m=\chi(q/2^m),\quad \bar \eta_m=1-\eta_m,\quad \eta_{m,n}=\bar\eta_m\eta_n,
\end{align*} 
and 
\begin{align*}
\theta=\theta_{R}^{\alpha,\delta'}
=\alpha (q-q^{1-2\delta'});\quad
\alpha\ge 0,\ \delta'\in(0,\delta).
\end{align*}
We are going to investigate the Heisenberg derivative 
of the `propagation observable' $P$ defined here as 
\begin{align}
P=P^{\alpha,\delta'}_{m,n,R}
=\Theta B_q\Theta.
\label{eq:17092721bbb}
\end{align} 
In the following we denote the derivatives of $\Theta$ and $\theta$ in $q$ by primes.

\begin{lemma}\label{lem:14.10.4.1.17ffaabbqqq}
Let $E\in\mathbb R$
and $\delta'\in(0,\delta)$.
Then there exist $c,C,C',C''>0$, $\beta_0\ge 1$, $n_0\ge 1$ and $R\ge 1$
such that 
uniformly in $\alpha\ge \beta_0$ and $n>m\ge n_0$, as quadratic forms on $\mathcal H^2$,
\begin{align}
\begin{split}
2\mathop{\mathrm{Im}}\bigl(P(H-E)\bigr)
&\ge 
c\alpha^2q^{-1-2\delta'}\Theta^2
-C\alpha\bigl(\eta_{m-1,m+1}^2+\eta_{n-1,n+1}^2\bigr)r^{-2}\mathrm e^{2\theta}
\\&\phantom{{}={}}{}
-\mathop{\mathrm{Re}}\bigl(Q(H-E)\bigr)
,
\end{split}
\label{14.9.26.9.53ffaabbqqqq}
\end{align}
where $Q=C'q^{-1-2\delta}\Theta^2+C''(\eta'_{m,n})^2\mathrm e^{2\theta}$.
\end{lemma}
\begin{proof}
We proceed in parallel with the proof of Lemma~\ref{lem:14.10.4.1.17ffaabb}. 
We first have 
\begin{align}
2\mathop{\mathrm{Im}}\bigl(P(H-E)\bigr)
=\mathbf DP
&
=
2\mathop{\mathrm{Re}}\bigl((\mathbf D\Theta)B_q\Theta\bigr)
+
\Theta(\mathbf DB_q)\Theta,
\label{eq:17111523}
\end{align}
and further calculate each term on the right-hand side of \eqref{eq:17111523}.
We note that 
\begin{align*}
\omega_q\cdot\nabla\omega_q^2=2\omega_q\cdot h_q\omega_q\ge 0  \text{ for }|x|>2R.
\end{align*}
Then for the first term of \eqref{eq:17111523} 
by letting $\beta_0\ge 1$ and $n_0\ge 1$ be sufficiently large 
\begin{align}
\begin{split}
&2\mathop{\mathrm{Re}}\bigl((\mathbf D\Theta)B_q\Theta\bigr)
\\&
=
2\mathop{\mathrm{Re}}\bigl(\theta'\Theta B_q^2\Theta\bigr)
+2\mathop{\mathrm{Re}}\bigl(\eta_{m,n}'\mathrm e^{\theta}B_q^2\Theta\bigr)
+\mathop{\mathrm{Im}}\bigl(\omega_q^2(\theta'{}^2+\theta'')\Theta B_q\Theta\bigr)
\\&\phantom{{}={}}{}
+\mathop{\mathrm{Im}}\bigl(\omega_q^2\eta_{m,n}''\mathrm e^{\theta}B_q\Theta\bigr)
+2\mathop{\mathrm{Im}}\bigl(\omega_q^2\eta_{m,n}'\theta'\mathrm e^{\theta}B_q\Theta\bigr)
\\&
\ge 
c_1\alpha \Theta B_q^2\Theta
-C_1\eta_{m,n}'\mathrm e^{\theta}H\mathrm e^{\theta}\eta_{m,n}'
+c_1\alpha^2q^{-1-2\delta'}\Theta^2
\\&\phantom{{}={}}{}
-C_1\alpha\bigl(\eta_{m-1,m+1}^2+\eta_{n-1,n+1}^2\bigr)q^{-2}\mathrm e^{2\theta}
.
\end{split}
\label{eq:17111523b}
\end{align}
On the other hand, as for the second term of \eqref{eq:17111523}, 
we use Lemma~\ref{lem:17100620bb}, and obtain
\begin{align}
\Theta(\mathbf DB_q)\Theta
&
\ge 
-C_2q^{-1-2\delta}\Theta^2
-C_3q^{-1/2-\delta}\Theta H\Theta q^{-1/2-\delta}.
\label{eq:17111523bb}
\end{align}
By \eqref{eq:17111523}, \eqref{eq:17111523b} and \eqref{eq:17111523bb} 
it follows that for sufficiently large $\beta_0\ge 1$ and $n_0\ge 1$
\begin{align*}
2\mathop{\mathrm{Im}}\bigl(P(H-E)\bigr)
&\ge 
c_2\alpha^2q^{-1-2\delta'}\Theta^2
-C_1\alpha\bigl(\eta_{m-1,m+1}^2+\eta_{n-1,n+1}^2\bigr)q^{-2}\mathrm e^{2\theta}
\\&\phantom{{}={}}{}
-C_3q^{-1/2-\delta}\Theta (H-E)\Theta q^{-1/2-\delta}
-C_1\eta_{m,n}'\mathrm e^{\theta}(H-E)\mathrm e^{\theta}\eta_{m,n}'
\\&
\ge 
c_3\alpha^2q^{-1-2\delta'}\Theta^2
-C_4\alpha\bigl(\eta_{m-1,m+1}^2+\eta_{n-1,n+1}^2\bigr)q^{-2}\mathrm e^{2\theta}
\\&\phantom{{}={}}{}
-C_3\mathop{\mathrm{Re}}\bigl(q^{-1-2\delta}\Theta^2(H-E)\bigr)
-C_1\mathop{\mathrm{Re}}\bigl((\eta'_{m,n})^2\mathrm e^{2\theta}(H-E)\bigr)
.
\end{align*}
This implies the assertion. 
\end{proof}

\begin{proof}[Proof of Theorem~\ref{thm:priori-decay-b_0b}]   
Let $\phi\in \mathcal B_0^*\cap H^1_{0,\mathrm{loc}}(\Omega)$,
$E\in\mathbb R$ and  $\rho\ge 0$ be as in the statement of Theorem~\ref{thm:priori-decay-b_0b}.
Fix any $\delta'\in (0,\delta)$,
and choose $\beta_0\ge 1$, $n_0\in\mathbb N$ and $R\ge 1$ in agreement 
with Lemma~\ref{lem:14.10.4.1.17ffaabbqqq}.
We may assume that $2^{n_0-3}\ge \rho$,
so that for all $n>m\ge n_0$
\begin{align*}
\eta_{m-2,n+2}\phi\in\mathcal D(H).
\end{align*}
Let us evaluate the inequality \eqref{14.9.26.9.53ffaabbqqqq} 
in the state $\eta_{m-2,n+2}\phi\in\mathcal D(H)$.
Then it follows that for any $\alpha\ge \beta_0$ and $n>m\ge n_0$
\begin{align}
\|q^{-1/2-\delta'}\Theta\phi\|^2
\leq 
C_1\alpha^{-1}
\Bigl(2^{-2m}\|\eta_{m-1,m+1}\mathrm e^{\alpha q}\phi\|^2
+2^{-2n}\|\eta_{n-1,n+1}\mathrm e^{\alpha q}\phi\|^2
\Bigr).
\label{eq:11.10.23.12.54}
\end{align}
The second term in the parentheses on the right of \eqref{eq:11.10.23.12.54} 
vanishes under the limit $n\to\infty$, 
and hence by Lebesgue's monotone convergence
theorem we obtain
\begin{align*}
\|\bar\eta_m q^{-1/2-\delta'}\mathrm e^{\theta}\phi\|^2
&
\le 
C_12^{-2m}\alpha^{-1}\|\eta_{m-1,m+1}\mathrm e^{\alpha q}\phi\|^2,
\end{align*}
or 
\begin{align}
\|\bar\eta_m q^{-1/2-\delta'}[\mathop{\mathrm{exp}}(\theta-2^{m+2}\alpha )]\phi\|^2
&
\le 
C_2(m)\|\eta_{m-1,m+1}\phi\|^2
.
\label{eq:11.7.16.3.43qq}
\end{align}
Now assume $\bar\eta_{m+2}\phi\not\equiv 0$.
Then the left-hand side of \eqref{eq:11.7.16.3.43qq}
 grows exponentially as $\alpha\to\infty$
whereas the right-hand side remains bounded.
This is a contradiction.
Thus $\bar\eta_{m+2}\phi\equiv 0$,
and hence we are done.
\end{proof}

\section{Proof of LAP bounds}\label{sec:LAP}

In this section we prove Theorem~\ref{thmlapBnd}.
The proof depends on a propagation estimate of commutator type
similarly to Section~\ref{sec:Proof},
but with different weight functions.

We shall here use the weight functions $\Theta$ defined as 
\begin{align*}
\Theta=\Theta_{\nu,R}=1-\parb{1+r/2^\nu}^{-1};\quad r=r_R,\,\,\nu\in\mathbb N_0,
\end{align*} 
and consider the `propagation observable'
\begin{align}
P=P_{R,\nu, \epsilon}^f=\Theta^{1/2}f(H)\zeta (B)f(H)\Theta^{1/2};
\quad f\in C^\infty_{\mathrm c}(\mathbb R),\ \epsilon\in(0,1),
\label{eq:171118}
\end{align} 
where $\zeta=\zeta_{\epsilon}\in C^\infty(\mathbb R)$ is 
the smooth sign function from Section~\ref{subsubsec:Smooth sign function}.

As in Section~\ref{sec:Proof}  we denote the derivatives of $\Theta$ in $r$ by primes
 and compute 
\begin{align*}
\Theta'=2^{-\nu}(1+r/2^\nu)^{-2},\quad
\Theta''=-2^{1-2\nu}(1+r/2^\nu)^{-3},
\end{align*}
and in general
\begin{align*}
\Theta^{(k)}=(-1)^{k-1}k!2^{-k\nu}\parb{1+r/2^\nu}^{-1-k}
\ \ \text{for }k=1,2,\dots.
\end{align*}
>From the above expression it follows that 
for any $k,l\in\mathbb N_0$ with $k\ge l$
\begin{align}\label{eq:elebnd}
 \begin{split}
0< (-1)^{k-1}(k!)^{-1}r^k\Theta^{(k)}\le (-1)^{l-1}(l!)^{-1}r^{l}\Theta^{(l)}\le \min\{1,r/2^\nu\}. 
 \end{split}
\end{align} Let $\kappa=\delta/(1+2\delta)$ as in   \eqref{eq:171114}.

\begin{lemma}\label{lemma:12.7.2.7.9} 
Suppose 
{Condition~\ref{cond:smooth2wea3n1} or Condition \ref{cond:smooth2wea3n12}},
and let $E\in \mathbb R\setminus \mathcal T(H)$. 
There exist $c,C>0$, $R\ge 1$, $\epsilon\in (0,1)$,
 real-valued $f\in C^\infty_{\mathrm c}(\mathbb R)$
and a neighbourhood $I\subset \mathbb R$ of $E$ such that 
for all $\nu\in\mathbb N_0$
and $z\in I_\pm$
\begin{align}\label{eq:bndLAPa}
  2\mathop{\mathrm{Im}}\bigl(P(H-z)\bigr) \ge c \Theta'
  -Cr^{-1-2\kappa}\Theta -\mathop{\mathrm{Re}}\bigl(Q(H-z)\bigr);
\end{align} 
here 
  $Q=Q_\nu\in \mathcal L(\mathcal B)\cap \mathcal L(\mathcal B^*)$ is bounded 
uniformly in $\nu\in\mathbb N_0$, and the
  estimate \eqref{eq:bndLAPa} is understood as a quadratic form on
  $\mathcal H^2$. 
\end{lemma}
\begin{proof} 
The proof is very similar to that of Lemma~\ref{lem:14.10.4.1.17ffaabb},
and we skip some of the details here.
To bound commutators of functions of $H$, $r$ and $B$
we are going to repeatedly use (small variations of)  
Lemmas~\ref{lem:171113}, \ref{lem:171113b} and \ref{lem:17111416} 
without references again.

Fix the variables in the following order: 
For any $\lambda=E\in\mathbb R\setminus \mathcal T(H)$ and $\sigma\in(0,\gamma(\lambda))$
choose $R\ge 1$ and a neighborhood $\mathcal U\subset\mathbb R$ of $\lambda$
in accordance with Corollary~\ref{cor:171007}.
Let $f\in C^\infty_{\mathrm c}(\mathcal U)$ be a real-valued function 
such that $0\le f\le 1$ in $\mathcal U$ and 
$f=1$ in an open neighborhood $\tilde I\subset \mathbb R$ of $\lambda$.
Finally we fix any $\epsilon\in(0,{\sigma})$. 
With these variables we 
consider the observable $P$ defined as \eqref{eq:171118}. 
Note that Corollary~\ref{cor:171007} asserts that 
\begin{align}
\begin{split}
f(H)(\mathbf DB)f(H)
\ge 
f(H)r^{-1/2}\bigl(\sigma^2-B^2\bigr)r^{-1/2}f(H)
-C_1r^{-2}
,
\end{split}
\label{eq:171025b}
\end{align}
and this bound will be implemented in the same manner as in the proof of Lemma~\ref{lem:14.10.4.1.17ffaabb}.
In the present  proof all the estimates are uniform in $\nu\in\mathbb
N_0$ and  $z\in I_\pm$.

Let $I\subset \tilde I$ be a compact neighbourhood of $\lambda=E$.
We calculate for $z\in I_\pm$
\begin{align}
\begin{split}
2\mathop{\mathrm{Im}}\bigl(P(H-z)\bigr)
&
=2\mathop{\mathrm{Re}}\bigl((\mathbf D\Theta^{1/2})f(H)\zeta (B)f(H)\Theta^{1/2}\bigr)
\\&
\phantom{{}={}}{}
+\Theta^{1/2}f(H)\bigl(\mathbf D\zeta (B)\bigr)f(H)\Theta^{1/2}
-2(\mathop{\mathrm{Im}}z)P,
\end{split}
\label{eq:17110517}
\end{align}
and bound each term on the right-hand side of \eqref{eq:17110517}.
As for the first term, we substitute 
\begin{align*}
\bD \Theta^{1/2}
=\tfrac12\mathop{\mathrm{Re}}\bigl(\Theta'\Theta^{-1/2}\omega\cdot p\bigr)
=\tfrac12\Theta'\Theta^{-1/2}B
-\tfrac{\mathrm i}4\omega^2\bigl(\Theta'\Theta^{-1/2}\bigr)',
\end{align*}
and then by commuting operators and noting \eqref{eq:elebnd} we obtain
\begin{align}
\begin{split}
2\mathop{\mathrm{Re}}\bigl[(\mathbf D\Theta^{1/2})f(H)\zeta (B)f(H)\Theta^{1/2}\bigr]
&\ge 
f(H)\Theta'^{1/2}B\zeta(B)\Theta'^{1/2}f(H)
\\&\phantom{{}={}}{}
-C_2r^{-2}\Theta. 
\end{split}
\label{eq:171105173}
\end{align} 
This part corresponds to Step I\hspace{-.05em}I\hspace{-.05em}I
of the proof of Lemma~\ref{lem:14.10.4.1.17ffaabb},
however the arguments are  simpler since we do not need to  consider
the square root  of $B\zeta(B)$ as before.
Hence we omit the details verifying \eqref{eq:171105173}.
As for the second term of \eqref{eq:17110517},
we proceed as in Step I\hspace{-.05em}V of the proof of Lemma~\ref{lem:14.10.4.1.17ffaabb}.
We omit the details  again, but with slightly simpler arguments we can actually show that 
\begin{align}
\begin{split}
\Theta^{1/2}f(H)\bigl(\mathbf D\zeta (B)\bigr)f(H)\Theta^{1/2}
&\ge 
c_1 \Theta'{}^{1/2} f(H)\zeta'(B) f(H)\Theta'{}^{1/2}
-C_3r^{-1-2\kappa}\Theta.
\end{split}
\label{eq:171105174}
\end{align}
Here we use that $\epsilon\in(0,\sqrt{\sigma})$, \eqref{eq:171025b}  and \eqref{eq:elebnd}.
By using $\|P\|_{\mathcal H}\le 1$ we can bound the last term of
\eqref{eq:17110517} as
\begin{align}
-2(\mathop{\mathrm{Im}}z)P
\ge -2|\mathop{\mathrm{Im}}z|
=\pm 2\mathop{\mathrm{Im}}(H-z).
\label{eq:171105174b}
\end{align}
Now by \eqref{eq:17110517}, \eqref{eq:171105173}, \eqref{eq:171105174}, \eqref{eq:171105174b}
and Lemma~\ref{lem:170928} we obtain the lower bound
\begin{align}
\begin{split}
2\mathop{\mathrm{Im}}\bigl(P(H-z)\bigr)
&\ge 
c_2\Theta'{}^{1/2} f(H)^2\Theta'{}^{1/2}
-C_4r^{-1-2\kappa}\Theta
\pm 2\mathop{\mathrm{Im}}(H-z).
\end{split}
\label{eq:171119}
\end{align} 

Finally it suffices to remove $f(H)^2$ from the first term of
\eqref{eq:171119} with a controllable error.  This corresponds to the
middle part of Step I\hspace{-.05em}I of the proof of
Lemma~\ref{lem:14.10.4.1.17ffaabb}, and again the arguments are
simpler.  Introducing  $f_1\in\mathcal F^{-1}$ by
\eqref{eq:f1}  
 we obtain after commutation
\begin{align}
&\Theta'{}^{1/2}\bigl(1-f(H)^2\bigr)\Theta'{}^{1/2}
\le 
C_5\mathop{\mathrm{Re}}\bigl[\Theta'{}^{1/2}f_1(H)\Theta'{}^{1/2}(H-z)\bigr]
+C_6r^{-2}\Theta
.
\label{eq:171119b}
\end{align}
The lemma
follows from \eqref{eq:171119} and \eqref{eq:171119b}.
\end{proof}

\begin{proof}[Proof of Theorem~\ref{thmlapBnd}]
Letting  $E\in \mathbb R\setminus(\sigma_{\mathrm{pp}}(H)\cup\mathcal
T(H))$ we prove the assertion for a  compact neighbourhood $I\subset\mathbb R$
of $E$. This suffices due to compactness.  
We choose such $I$ and  other variables in agreement with Lemma~\ref{lemma:12.7.2.7.9}. 
We may assume $I\subset \mathbb R\setminus(\sigma_{\mathrm{pp}}(H)\cup\mathcal T(H))$. 
Then Lemma~\ref{lemma:12.7.2.7.9} and the Cauchy--Schwarz inequality imply  that   
uniformly in $\nu\in\mathbb N_0$ and $\phi=R(z)\psi$ with $z\in I_\pm$ and $\psi\in \mathcal B$
\begin{align}
\label{eq:13.8.22.4.59cc22}
\|\Theta'{}^{1/2}\phi\|_{\mathcal H}^2
&\le C_1\bigl(\|\phi\|_{\mathcal B^*}\|\psi\|_{\mathcal B}
+\|r^{-1/2-\kappa}\Theta^{1/2}\phi\|_{\mathcal H}^2\bigr).
\end{align}
Thanks to \eqref{eq:elebnd} we can bound  $\Theta\le \Theta^\kappa\le
r^\kappa/2^{\kappa\nu}$, and by implementing  these estimates  
in \eqref{eq:13.8.22.4.59cc22}  it follows that 
\begin{align}
\label{eq:13.8.22.4.59cc22b}
\|\Theta'{}^{1/2}\phi\|_{\mathcal H}^2
&\le C_2\bigl(\|\phi\|_{\mathcal B^*}\|\psi\|_{\mathcal B}
+2^{-\kappa \nu}\|\phi\|_{\mathcal B^*}^2\bigr).
\end{align}
On the other hand by taking supremum over $\nu\in\mathbb N_0$ 
 it also follows that 
\begin{align}
\|\phi\|^2_{\mathcal B^*}&\leq C_3\bigl(\|\psi\|_{\mathcal B}^2+\|r^{-1/2-\kappa}\phi\|_{\mathcal H}^2\bigr)
.
\label{eq:13.8.22.4.59cc24}
\end{align}

Now suppose by contradiction that there exist 
$\phi_n=R(z_n)\psi_n$ with $z_n\in I_\pm$ and $\psi_n\in\mathcal B$ 
such that 
\begin{align}
z_n\to E'\in I,\quad 
\|\psi_n\|_{\mathcal B}\to 0,\quad 
\|\phi_n\|_{\mathcal B^*}=1.
\label{eq:1711195}
\end{align}
By using a subsequence we can assume that there exists 
\begin{align}
\phi:=\wslim_{n\to \infty}\phi_n\ \ \text{in }L^2_{-(1+\kappa)/2}.
\label{eq:1711196}
\end{align}
By local compactness and an energy estimate we easily see
that for any $\chi\in C^\infty_{\mathrm c}(\mathbf X)$
\begin{align*}
\quad \lim_{n\to \infty}\chi\phi_n=\chi\phi 
\ \ \text{in } H^1_0(\Omega),
\end{align*} 
which implies that $\phi\in H^1_{0,\mathrm{loc}}(\Omega)$.
In addition, we have by \eqref{eq:13.8.22.4.59cc22b} that 
\begin{align*}
\|\Theta'{}^{1/2}\phi\|_{\mathcal H}^2
=\lim_{n\to \infty}\|\Theta'{}^{1/2}\phi_n\|_{\mathcal H}^2
\leq C_22^{-\kappa\nu},
\end{align*}
which implies $\phi\in \mathcal B_0^*$.
By taking limit in $(H-z_n)\phi_n=\psi_n$ we have 
$$(H-E')\phi=0\ \ \text{in the distributional sense}.$$ 
Since $E'\notin \sigma_{\mathrm{pp}}(H)\cup\mathcal T(H)$, we obtain from Theorem \ref{thm:priori-decay-b_0} that $\phi=0$.
By local compactness and \eqref{eq:1711196} it then follows that 
$\phi_n\to 0$ in $L^2_{-1/2-\kappa}$.
In turn this implies  by \eqref{eq:13.8.22.4.59cc24} that 
$\phi_n\to 0$ in $\mathcal B^*$,
contradicting the assumption \eqref{eq:1711195}.
Hence we obtain that for all  $z\in I_\pm $ and $\psi\in \mathcal B$
\begin{align}
\|R(z)\psi\|_{\mathcal B^*}\leq C_4\norm{\psi}_{\mathcal B}.
\label{eq:171216}
\end{align}

To bound the second term on the left-hand side of \eqref{eq:lap-besov-spacebnd}
it suffices to note that 
\begin{align}
p^*\Theta'p
\le C_5\mathop{\mathrm{Re}}\bigl(\Theta'(H-z)\bigr)+C_6\Theta'.
\label{eq:17121621}
\end{align}
By \eqref{eq:171216}, \eqref{eq:17121621} and the Cauchy--Schwarz inequality we obtain the assertion.
\end{proof}

\section{Proof of microlocal resolvent bounds and applications}
\label{sec:Holder continuity of boundary value of resolvent}

We prove Theorem~\ref{microLoc} in Section~\ref{subsec:1712161},
and Corollaries~\ref{cor:holdCon} and \ref{cor:Som} in Section~\ref{subsec:1712162}.

\subsection{Microlocal resolvent bounds}\label{subsec:1712161}

Here we prove Theorem~\ref{microLoc} partly  by using the scheme of  \cite{GIS}.
The proof consists of two  parts, according to spectral components of the radial momentum $B$.

We first consider the components of high radial momentum.  Using the 
functions  $\bar\chi_m\in C^\infty(\mathbb R)$ from
\eqref{eq:1711021}, $m\in \N$ large,  we first introduce $F_m\in\mathcal F^0$ by
\begin{align}
  F_{m}(b)=\bar\chi_m(|b|);\, b\in\mathbb R.
\label{eq:17112421}
\end{align}

\begin{lemma}\label{lem:17112122}
Suppose 
{Condition~\ref{cond:smooth2wea3n1} or Condition \ref{cond:smooth2wea3n12}}
and let $R\ge 1$ be sufficiently large.
Let $I\subset \mathbb R$ be a compact interval and $s\in\mathbb R$. 
Then for all  $m$ large  enough there exists $C>0$ such that for all $z\in I_\pm$,
as a quadratic form on $(H-\mathrm i)^{-1}L^2_s$,
\begin{align*}
F_m(B)r^{2s}F_m(B)
\le 
Cr^{2s-1}
+\mathop{\mathrm{Re}}\bigl(F_m(B)Q(H-z)\bigr)
,
\end{align*}
where $Q=Q_z\in \mathcal L(L^2_s,L^2_{-s})$ 
is bounded uniformly in $z\in I_\pm$.
\end{lemma}
\begin{proof}
Fix real $f\in C^\infty_{\mathrm c}(\mathbb R)$ such that 
$f=1$ in a neighbourhood of $I$ 
and decompose with $T=F_m(B)r^{2s}F_m(B)$
\begin{align}
\begin{split}
T
=
f(H)Tf(H)
+\mathop{\mathrm{Re}}\bigl((1+f(H))T(1-f(H))\bigr).
\end{split}
\label{eq:171121}
\end{align}
We bound each term on the right-hand side of \eqref{eq:171121}.
The first term is estimated as 
\begin{align*}
f(H)Tf(H)
&
\le 
C_12^{-2m}f(H)r^{s}F_m(B)B^2F_m(B)r^{s}f(H)
+C_2(m)r^{2s-1}
\\&
\le 
C_12^{-2m} r^{s}F_m(B)f(H)B^2f(H)F_m(B)r^{s}
+C_3(m)r^{2s-1}
\\&
\le 
C_42^{-2m} \mathop{\mathrm{Re}}\bigl(r^{s}F_m(B)f(H)(H-z)f(H)F_m(B)r^{s}\bigr)
\\&\phantom{{}={}}{}
+C_52^{-2m}r^{s}F_m(B)f(H)^2F_m(B)r^{s}
+C_3(m)r^{2s-1}
\\&
\le 
C_42^{-2m} \mathop{\mathrm{Re}}\bigl(F_m(B)r^{s}f(H)F_m(B)r^{s}f(H)(H-z)\bigr)
\\&\phantom{{}={}}{}
+C_52^{-2m}f(H)Tf(H)
+C_6(m)r^{2s-1}
.
\end{align*}
\begin{subequations}
  Therefore if we $m$ is large enough  we obtain
\begin{align}
f(H)Tf(H)
&\le 
C_7r^{2s-1}
+\mathop{\mathrm{Re}}\bigl(F_m(B)Q_1(H-z)\bigr)
\label{eq:171121b}
\end{align}
with
$$Q_1=C_82^{-2m}r^{s}f(H)F_m(B)r^{s}f(H).$$
As for the second term of \eqref{eq:171121} we can
estimate 
\begin{align}
\mathop{\mathrm{Re}}\bigl((1+f(H))T(1-f(H))\bigr)
\le 
C_9r^{2s-1}+
\mathop{\mathrm{Re}}\bigl(F_m(B)Q_2(H-z)\bigr),
\label{eq:171121bb}
\end{align}
where 
\begin{align*}
Q_2=(1+f(H))r^{2s}F_m(B)(1-f(H))R(z).
\end{align*}

The lemma  follows from  \eqref{eq:171121}, \eqref{eq:171121b}, 
 and \eqref{eq:171121bb}.
\end{subequations}
\end{proof}

\begin{corollary}\label{cor:1711241720}
Suppose 
{Condition~\ref{cond:smooth2wea3n1} or Condition \ref{cond:smooth2wea3n12}.
Let} $R\ge 1$ be sufficiently large.
Let $I\subset \mathbb R\setminus (\sigma_{\mathrm{pp}}(H)\cup\mathcal T(H))$ be a compact interval 
and $s\in [-1/2,0)$. Then for all  $m$ large  enough 
there exists $C>0$ such for all  $z\in I_\pm$
and $\psi\in L^2_{s}$
\begin{align*}
\|F_m(B)R(z)\psi\|_{L^2_{s}}
\le 
C\|\psi\|_{\mathcal B}
.
\end{align*}
\end{corollary}
\begin{proof}
The assertion is a consequence of Lemma~\ref{lem:17112122}, 
Theorem~\ref{thmlapBnd} and the Cauchy--Schwarz inequality. 
\end{proof}

We next study  middle components 
of the radial momentum for the outgoing and the incoming resolvents
using  a  modification of an
  `induction start' given in \cite{GIS}.
  We use  the function $\chi$ of \eqref{eq:14.1.7.23.24} and define 
 $\chi_\epsilon(t)=\chi(t/\epsilon)$ for
 $\epsilon\in (0,1)$.

\begin{lemma} \label{microLocb} Suppose 
{Condition~\ref{cond:smooth2wea3n1} or Condition \ref{cond:smooth2wea3n12}}.
Let $E\in\mathbb R\setminus \mathcal T(H)$ and  $0<\sigma'
<\sigma<\gamma(E)$.
 Take $R\ge 1$  large  and $\epsilon>0$ small. Consider 
 for real $f\in
 C_\c^\infty(\R)$ and  $\kappa'\in(0,\kappa)$
    \begin{align}\label{eq:171118v}
      \begin{split}
      P_\pm=\mp f(H)
      r^{\kappa'}F_\epsilon(\pm B)r^{\kappa'}f(H);\quad
F_\epsilon(b)=(\sigma'+3\epsilon-b)^{2\kappa'}\chi^2_\epsilon(b-\sigma').  
      \end{split}
\end{align}
 There exist $c,C,>0$,  real 
$f\in C^\infty_{\mathrm c}(\mathbb R)$ and a neighbourhood 
$I\subset\mathbb R$ of $E$ (possibly depending on
$\kappa'\in(0,\kappa)$) 
such that for all  $z\in I_\pm$, as quadratic forms  on $(H-\mathrm i)^{-1}L^2_{1/2+\kappa'}$,
\begin{align*}
2\mathop{\mathrm{Im}}\bigl(P_\pm(H-z)\bigr)
&\ge 
c\chi_\epsilon(\pm B-\sigma')r^{-1+2\kappa'}\chi_\epsilon(\pm B-\sigma')\\&
-Cr^{-1-2(\kappa-\kappa')}
-\mathop{\mathrm{Re}}\bigl(\chi_\epsilon(\pm B-\sigma')Q (H-z)\bigr),
\end{align*} 
where $Q=Q_z\in \mathcal L(L^2_{1/2+\kappa'},L^2_{3/2-\kappa'})$ 
is bounded uniformly in $z\in I_\pm$.
\end{lemma}
\begin{proof}  The proof is  similar to those of Lemmas~\ref{lem:14.10.4.1.17ffaabb} and \ref{lemma:12.7.2.7.9},
and we skip some of the details.
In particular we use Lemmas~\ref{lem:171113}, \ref{lem:171113b} and
\ref{lem:17111416} without references again. 
Fix  $E=\lambda\in\mathbb R\setminus \mathcal T(H)$ and $\sigma<\gamma(E)$ 
and 
choose then  $R\ge 1$ and a neighbourhood $\mathcal U\subset\mathbb R$ of $\lambda$
in accordance with Corollary~\ref{cor:171007}.
Let $f\in C^\infty_{\mathrm c}(\mathcal U)$ be a real-valued function 
such that $f=1$ in an open neighborhood $\tilde I\subset \mathbb R$ of $\lambda$.
Let $I\subset \tilde I$ be a compact neighbourhood of $\lambda$. Let
$\kappa'\in(0,\kappa)$ and $\sigma'\in(0,\sigma)$.
With these quantities fixed  we now  consider for small $\epsilon>0$
the  operator $P_+$ defined by  \eqref{eq:171118v} (treating   only the upper sign). 

We use the operator $\tilde H=g(H)=H\tilde f(H)$ from Step IV of the
proof of Lemma \ref{lem:14.10.4.1.17ffaabb} and let $\tilde \bD$
denote the corresponding Heisenberg derivative.  We compute with
$\theta_\epsilon=\sqrt{-F_\epsilon'}$ (and using the notation
\eqref{eq:1712022})
\begin{align*}
\begin{split}
  f(H)r^{\kappa'}\parb{\tilde \bD F_\epsilon(B)}&r^{\kappa'}f(H)\\
= -r^{\kappa'}\theta_\epsilon(B)&f(H)(\tilde \bD
  B)f(H)\theta_\epsilon(B)r^{\kappa'}+
   O(r^{2\kappa'-2\kappa-1}).
\end{split}
\end{align*} Note  that 
\begin{align*}
  -\tfrac12(\sigma'+3\epsilon-b)^{1-2\kappa'}F_\epsilon'(b)=\kappa'\chi^2_\epsilon(b-\sigma')-(\sigma'+3\epsilon-b)(\chi_\epsilon\chi'_\epsilon)(b-\sigma')\geq0.
\end{align*}
 We also compute
   \begin{align*}
   \tilde \bD
   {r}^{\kappa'}=\kappa 'g'(H)r^{\kappa'-1}B+O(r^{\kappa'-2}).
 \end{align*}

Introducing  $T=(\sigma'+3\epsilon-B)^{\kappa'-1/2}\chi_\epsilon(B-\sigma')r^{\kappa'-1/2}f(H)$ this leads to the lower bound
\begin{align}\label{eq:Tbnd}
  \begin{split} &2\mathop{\mathrm{Im}}\bigl(P_+(H-z)\bigr)\\
  &\geq \bD P_+\\&\geq 2\kappa'T^*\parb{\sigma^2-B^2-B(\sigma'+3\epsilon-B)}T+O(r^{2\kappa'-2\kappa-1})\\
&\geq 2\kappa'T^*\parb{\sigma^2-(\sigma'+2\epsilon)(\sigma'+3\epsilon)}T+O(r^{2\kappa'-2\kappa-1})\\
&= c_1T^*T+O(r^{2\kappa'-2\kappa-1});
\quad
c_1=2\kappa'\parb{\sigma^2-(\sigma'+2\epsilon)(\sigma'+3\epsilon)}>0. 
\end{split}
\end{align}

Introducing 
$S=\chi_\epsilon(B-\sigma')r^{\kappa'-1/2}$ we claim that for some $c_2,C_2>0$
\begin{align}\label{eq:lowST}
  T^*T\geq c_2S^*S-C_2 r^{2\kappa'-2}-\mathop{\mathrm{Re}}\bigl(\chi_\epsilon(B-\sigma')Q (H-z)\bigr),
\end{align}  where $Q=Q_z\in \mathcal
L(L^2_{1/2+\kappa'},L^2_{3/2-\kappa'})$ is bounded uniformly in $z\in
I_+$.  The combination of \eqref{eq:Tbnd} and
\eqref{eq:lowST} completes the proof of the lemma.

To show \eqref{eq:lowST} we first remove the  factor
$(\sigma'+3\epsilon-B)^{\kappa'-1/2}$ of  $T$. 
 We write $Sf(H)=Sf(H)\tilde{ f}(H)$ and note that
\begin{align*}
  [Sf(H),\tilde{f}(H)]=O(r^{\kappa'-3/2}).
\end{align*}
 Using the notation $\|\cdot\|_s=\|\cdot\|_{L^2_s}$ this leads to
 \begin{align*}
   \norm{Sf(H)\phi}&\leq \norm{\tilde f(H)Sf(H)\phi}+C_1\norm{\phi}_{\kappa'-3/2}\\
&\leq C'_1\norm{T\phi}+C_1\norm{\phi}_{\kappa'-3/2},\, \phi\in L^2_{\kappa'-1/2},
 \end{align*} and therefore 
\begin{align}\label{eq:lowSTb}
  T^*T\geq c_3 f(H)S^*Sf(H)-C_3 r^{2\kappa'-3}.
\end{align}

Next we remove the factors $f(H)$ of \eqref{eq:lowSTb} writing as in
\eqref{eq:171121} and using the notation  $ P_z=S^*S(1-f(H))R(z)$
\begin{align}\label{eq:bmang}
  \begin{split}
  &f(H)S^*Sf(H)\\ & =
  S^*S-\mathop{\mathrm{Re}}\bigl((1+f(H))S^*S(1-f(H))\bigr)
  \\ & = S^*S-\mathop{\mathrm{Re}}\bigl(
(1+f(H))P_z(H-z)\bigr)\\
&\geq
S^*S-{\mathrm{Re}}\bigl(\chi_\epsilon(B-\sigma')Q_z(H-z)\bigr)-C_4
r^{2\kappa'-2},
\end{split}
\end{align} where $Q=Q_z\in \mathcal
L(L^2_{1/2+\kappa'},L^2_{3/2-\kappa'})$ is bounded uniformly in $z\in
I_+$.  

 We obtain
\eqref{eq:lowST} from \eqref{eq:lowSTb} and \eqref{eq:bmang}.
\end{proof}

\begin{corollary}\label{cor:1711241722}
Suppose 
Condition~\ref{cond:smooth2wea3n1} or Condition \ref{cond:smooth2wea3n12}.
Let $E\in\mathbb R\setminus (\sigma_{\mathrm{pp}}(H)\cup \mathcal T(H))$, $\sigma'
\in (0, \gamma(E))$, 
and take $R\ge 1$ sufficiently large and $\epsilon>0$ sufficiently
small. There exist for all $\kappa'\in(0,\kappa)$ a constant $C>0$ and a neighbourhood $I\subset \mathbb R$ of $E$ 
such that for all $z\in I_\pm$
and $\psi\in L^2_{1/2+\kappa'}$
\begin{align*}
\|\chi_\epsilon(\pm B-\sigma')R(z)\psi\|_{L^2_{-1/2+\kappa'}}
\le 
C\|\psi\|_{L^2_{1/2+\kappa'}},
\end{align*}
respectively.
\end{corollary}
\begin{proof}
The assertion follows from 
Lemma~\ref{microLocb}, Theorem~\ref{thmlapBnd} and the Cauchy--Schwarz inequality. 
\end{proof}

\begin{proof}[Proof of Theorem~\ref{microLoc}] 
The assertion follows  by a covering argument using  Corollaries~\ref{cor:1711241720}  and \ref{cor:1711241722}.
\end{proof}

\subsection{Applications}\label{subsec:1712162}
Here we prove Corollaries~\ref{cor:holdCon} and \ref{cor:Som}.

\begin{proof}[Proof of Corollary~\ref{cor:holdCon}]
We discuss only the upper sign.
Let $I\subset\mathbb R\setminus(\sigma_{\mathrm{pp}}(H)\cup\mathcal T(H))$
be a compact interval, and let $s>1/2$ and $\beta\in (0, \min\{\kappa,s-1/2\})$.
Decompose for any $z,z'\in I_+$ and $m\in \N_0$
\begin{align}\label{eq:tu}
\begin{split}
R(z)-R(z') &= \chi_m R(z)\chi_m -\chi_m R(z')\chi_m
\\&\phantom{{}={}}{}
+\parb{R(z)-\chi_m R(z)\chi_m}-\parb{R(z')-\chi_m R(z')\chi_m}. 
\end{split} 
\end{align}
We estimate the last two terms of \eqref{eq:tu} as
follows. Take any $s'\in (1/2,s-\beta]$. Then by Theorem \ref{thmlapBnd} we have
uniformly in $z\in I_+$ and $m\in\mathbb N_0$ 
\begin{align}
\begin{split}
\|R(z)-\chi_mR(z)\chi_m\|_{\mathcal L(L^2_s,L^2_{-s})}
&
\le 
C_1 \|r^{-s}R(z)(1-\chi_m)r^{-s}\|_{\mathcal L(L^2)}
\\&\phantom{{}={}}{}
+C_1 \|r^{-s}(1-\chi_m)R(z)\chi_mr^{-s}\|_{\mathcal L(L^2)}
\\&
\le{C_2} 2^{-\beta m}.
\end{split}
\label{eq:17112514a}
\end{align}
The same holds true  uniformly in $z'\in I_+$ and $m\in\mathbb N_0$ 
\begin{align}
\|R(z')-\chi_mR(z')\chi_m\|_{\mathcal L(L^2_s,L^2_{-s})}
\le C_2 2^{-\beta m}.
\label{eq:17112514b}
\end{align}
As for the first and second terms of 
\eqref{eq:tu} we write 
\begin{align}
\begin{split}
&\chi_mR(z)\chi_m-\chi_mR(z')\chi_m
\\&=\chi_mR(z)\parb{\chi_{m+1}(H-z')-(H-z)\chi_{m+1}}R(z')\chi_m
\\&
=\chi_mR(z)\parb{(z-z')\chi_{m+1}
-[H,\chi_{m+1}]}R(z')\chi_m.
\end{split}
\label{eq:171207}
\end{align}
Now let us choose $F_\pm\in \mathcal F^0$ such that 
\begin{align*}
F_-+F_+=1,\quad 
\mathop{\mathrm{supp}} F_-\subset (-\infty,\gamma_-(I))
,\quad 
\mathop{\mathrm{supp}} F_+\subset (-\gamma_-(I),\infty)
.
\end{align*}
We write the first term in the parentheses on the right-hand side of \eqref{eq:171207} as 
\begin{align*}
(z-z')\chi_{m+1}
=
(z-z')
\bigl(\chi_{m+1}F_-(B)+F_+(B)\chi_{m+1}+[\chi_{m+1},F_+(B)]\bigr)
\end{align*}
and the second term as 
\begin{align*}
\begin{split}
\mathrm{i}[H,\chi_{m+1}]
&
=
\bigl(\mathop{\mathrm{Re}}(\chi_{m+1}' B)\bigr)F_-(B)
+F_+(B)\mathop{\mathrm{Re}}(\chi_{m+1}' B)
\\&\phantom{{}={}}{}
+\bigl[\mathop{\mathrm{Re}}(\chi_{m+1}' B),F_+(B)\bigr]
.
\end{split}
\end{align*}
Then by \eqref{eq:171207} and Theorem~\ref{microLoc} 
it follows that uniformly in $z,z'\in I_+$ and $m\in\mathbb N_0$ 
\begin{align}
\begin{split}
\|\chi_mR(z)\chi_m-\chi_mR(z')\chi_m\|_{\mathcal L(L^2_s,L^2_{-s})}
&
\le 
|z-z'|\bigl\|r^{-s} R(z)\chi_{m+1} R(z')r^{-s}\bigr\|_{\mathcal L(\mathcal H)}
\\&\phantom{{}={}}{}
+\bigl\|r^{-s}R(z)[H,\chi_{m+1}]R(z')r^{-s}\bigr\|_{\mathcal L(\mathcal H)}
\\&
\le C_32^{(1-\beta)m}|z-z'|+C_32^{-\beta m}.
\end{split}
\label{eq:14.12.30.20.43}
\end{align} 

Summing up \eqref{eq:17112514a}, 
\eqref{eq:17112514b} and \eqref{eq:14.12.30.20.43},
we obtain uniformly in $z,z'\in I_+$ and $m\in\mathbb N_0$ 
\begin{align*}
\|R(z)-R(z')\|_{\mathcal L(L^2_s,L^2_{-s})}
\le C_3 2^{(1-\beta)m} |z-z'|+C_4 2^{-\beta m}.
\end{align*}
 For $|z-z'|\leq 1$ 
we choose $m\in\mathbb N_0$ such that $2^m\leq |z-z'|^{-1}<2^{m+1}$,
yielding 
\begin{align*}
\|R(z)-R(z')\|_{\mathcal L(L^2_s,L^2_{-s})}
\le C_5|z-z'|^\beta.
\end{align*} 
This bound is trivial for $|z-z'|\geq 1$. Therefore we obtain \eqref{eq:171125}
for $k=0$.
For $k=1$ we may use the bound with $k=0$ and the first reolvent equation.
The rest of the assertions follow from Theorem \ref{thmlapBnd} and \eqref{eq:171125}.
\end{proof}

\begin{proof}[Proof of Corollary~\ref{cor:Som}]
We discuss only the upper sign.
Let $\lambda\in\mathbb R\setminus(\sigma_{\mathrm{pp}}(H)\cup\mathcal T(H))$, 
and take $R\ge 1$ and $\tilde \gamma>0$
sufficiently large as in Theorem~\ref{microLoc} with $I=\{\lambda\}$.
We let $\psi\in r^{-\beta}\mathcal B$ with
$\beta\in [0,\kappa)$,
and set $\phi=R(\lambda+\mathrm i0)\psi$.
Then \ref{item:13.7.29.0.28} and \ref{item:13.7.29.0.29} follows by
Theorems \ref{thmlapBnd}, \ref{microLoc} and Corollary~\ref{cor:holdCon}.

Conversely, 
let $\phi'\in L^2_{-\infty}\cap H^1_{0,\mathrm{loc}}(\Omega)$ satisfy
\ref{item:13.7.29.0.28b} and \ref{item:13.7.29.0.29b}.
Set 
\begin{align*}
\phi''=\phi'-\phi;\quad \phi=R(\lambda+\mathrm i0)\psi.
\end{align*}
Since we proved that $\phi$ 
satisfies \ref{item:13.7.29.0.28b} and \ref{item:13.7.29.0.29b} it
follows that 
$\phi''$ 
satisfies \ref{item:13.7.29.0.28b} and \ref{item:13.7.29.0.29b} 
with $\psi=0$. 
Due to Theorem~\ref{thm:priori-decay-b_0} it suffices to show that 
$\phi''\in \mathcal B^*_0$. 

We first claim that $\phi''\in L^2_{-1}$.
Choose  $s<-1$ such that  $\phi''\in L^2_s$ and 
choose $F_\pm\in \mathcal F^0$ such that 
\begin{align*}
F_-+F_+=1,\quad 
\mathop{\mathrm{supp}} F_-\subset (-\infty,\gamma')
,\quad 
\mathop{\mathrm{supp}} F_+\subset (\gamma'/2,\infty)
.
\end{align*}
with $\gamma'>0$ sufficiently small.  
For any  $t\le 0$ we estimate uniformly in $m\in\mathbb N_0$
\begin{align*}
2\mathop{\mathrm{Im}}\bigl(\chi_mr^t(H-\lambda)\bigr)
&
=
-\bigl|(\chi_mr^t)'\bigr|^{1/2} B\bigl(F_-(B)+F_+(B)\bigr)\bigl|(\chi_mr^t)'\bigr|^{1/2}
\\&
\le 
\bigl|(\chi_mr^t)'\bigr|^{1/2} (c_1-B)F_-(B)\bigl|(\chi_mr^t)'\bigr|^{1/2}
-c_1\bigl|(\chi_mr^t)'\bigr|
,
\end{align*}
so that 
\begin{align}
\begin{split}
\bigl|(\chi_mr^t)'\bigr|
&
\le 
\bigl|(\chi_mr^t)'\bigr|^{1/2} (1-C_2B)F_-(B)\bigl|(\chi_mr^t)'\bigr|^{1/2}
\\&\phantom{{}={}}{}
-C_3\mathop{\mathrm{Im}}\bigl(\chi_mr^t(H-\lambda)\bigr)
.
\end{split}
\label{eq:171217}
\end{align}
We take $t=2s+2(<0)$ and apply \eqref{eq:171217}  to 
$\phi''=f(H)\phi''$; here $f\in C^\infty_{\mathrm c}(\mathbb R)$
satisfies $f(\lambda)=1$.  Then we obtain
\begin{align*}
\bigl\|\bigl|(\chi_mr^{2s+2})'\bigr|^{1/2}\phi''\bigr\|_{\mathcal H}^2
&\le C_4\bigl\|\bigl|(\chi_mr^{2s+2})'\bigr|^{1/2}\phi''\bigr\|_{\mathcal H}
\bigl\|\bigl|(\chi_mr^{2s+2})'\bigr|^{1/2}F_-(B)\phi''\bigr\|_{\mathcal H}
\\&\phantom{{}={}}{}
+C_4\|r^s\phi''\|_{\mathcal H}^2
,
\end{align*}
so that by the Cauchy--Schwarz inequality
\begin{align*}
\|\chi_mr^{s+1/2}\phi''\|_{\mathcal H}^2
&\le C_5
\|F_-(B)\phi''\|_{\mathcal B^*}^2
+C_5\|\phi''\|_s
.
\end{align*}
This implies that $\phi''\in L^2_{s+1/2}$, and hence inductively that
indeed $\phi''\in L^2_{-1}$. 

Finally we prove $\phi''\in\mathcal B^*_0$.
We take $t=0$ and apply again \eqref{eq:171217}  to 
$\phi''=f(H)\phi''$.
Then we obtain  
\begin{align*}
\bigl\||\chi_m'|^{1/2}\phi''\bigr\|_{\mathcal H}^2
&\le C_6\bigl\||\chi_m'|^{1/2}\phi''\bigr\|_{\mathcal H}
\bigl\||\chi_m'|^{1/2}F_-(B)\phi''\bigr\|_{\mathcal H}
\\&\phantom{{}={}}{}
+C_6\bigl\||\chi_m'|^{1/2}r^{-1/2}\phi''\bigr\|_{\mathcal H}\|r^{-1}\phi''\|_{\mathcal H}
,
\end{align*}
which implies that 
\begin{align*}
\lim_{m\to \infty}\inp{|\chi'_m|}_{\phi''}=0
,
\end{align*} 
or equivalently that $\phi''\in \vB^*_0$. By Theorem~\ref{thm:priori-decay-b_0} it follows that 
$\phi''=0$, 
and we are done. 
\end{proof}


\begin{thebibliography}{BGM1}
 


\bibitem[AH]{AH} S. Agmon, L. H\"ormander,
\textit{ Asymptotic properties of solutions of differential equations
 with simple characteristics}, J. d'Analyse Math. {\bf 30} (1976),
 1{--}38.


\bibitem[ABG1]{ABG1} W. Amrein, A. Boutet de Monvel-Bertier, V. Georgescu,
 \textit{On Mourre's approach to spectral theory}, Helv. Phys. Acta
 {\bf 62} no. 1 (1989), 1--20.

\bibitem[ABG2]{ABG2} W. Amrein, A. Boutet de Monvel-Bertier, V. Georgescu, \emph{
 $C_0$-groups, commutator methods and spectral theory of $N$-body Hamiltonians},
Basel--Boston--Berlin, Birkh\"auser, 1996.

\bibitem[BGM1]{BGM1} A. Boutet de Monvel, V. Georgescu and M. Mantoiu, 
\emph{Mourre Theory in a Besov
Space Setting}, C.R. Acad. Sci. Paris, Ser. I \textbf{310} (1990),
233--237.

\bibitem[BGM2]{BGM2} A. Boutet de Monvel, V. Georgescu and M. Mantoiu, 
\emph{Locally Smooth Operators and the Limiting Absorption Principle for N-Body
Hamiltonians}, Reviews in Math. Physics \textbf{5} no. 1 (1993), 105--189.
 
 
\bibitem[BGS]{BGS} A. Boutet de Monvel-Bertier, V. Georgescu, A. Soffer,
 \emph{$N$-body Hamiltonians with hard-core interactions},
 Rev. Math. Phys. \textbf{6} no. 4 (1994), 515--596.
 
 


\bibitem[BMP]{BMP} A. Boutet de Monvel-Bertier, D. Manda, R. Purice, \emph{The commutator method for form-relatively
compact perturbations}, Lett. Math. Phys. \textbf{22} (1991), 211--223.



\bibitem[De]{De} J. Derezi\'nski, \emph{Asymptotic completeness for
 $N$-particle long-range quantum systems}, Ann. of Math.
 \textbf{38} (1993), 427--476.


\bibitem[DG]{DG}
J. Derezi{\'n}ski and C. G{\'e}rard, \emph{Scattering theory of
 classical and quantum {$N$}-particle systems}, Texts and Monographs in
 Physics, Berlin, Springer 1997.


 \bibitem[FH]{FH} R. Froese and I. Herbst, \emph{Exponential bounds
and absence of positive eigenvalues for $N$-body Schr{\"o}dinger
operators}, Comm. Math. Phys. {\bf 87} no. 3 (1982/83), 429--447.

\bibitem[Geo]{Geo} V. Georgescu, \emph{On the unique continuation
    property for Schr{\"o}dinger Hamiltonians}, Helvetica Physica Acta
  {\bf 52} (1979), 655--670.


 
\bibitem[G\'er]{Ge} C. G{\'e}rard, \emph{A proof of the abstract limiting absorption principle
by energy estimates}, J. Funct. Anal. {\bf 254} (2008), 2707--2724.



\bibitem[GIS]{GIS} C. G{\'e}rard, H. Isozaki and E. Skibsted,
 \emph{$N$-body resolvent estimates},
 J. Math. Soc. Japan \textbf{48} no. 1 (1996), 135--160.

\bibitem[Gr]{Gr} G.M. Graf, \emph{Asymptotic completeness for
 $N$-body short-range quantum systems: a new proof},
 Commun. Math. Phys. \textbf{132} (1990), 73--101.

\bibitem[GJ]{GJ} S. Gol{\'e}nia, T. Jecko, \emph{A new look at
 Mourre's commutator theory}, Complex Anal. Oper. Theory \textbf{1} no. 3
(2007), 399--422

\bibitem[HeS]{HeS} I. Herbst, E. Skibsted, \emph{Absence of quantum states corresponding to unstable classical channels}, Ann. Henri Poincar\'e {\bf 9} (2008), 509--552.


\bibitem[H{\"o}]{Ho} L. H{\"o}rmander, \emph{The analysis of linear
 partial differential operators. {I\hspace{-.05em}V}}, Berlin, Springer
 1983--85.

\bibitem[HuS]{HuS}
W. Hunziker, I. M. Sigal, 
\emph{The quantum N-body problem},
J. Math.\ Phys.\ {\bf 41} (2000) no.\ 6, 3448--3510.
 
 

\bibitem[Is]{Is}
 H. Isozaki, \emph{A generalization of the radiation condition of
 Sommerfeld for $N$-body Schr\"odinger operators}, Duke
 Math. J. \textbf{74} no. 2 (1994), 557--584.




\bibitem[IS1]{IS1} K. Ito, E. Skibsted, \emph{Absence of positive
 eigenvalues for hard-core $N$-body systems}, Ann. Inst. Henri Poincar{\'e} {\bf 15} 
 (2014), 2379--2408.

\bibitem[IS2]{IS2} K. Ito, E. Skibsted, \emph{Rellich's theorem and
 $N$-body Schr\"odinger operators}, Reviews Math. Phys. \textbf
 {28} no. 5 (2016), 12 pp.

\bibitem[IS3]{IS3} K. Ito, E. Skibsted, \emph{Spectral 
 theory on manifolds}, preprint, submitted.






\bibitem[JP]{JP} A. Jensen, P. Perry, \emph{Commutator methods and
 Besov space estimates for {S}chr\"odinger operators}, J. Operator
 Theory \textbf{14} 
 (1985), 181--188.

\bibitem[La1]{La1} R. Lavine, \emph{Absolute continuity of Hamiltonian
operators with repulsive potential}, Proc. Amer. Math. Sot. {\bf 22} (1969), 55--60.


 

\bibitem[La2]{La2} R. Lavine, \emph{Absolute continuity of positive
 spectrum for Schr\"odinger operators with long-range potentials},
 J. Funct. Anal. {\bf 12} (1973), 30--54.





\bibitem[Mo1]{Mo1}
{\'E}. Mourre, \emph{Absence of singular continuous spectrum for certain
 selfadjoint operators}, Commun. Math. Phys. \textbf{78} no.~3 (1980/81),
 391--408.

\bibitem[Mo2]{Mo2}
{\'E}. Mourre, \emph{Operateurs conjugu{\'e}s et propri{\'e}t{\'e}s de propagation}, Commun. Math. Phys. \textbf{91} (1983),
 297--300.




\bibitem[Pe]{Pe} P. Perry, \emph{Exponential bounds and semifiniteness
 of point spectrum for $N$-body {S}chr\"odinger operators},
 Commun. Math. Phys. \textbf{92} (1984), 481--483.





\bibitem[RS]{RS} M.~Reed, B.~Simon, \emph{Methods of modern
 mathematical physics {I}-{I\hspace{-.05em}V}}, New York, Academic Press 1972-78.




\bibitem[Sk]{Sk} E. Skibsted, \emph{Propagation estimates for $N$-body
    Schr\"odinger operators}, Commun. Math. Phys. \textbf{142} (1991), 67--98.

 \bibitem[Ta]{Ta} H. Tamura, \emph{Principle of limiting absorption
 for $N$-body {S}chr\"odinger operators
- a remark
on the commutator method}, Letters Math. Phys. {\bf 17} (1989),
31--36.

\bibitem[Va]{Va} A. Vasy, \emph{Propagation of singularities in
 three-body scattering}, Ast{\'e}rique, \textbf{262} (2000).

 

\bibitem[Ya1]{Ya} D. Yafaev, \emph{Eigenfunctions of the continuous spectrum for the N-particle Schr\"odinger operator}, Spectral and scattering theory (Sanda, 1992), 259--286,
Lecture Notes in Pure and Appl. Math. \textbf{161}, Dekker, New York, 1994. 


\bibitem[Ya2]{Ya2} D. Yafaev, \emph{Radiation conditions and scattering
 theory for $N$-particle Hamiltonians},
 Commun. Math. Phys. \textbf{154} no. 3 (1993), 523--554.


\bibitem[Yo]{Yo}
K. Yosida, 
\emph{Functional analysis}. 
Reprint of the sixth (1980) edition. Classics in Mathematics. Springer-Verlag, Berlin, 1995. xii+501 
pp.\ ISBN: 3-540-58654-7.
\end{thebibliography}
\end{document}